\documentclass[reqno]{amsart}
\usepackage{fullpage,amsfonts,amssymb,amsthm,mathrsfs,graphicx,color,framed}
\usepackage{cancel,bm}

\usepackage[colorlinks=true, pdfstartview=FitV, linkcolor=blue, citecolor=blue, urlcolor=blue]{hyperref}
\definecolor{shadecolor}{rgb}{0.95, 0.95, 0.86}

\renewcommand{\d}{{\mathrm d}}

\newcommand{\im}{\mathrm{i}}
\newcommand{\e}{\mathrm{e}}
\newcommand{\z}{\zeta}

\renewcommand{\b}{\beta}

\def\res{\mathop{\mathrm{res}}\limits}

\numberwithin{equation}{section}

\newtheorem{theo}{Theorem}[section]

\newtheorem{lem}[theo]{Lemma}
\newtheorem{rem}[theo]{Remark}
\newtheorem{problem}[theo]{Riemann-Hilbert Problem}

\newtheorem{prop}[theo]{Proposition} 
\newtheorem{cor}[theo]{Corollary}

\begin{document}

\title[Extended Jimbo-Miwa-Ueno differential]{On the analysis of incomplete spectra in random matrix theory through an extension of the Jimbo-Miwa-Ueno differential}

\author{Thomas Bothner}
\address{Department of Mathematics, University of Michigan, 2074 East Hall, 530 Church Street, Ann Arbor, MI 48109-1043, United States}
\email{bothner@umich.edu}
\author{Alexander Its}
\address{Department of Mathematical Sciences, Indiana University-Purdue University Indianapolis, 402 N. Blackford St., Indianapolis, IN 46202, U.S.A.}
\email{aits@iupui.edu}
\author{Andrei Prokhorov}
\address{Department of Mathematical Sciences, Indiana University-Purdue University Indianapolis, 402 N. Blackford St., Indianapolis, IN 46202, U.S.A.}
\email{aprokhor@iupui.edu}



\thanks{T.B. acknowledges support of the AMS and the Simons Foundation through a travel grant as well as the hospitality of Indiana University-Purdue University Indianapolis in August $2016$. A.I and A.P. acknowledge support by the  NSF Grant DMS-1361856, the NSF Grant DMS-1700261, the SPbGU grant \#11.38.215.2014 and the RSF grant  \#17-11-01126}

\begin{abstract}
Several distribution functions in the classical unitarily invariant matrix ensembles are prime examples of isomonodromic tau functions as introduced by Jimbo, Miwa and Ueno (JMU) in the early 1980s \cite{JMU}. Recent advances in the theory of tau functions  \cite{ILP}, based on earlier works of B. Malgrange and M. Bertola, have allowed to extend the original Jimbo-Miwa-Ueno differential form to a 1-form closed on the full space of
extended monodromy data of the underlying Lax pairs. This in turn has yielded a novel approach for the asymptotic evaluation of isomonodromic tau functions,  including the exact computation of all relevant constant factors. We use  this method to efficiently compute  the tail asymptotics of soft-edge, hard-edge and bulk scaled distribution and gap functions in the complex Wishart ensemble, provided each eigenvalue particle has been removed independently with probability $1-\gamma\in(0,1]$.
\end{abstract}

\date{\today}
\maketitle
\section{Introduction and statement of results}

This paper is concerned with the large gap asymptotics of the universal limiting distributions in random matrix theory. 
The issue which we will specifically address  is the evaluation of the constant factors appearing in these asymptotics, the so-called
``constant problem''. We will present a new method for the derivation of tail expansions which does not rely on Fredholm, or Toeplitz, or Hankel determinant formul\ae\,, which are  the usual tools in the analysis of distribution functions. Instead,  our approach is based on the interpretation of the distribution 
functions as tau functions of the theory of isomonodromic deformations of  certain systems of linear ODEs with rational coefficients. 
Specifically we shall evaluate, including the constant factors, the tail asymptotics of soft-edge, hard-edge and bulk scaled distribution and gap functions in the complex Wishart ensemble, provided each eigenvalue particle has been removed independently with probability $1-\gamma\in(0,1]$.
In what follows, we shall describe the content of our work and its principal results in detail.

\subsection{Complete Wishart ensemble}\label{sec:11}
The complex Wishart ensemble \cite{M,F} can be realized as a log-gas system of (eigenvalue) particles $0<\lambda_1<\lambda_2<\ldots<\lambda_n$ on the positive real axis with probability density function for the location of the $\lambda_j$'s given by
\begin{equation}\label{JME:1}
	f(\lambda_1,\ldots,\lambda_n)=\frac{1}{Z_n}\prod_{1\leq j<k\leq n}(\lambda_k-\lambda_j)^2\prod_{j=1}^n\lambda_j^{\alpha}\e^{-\lambda_j},\ \ \ \alpha>-1.
\end{equation}
The constant $Z_n$ serves as normalization and \eqref{JME:1} is also well-known under the name Laguerre Unitary Ensemble (LUE). The great benefit and applicability of a unitarily invariant ensemble (such as the LUE) stems from the fact that the point process \eqref{JME:1} is determinantal, i.e. the underlying rescaled marginal densities (a.k.a. $k$-point correlation functions) can be computed in closed determinantal form, cf. \cite{M,F},
\begin{equation}\label{JME:2}
	R_k(\lambda_1,\ldots,\lambda_k)=\frac{n!}{(n-k)!}\int_0^{\infty}\cdots\int_0^{\infty}f(\lambda_1,\ldots,\lambda_n)\prod_{j=k+1}^n\d\lambda_j=\det\big[K_n(\lambda_j,\lambda_{\ell})\big]_{j,\ell=1}^k,\ \ \ k=1,\ldots,n.
\end{equation}
Here, the kernel function $K_n(x,y)$ is a Christoffel-Darboux kernel expressed in terms of classical Laguerre polynomials. It is well-known that \eqref{JME:2} encodes the core integrable structure of the LUE and at the same time paths the way to a rigorous analysis of the thermodynamical limit $n\rightarrow\infty$ of the $k$-point correlation function. Indeed, cf. \cite{F2,F}, the eigenvalue density obeys the Marchenko-Pastur law,
\begin{equation}\label{JME:3}
	\lim_{n\rightarrow\infty}R_1(n\lambda)=\rho_{_\textnormal{MP}}(\lambda)\equiv\frac{1}{2\pi}\sqrt{\frac{4-\lambda}{\lambda}}\chi_{(0,4)}(\lambda),\ \ \ \ \ \ \ \chi_A(x)=\begin{cases}1,&x\in A\\ 0,&x\notin A\end{cases},
\end{equation}
shown in Figure \ref{figure0} below. The global limiting law \eqref{JME:3} leads in turn to three qualitatively different local scenarios: provided we center and scale correctly, 
\begin{equation*}
	\begin{cases}
		\mu_j^{S}\,=2^{-\frac{4}{3}}n^{-\frac{1}{3}}(\lambda_j-4n)&\,\,\textnormal{soft-edge}\smallskip\\
		\mu_j^{H}=4^{-1}n^{-1}\lambda_j&\,\,\textnormal{hard-edge}\smallskip\\
		\mu_j^{B}=2\pi\rho_{_{\textnormal{MP}}}(cn)(\lambda_j-cn)&\,\,\textnormal{bulk}
	\end{cases},\ \ \ \ \ \ \ c\in(0,4)\ \ \textnormal{fixed};\ \ \ \ \ j=1,\ldots,n,
\end{equation*}
then, cf. \cite{F}, for any $t\in\mathbb{R}$,
\begin{equation}\label{JME:4}
	\lim_{n\rightarrow\infty}\mathbb{P}\big(\#\{\mu_j^S\in(t,+\infty)\}=0\big)=F_S(t)
\end{equation}
which is the limiting distribution function of the largest eigenvalue in the LUE, and for $t\in\mathbb{R}_{>0}$,
\begin{equation}\label{JME:5}
	\lim_{n\rightarrow\infty}\Big(1-\mathbb{P}\big(\#\{\mu_j^H\in(0,t)\}=0\big)\Big)=1-F_H(t,\alpha);\ \ \textnormal{resp.}\ \ \lim_{n\rightarrow\infty}\mathbb{P}\left(\#\left\{\mu_j^B\in\left(-\frac{t}{\pi},\frac{t}{\pi}\right)\right\}=0\right)=F_B(t)
\end{equation}
which is the limiting distribution function of the smallest eigenvalue, resp. the limiting bulk gap function in the LUE.
\begin{center}
\begin{figure}[tbh]
\resizebox{0.8\textwidth}{!}{\includegraphics{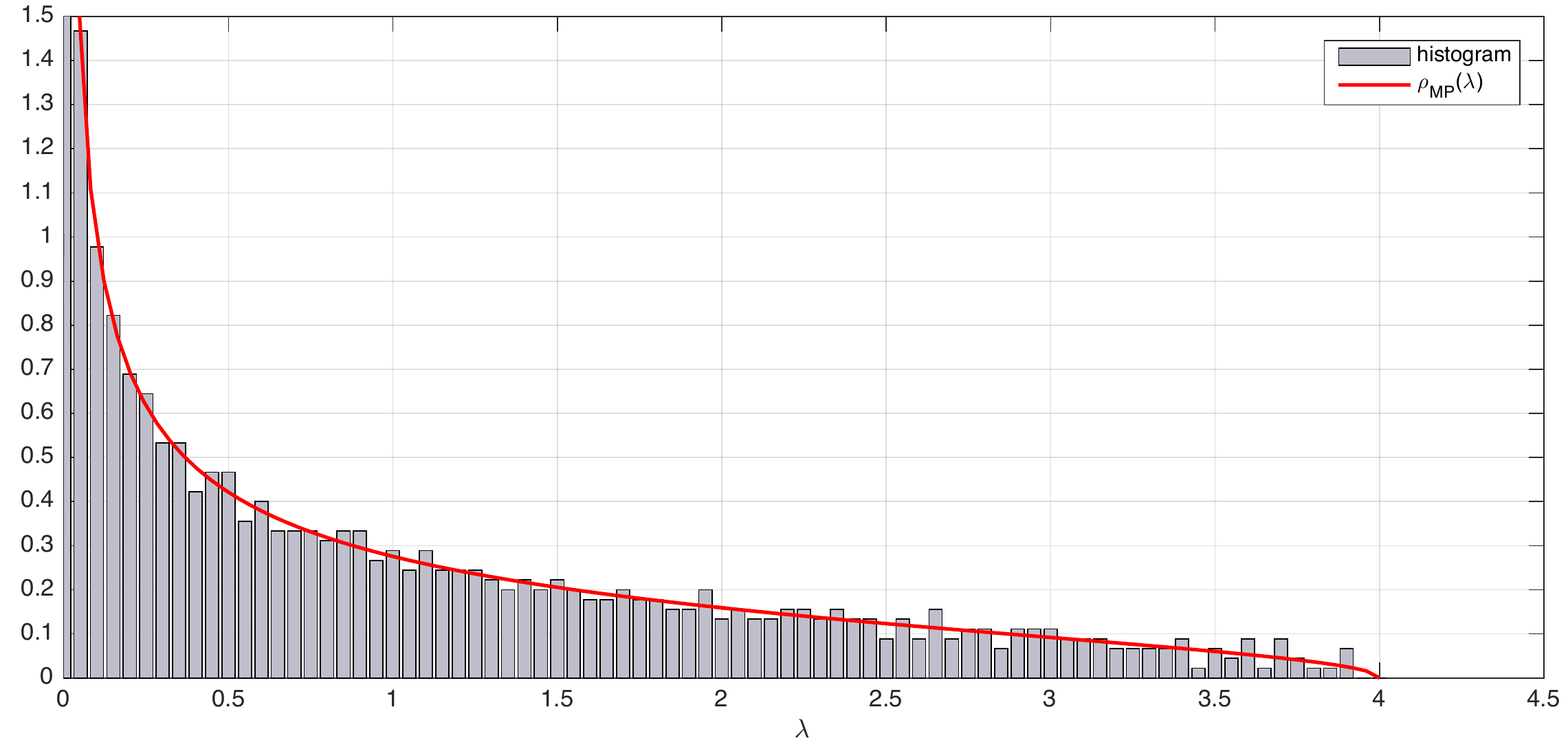}}
\caption{Histogram of the eigenvalues of a $900\times 900$ (rescaled) complex Wishart matrix in comparison with the Marchenko-Pastur density \eqref{JME:3}.}
\label{figure0}
\end{figure}
\end{center} 
The intimate connection of the three functions $F_B(t),F_S(t)$ and $F_H(t,\alpha)$ defined in \eqref{JME:4} and \eqref{JME:5} to the theory of integrable systems is remarkable and well-known: first, for the bulk function, as proven by Jimbo-Miwa-Mori-Sato \cite{JMMS}, 
\begin{equation}\label{JME:6}
	\ln F_B(t)=\int_0^t\mathcal{H}_B\big(q(s),p(s),s)\,\d s,\ \ \ \ \ t\in\mathbb{R}_{\geq 0}
\end{equation}
in terms of the Hamiltonian dynamical system 
\begin{equation}\label{HH:1}
	\frac{\d p}{\d t}=-\frac{\partial\mathcal{H}_B}{\partial q},\ \ \ \frac{\d q}{\d t}=\frac{\partial\mathcal{H}_B}{\partial p};\ \ \ \ \ \ \ \ \ \ \mathcal{H}_B(q,p,t)=-4\im q+\frac{4}{t}q^2\sinh^2\left(\frac{p}{2}\right).
\end{equation}
The required solutions $(q,p)$ to this system are smooth on the positive real axis and uniquely determined by the boundary behavior $q(t)\sim\frac{1}{2\pi\im}$ and $p(t)\sim4\im t$ as $t\downarrow 0$. The dynamical system (\ref{HH:1}) is equivalent to a special case of the Painlev\'e-V equation for
the function $\omega(t)=\exp(p(\frac{t}{2}))$, cf. \cite{JMMS},
\begin{equation}\label{PV.1}
\frac{\d^2\omega}{\d t^2} = \left(\frac{\d\omega}{\d t}\right)^2\frac{3\omega-1}{2\omega(\omega-1)}
+\frac{2\omega(\omega+1)}{\omega-1} +\frac{2i\omega}{t} - \frac{1}{t}\frac{\d\omega}{dt}
\end{equation} 
Second, for the distribution function of the largest eigenvalue, as proven by Tracy-Widom \cite{TW1}, 
\begin{equation}\label{JME:7}
	\ln F_S(t)=-\int_t^{\infty}\mathcal{H}_S\big(q(s),p(s),s\big)\,\d s,\ \ \ \ \ t\in\mathbb{R}
\end{equation}
in terms of the Hamiltonian dynamical system
\begin{equation}\label{JME:8}
	\frac{\d p}{\d t}=-\frac{\partial\mathcal{H}_S}{\partial q},\ \ \ \frac{\d q}{\d t}=\frac{\partial\mathcal{H}_S}{\partial p};\ \ \ \ \ \ \ \ \ \ \mathcal{H}_S(q,p,t)=\frac{1}{4}p^2-tq^2-q^4.
\end{equation}
Here the solutions $(q,p)$ are smooth on the real line and fixed in such a way that $q(t)\sim\textnormal{Ai}(t)$ and $p(t)=2q_t(t)$ as $t\rightarrow+\infty$, where $\textnormal{Ai}(z)$ is the Airy function, cf. \cite{NIST}. The system (\ref{JME:8}) is equivalent to a special case of the
Painlev\'e II  equation for the function $q(t)$,
\begin{equation}\label{PII0}
\frac{\d^2q}{\d t^2} = tq + 2 q^3,
\end{equation}
and the solution $q(t)$ selected by the condition $q(t)\sim\textnormal{Ai}(t)$ is known as the Hastings-McLeod solution to (\ref{PII0}), see \cite{HM}. Third, again by Tracy-Widom \cite{TW2},
\begin{equation}\label{JME:9}
	\ln F_H(t,\alpha)=\int_0^t\mathcal{H}_H\big(q(s,\alpha),p(s,\alpha),s,\alpha)\,\d s,\ \ \ \ \ \ \ \ t\in\mathbb{R}_{\geq 0},\ \ \alpha\in\mathbb{R}_{>-1}
\end{equation}
in terms of the Hamiltonian dynamical system
\begin{equation}\label{HH:2}
	\frac{\d p}{\d t}=-\frac{\partial\mathcal{H}_H}{\partial q},\ \ \ \frac{\d q}{\d t}=\frac{\partial\mathcal{H}_H}{\partial p};\ \ \ \ \ \ \ \ \ \ \mathcal{H}_H(q,p,t,\alpha)=\frac{q^2-1}{4t}p^2-\frac{\alpha^2 q^2}{4t(q^2-1)}-\frac{q^2}{4}.
\end{equation}
We enforce
\begin{equation*}
	q(t,\alpha)\sim\frac{t^{\frac{1}{2}\alpha}}{2^{\alpha}\Gamma(1+\alpha)},\ \ \ \ p(t,\alpha)=\frac{2tq_t(t,\alpha)}{q^2(t,\alpha)-1},\ \ \ t\downarrow 0,\ \ \ \ q_t=\frac{\d q}{\d t},
\end{equation*}
where $\Gamma(z)$ is Euler's Gamma function. In addition, $q^2(t,\alpha)$ is smooth and real-valued on the half ray $(0,+\infty)\subset\mathbb{R}$. Dynamical system (\ref{HH:2}) is equivalent to yet another special case of the Painlev\'e-V equation for the function 
$y(t)=(q(t^2)-1)/(q(t^2)+1)$,
\begin{equation}\label{PV.2}
\frac{\d^2y}{\d t^2} = \left(\frac{\d y}{\d t}\right)^2\frac{3y-1}{2y(y-1)}
-\frac{2y(y+1)}{y-1} - \frac{1}{t}\frac{\d y}{\d t} +\frac{\alpha^2}{8}\frac{(y-1)^2}{t^2}\left(y-\frac{1}{y}\right).
\end{equation} 
\begin{rem}The aforementioned smoothness  properties of $(q,p)$ in \eqref{HH:1} and \eqref{JME:8} are well-known, cf. \cite{DIZ,HM}. 
The smoothness of $q$ and $p$ in the case  \eqref{HH:2} is proven in Appendix \ref{AppC}.
\end{rem}
\begin{rem} Each Hamiltonian $\mathcal{H}$ listed above solves itself a $\sigma$-Painlev\'e equation in the variable $t$, see for instance \cite{F}, Chapter $8$.
\end{rem}

From \eqref{JME:6}, \eqref{JME:7} and \eqref{JME:9} we see that $F_B(t),F_S(t)$ and $F_H(t,\alpha)$ are generating functions of Hamiltonians
associated with specific Painlev\'e systems. As such they are directly related to the theory of isomonodromic tau-functions in the sense of Jimbo-Miwa-Ueno \cite{JMU}. We will discuss this connection in more detail in Section \ref{sec:3_0} below.

\subsection{Incomplete Wishart ensemble}\label{sec:13} We now return to the discussion of the complex Wishart ensemble and the collection of soft-edge, hard-edge and bulk scaled eigenvalues $\{\mu_j^S,\mu_j^H,\mu_j^B\}_{j=1}^n$. But instead of the complete setup \eqref{JME:1} we will be interested in the following thinned/incomplete Wishart ensemble (cf. \cite{BP1,BP2,BP3}): fix $\gamma\in[0,1]$ and discard each (either soft-edge, or hard-edge or bulk scaled) eigenvalue $\mu_j^r,r=S,H,B$ independently with probability $1-\gamma$. This operation reduces correlation in our initial setup and introduces a new particle system on the real line,
\begin{equation*}
	\mu_{1,\gamma}^r<\mu_{2,\gamma}^r<\ldots<\mu_{N,\gamma}^r\ \ \ \ \ \ r=S,H,B;\ \ \ \ \ \ N=N(n,\gamma)\leq n.
\end{equation*}
Quite naturally we are interested in the statistical properties of this new system, in particular what can be said about the thinned extremal distributions
\begin{equation}\label{JME:15}
	\lim_{n\rightarrow\infty}\mathbb{P}\big(\#\{\mu_{j,\gamma}^S\in(t,+\infty)\}=0)=F_S(t;\gamma);\ \ \ \ \lim_{n\rightarrow\infty}\Big(1-\mathbb{P}\big(\#\{\mu_{j,\gamma}^H\in(0,t)\}=0)\Big)=1-F_H(t,\alpha;\gamma),
\end{equation}
and the thinned bulk gap function
\begin{equation}\label{JME:16}
	\lim_{n\rightarrow\infty}\mathbb{P}\left(\#\left\{\mu_{j,\gamma}^B\in\left(-\frac{t}{\pi},\frac{t}{\pi}\right)\right\}=0\right)=F_B(t;\gamma),
\end{equation}
which directly generalize \eqref{JME:4} and \eqref{JME:5}? As it turns out the thinning mechanism preserves Hamiltonian structure as summarized in our first result below.
\begin{theo}\label{theo:1} Fix $\gamma\in[0,1]$ and let $F_B(t;\gamma),F_S(t;\gamma)$ and $F_H(t,\alpha;\gamma)$
denote the functions defined in \eqref{JME:15} and \eqref{JME:16}. Then
\begin{equation*}
	\ln F_B(t;\gamma)=\int_0^t\mathcal{H}_B\big(q(s;\gamma),p(s;\gamma),s\big)\,\d s,\ \ \ \ \ \ \ \ \ln F_S(t;\gamma)=-\int_t^{\infty}\mathcal{H}_S\big(q(s;\gamma),p(s;\gamma),s\big)\,\d s,
\end{equation*}
and
\begin{equation*}
	\ln F_H(t,\alpha;\gamma)=\int_0^t\mathcal{H}_H\big(q(s,\alpha;\gamma),p(s,\alpha;\gamma),s,\alpha\big)\,\d s,
\end{equation*}
using the Hamiltonians in \eqref{HH:1}, \eqref{JME:8}, \eqref{HH:2} and the solutions to the underlying dynamical systems are fixed as follows: for the bulk function after thinning,
\begin{equation}\label{JME:17}
	 q(t;\gamma)\sim\frac{\gamma}{2\pi\im},\ \ \ p(t;\gamma)\sim4\im t,\ \ \ \ \ \ \ t\downarrow 0;
\end{equation}
whereas for the distribution function of the largest eigenvalue after thinning the Hastings-McLeod solution is replaced by the Ablowitz-Segur solution \cite{AS}, i.e.
\begin{equation}\label{JME:18}
	 q(t;\gamma)\sim\sqrt{\gamma}\,\textnormal{Ai}(t),\ \ \ p(t;\gamma)=2q_t(t,\gamma),\ \ \ \ \ \ \ \  t\rightarrow+\infty.
\end{equation}
In addition, related to the distribution function of the smallest eigenvalue after thinning,
\begin{equation}\label{JME:19}
	q(t,\alpha;\gamma)\sim\frac{\sqrt{\gamma}\,t^{\frac{1}{2}\alpha}}{2^{\alpha}\Gamma(1+\alpha)},\ \ \ p(t,\alpha;\gamma)=\frac{2tq_t(t,\alpha;\gamma)}{q^2(t,\alpha;\gamma)-1},\ \ \ \ \ \ \ \ t\downarrow 0.
\end{equation}
\end{theo}
The proof of Theorem  \ref{theo:1} follows from a combination of standard arguments based on the Fredholm determinant representations of $F_B(t),F_S(t)$ and $F_H(t,\alpha)$, see Section \ref{sec:2} below. In order to prepare for our next objective we remind the reader that the thinning process removes correlations from the initial setup $\{\mu_j^S,\mu_j^H,\mu_j^B\}_{j=1}^n$, thus varying $\gamma$ we are able to interpolate between particle systems that obey random matrix theory statistics and systems modeled by more classical distribution families, e.g. Poisson and Weibull, see Subsection \ref{sec:16} below. This interpolation mechanism is well-known by now, see e.g. \cite{BP1,BP3,BB,BDIK2} and the analytic challenge lies in the derivation of tail expansions for $F_B(t;\gamma),F_S(t;\gamma)$ and $F_H(t,\alpha;\gamma)$ as $t\downarrow 0,t\rightarrow+\infty$ (bulk), $t\rightarrow\pm\infty$ (soft-edge) and $t\downarrow 0,t\rightarrow+\infty$ (hard-edge) which are uniform with respect to $\gamma\in[0,1]$.
\begin{rem} The uniformity requirement poses a clear challenge: the introduction of $\gamma$ into the boundary conditions in Theorem \ref{theo:1} 
has a very subtle effect on, both, analytic and asymptotic properties of $(q,p)$. For instance, in case of \eqref{JME:18}, solutions are bounded on the entire real axis for $\gamma\in[0,1)$, but unbounded (as $t\rightarrow-\infty$) once $\gamma=1$. Similar phenomena also occur for \eqref{JME:17} and \eqref{JME:19} and we shall return to these interesting phase transitions after the next two subsections.
\end{rem}
\subsection{Tail asymptotics and action integral formul\ae}\label{sec:14} The principal analytical question concerning the distribution functions  $F_B(t;\gamma)$, $F_H(t,\alpha;\gamma)$  and of $F_S(t;\gamma)$ is their {\it tail asymptotics}, i.e., the behavior of $F_B(t;\gamma),F_H(t,\alpha;\gamma)$ as $t\downarrow 0$ or $t \rightarrow \infty$, and of $F_S(t;\gamma)$ as $t\rightarrow\pm\infty$. In view of Theorem \ref{theo:1} we  realize at once that  half of the tail expansions are easy to compute. Indeed upon substitution of the boundary data \eqref{JME:17}, \eqref{JME:18}, \eqref{JME:19} into the Hamiltonian formul\ae\,we obtain immediately the leading order behavior of $F_B(t;\gamma),F_H(t,\alpha;\gamma)$ as $t\downarrow 0$ and of $F_S(t;\gamma)$ as $t\rightarrow+\infty$,
\begin{equation*}
	F_B(t;\gamma)=1-\frac{2\gamma}{\pi}t\big(1+o(1)\big),\ \ t\downarrow 0;\ \ \ \ \ \ \ F_S(t;\gamma)=1-\frac{\gamma}{16\pi}t^{-\frac{3}{2}}\e^{-\frac{4}{3}t^{\frac{3}{2}}}\big(1+o(1)\big),\ \ t\rightarrow+\infty;
\end{equation*}
and
\begin{equation*}
	F_H(t,\alpha;\gamma)=1-\frac{\gamma}{\Gamma^2(2+\alpha)}\left(\frac{t}{4}\right)^{\alpha+1}\big(1+o(1)\big),\ \ t\downarrow 0.
\end{equation*}
Most importantly, these expansions are uniform with respect to $\gamma\in[0,1]$ and (in case of $F_H(t,\alpha;\gamma)$) $\alpha$ chosen from compact subsets of $(-1,+\infty)\subset\mathbb{R}$. Much more challenging are the remaining three tails: for these we could in principle use Painlev\'e asymptotic information, see \cite{NIST}, Chapter $32$. For instance in case of \eqref{JME:8} it is known that
\begin{equation*}
	q(t;\gamma)=(-t)^{-\frac{1}{4}}\sqrt{\frac{v}{\pi}}\cos\left(\frac{2}{3}(-t)^{\frac{3}{2}}-\frac{v}{2\pi}\ln\big(8(-t)^{\frac{3}{2}}\big)+\phi\right)+\mathcal{O}\left(t^{-1}\right),\ \ t\rightarrow-\infty,
\end{equation*}
where $v=-\ln(1-\gamma),\phi=\frac{\pi}{4}-\textnormal{arg}\,\Gamma(\frac{v}{2\pi\im})$ and $\gamma\in[0,1)$ is fixed. In addition,
\begin{equation*}
	q(t;1)=\sqrt{-\frac{t}{2}}\left(1+\frac{1}{8t^3}+\mathcal{O}\left(t^{-6}\right)\right),\ \ \ t\rightarrow-\infty.
\end{equation*}
Thus, upon $t$-differentiation of the Hamiltonian formula in Theorem \ref{theo:1} and subsequent indefinite integration,
\begin{equation}\label{JME:20}
	\ln F_S(t;\gamma)=-\frac{2v}{3\pi}(-t)^{\frac{3}{2}}+\frac{v^2}{4\pi^2}\ln\big(8(-t)^{\frac{3}{2}}\big)+\mathcal{O}(1),\ \ t\rightarrow-\infty,\ \ \gamma\in[0,1);
\end{equation}
as well as
\begin{equation}\label{JME:21}
	\ln F_S(t;1)=\frac{t^3}{12}-\frac{1}{8}\ln(-t)+\mathcal{O}(1),\ \ t\rightarrow-\infty.
\end{equation}
A similar approach can be carried out for \eqref{HH:1} and \eqref{HH:2} once we use the relevant asymptotic information given in \cite{S, MT1, MT2, AK}. But in either case, the outlined method does not allow us to compute the $\mathcal{O}(1)$ terms (see \eqref{JME:20} and \eqref{JME:21}) in an efficient way.
And these terms are needed for the rigorous analysis of the phase transition as $\gamma\uparrow 1$, for instance $\mathcal{O}(1)$ in \eqref{JME:20} is bounded as $t\rightarrow-\infty$ but not as $\gamma\uparrow 1$. The problem of finding these terms  is sometime referred to as  the ``constant problem'' and this is  the main issue we are addressing in this paper.\smallskip

The usual approaches to the computation of the above mentioned outstanding constant factors,
as well as the description of the full transitional regime, are based on the utilization of additional integrable structures, i.e., 
 Fredholm determinant 
 formul\ae\,(see Section \ref{sec:2} below),
\begin{equation}\label{Fred:1}
	F_B(t;\gamma)=\det\big(1-\gamma K_{\sin}\upharpoonright_{L^2(-t,t)}\big),\ \ \ \ \ \ \ \ \ \ F_S(t;\gamma)=\det\big(1-\gamma K_{\textnormal{Ai}}\upharpoonright_{L^2(t,\infty)}\big),
\end{equation}
and
\begin{equation}\label{Fred:2}
	F_H(t,\alpha;\gamma)=\det\big(1-\gamma K_{\textnormal{Bess}}^{\alpha}\upharpoonright_{L^2(0,t)}\big).
\end{equation}
Here, $K_{\sin},K_{\textnormal{Ai}}$ and $K_{\textnormal{Bess}}^{\alpha}$ are the trace-class integral operators with kernels
\begin{equation}\label{JME:22}
	K_{\sin}(\lambda,\mu)=\frac{\sin(\lambda-\mu)}{\pi(\lambda-\mu)},\ \ \ \ \ \ \ \ \ \ \ K_{\textnormal{Ai}}(\lambda,\mu)=\int_0^{\infty}\textnormal{Ai}(\lambda+s)\textnormal{Ai}(\mu+s)\,\d s,
\end{equation}
and 
\begin{equation}\label{JME:23}
	K_{\textnormal{Bess}}^{\alpha}(\lambda,\mu)=\frac{1}{4}\int_0^1J_{\alpha}(\sqrt{\lambda s}\,)J_{\alpha}(\sqrt{\mu s}\,)\,\d s,
\end{equation}
where $J_{\alpha}(z)$ is the Bessel function of order $\alpha$. 
Based on these formul\ae\, one can now either use discretization techniques (e.g. representing $F_B(t;\gamma)$ as limit of a Toeplitz determinant \cite{D,K,DIKZ}, or $F_S(t;\gamma)$ as Hankel determinant limit \cite{DIK,K2}), or apply operator theoretical arguments \cite{BW,BuBu,E1,E2,W1}, or refer to the algebra of integrable operators \cite{IIKS,DIZ}.  With these tools at hand, it is possible to  improve \eqref{JME:20} and \eqref{JME:21} (see \cite{BB,DIK,BBD}),
\begin{equation*}
	\ln F_S(t;\gamma)=-\frac{2v}{3\pi}(-t)^{\frac{3}{2}}+\frac{v^2}{4\pi^2}\ln\big(8(-t)^{\frac{3}{2}}\big)+\ln\left(G\left(1+\frac{\im v}{2\pi}\right)G\left(1-\frac{\im v}{2\pi}\right)\right)+o(1),\ \ t\rightarrow-\infty,\ \gamma\in[0,1);
\end{equation*}
and
\begin{equation*}
	\ln F_S(t;1)=\frac{t^3}{12}-\frac{1}{8}\ln(-t)+\zeta'(-1)+\frac{1}{24}\ln 2+o(1),\ \ \ t\rightarrow-\infty
\end{equation*}
in terms of the Barnes $G$-function $G(z)$ and the Riemann zeta-function $\z(z)$.\smallskip

In this paper we present a new method for the derivation of tail expansions which does not rely on Fredholm, or Toeplitz, or Hankel determinant formul\ae\,but instead on the Hamiltonian system approach (see Theorem \ref{theo:1}) to the gap and distribution functions. In this approach the already available Painlev\'e asymptotic information (compare derivation of \eqref{JME:20} and \eqref{JME:21}) will be sufficient to obtain full leading order asymptotic information for $F_B(t;\gamma),F_S(t;\gamma)$ and $F_H(t,\alpha;\gamma)$ provided $\gamma\in[0,1)$ is fixed. Our method is based on the following action integral formul\ae\,which form our second result.
\begin{theo}\label{theo:2} Let $F_B(t;\gamma),F_S(t;\gamma)$ and $F_H(t,\alpha;\gamma)$ be defined as in \eqref{JME:15}, \eqref{JME:16} for $\gamma\in[0,1]$ and $(q,p)$ specified as in Theorem \ref{theo:1}. Then
\begin{equation}\label{JME:24}
	\ln F_B(t;\gamma)=t\,\mathcal{H}_{B}(q,p,t)-pq+I_{B}(t;\gamma),\ \ \ \ I_{B}(t;\gamma)=\int_0^t\big(pq_s-\mathcal{H}_{B}(q,p,s)\big)\,\d s,
\end{equation}
and
\begin{equation}\label{JME:25}
	\ln F_{S}(t;\gamma)=\frac{1}{3}\big(2t\,\mathcal{H}_{S}(q,p,t)-pq\big)+I_{S}(t;\gamma),\ \ \ \ I_{S}(t;\gamma)=-\int_t^{\infty}\big(pq_s-\mathcal{H}_{S}(q,p,s)\big)\,\d s
\end{equation}
where the integration path in the action integrals $I_B$ and $I_S$ is chosen on the real line. In addition,
\begin{equation}\label{JME:26}
	\ln F_{H}(t,\alpha;\gamma)=2t\,\mathcal{H}_{H}(q,p,t,\alpha)-L(t,\alpha;\gamma)+I_{H}(t,\alpha;\gamma),
\end{equation}
with
\begin{equation*}
	L(t,\alpha;\gamma)=\frac{\alpha^2}{2}\int_0^t\frac{q^2\,\d s}{s(q^2-1)},\ \ \ \ \ I_{H}(t,\alpha;\gamma)=\int_0^t\big(pq_s-\mathcal{H}_{H}(q,p,s,\alpha)\big)\,\d s,\ \ \ \ \alpha\geq 0;
\end{equation*}
and
\begin{equation*}
	L(t,\alpha;\gamma)=\frac{\alpha^2}{2}\int_0^t\frac{\d s}{s(q^2-1)},\ \ \ \ I_{H}(t,\alpha;\gamma)=\int_0^t\left(pq_s-\frac{\alpha^2}{2s}-\mathcal{H}_{H}(q,p,s,\alpha)\right)\,\d s,\ \ -1<\alpha<0.
\end{equation*}
The integration paths for $L$ and $I_H$ lie in the half-plane $\Re s>0$ and avoid the discrete set $\{s\in\mathbb{C}:\ q^2(s,\alpha;\gamma)=1\}$.
\end{theo}
In order to appreciate the usefulness of this theorem for the evaluation of constant factors in tail asymptotics, let us highlight 
the difficulties which one faces in the existing approaches to the problem. We will restrict ourselves to the sine -  kernel
distribution $F_B$ and  consider the asymptotic  scheme based on the theory of integrable Fredholm operators, cf. \cite{IIKS,DIZ}.\smallskip

We shall start with the classical differential identity which is the beginning of almost every study of Fredholm determinants, 
\begin{equation}\label{gammadif0}
\frac{\partial}{\partial\gamma} \ln F_{B}(t;\gamma) = \frac{\partial}{\partial\gamma} \ln
\det\bigl(1 -\gamma K_{\sin}\upharpoonright_{L^2(-t,t)}\bigr)= -\frac{1}{\gamma}\mbox{Trace}\, R_B
= -\frac{1}{\gamma} \int_{-t}^tR_{B}(\lambda, \lambda;t,\gamma)d\lambda.
\end{equation}
Here, $R_{B}$ is  the resolvent of the operator $K_{\sin}$ defined by the usual formula, 
$R_{B} = \gamma\bigl(1-\gamma K_{\sin}\bigr)^{-1}K_{\sin}$ and $R_{B}(\lambda, \mu;t,\gamma)$ is its kernel.
As the sine-kernel belongs to the class of integrable Fredholm operators (see, e.g., \cite{DIZ}),
the resolvent kernel $R_{B}(\lambda, \mu;t,\gamma)$ admits the following explicit representation in terms of the $2\times 2$ matrix
valued solution ${\bf Y}(\lambda) \equiv {\bf Y}(\lambda;t,\gamma)$ to a Riemann-Hilbert
problem (see RH problem (\ref{sineRHP}) in Section \ref{sec:4} for more detail),
\begin{equation}\label{resRHsine}
R(\lambda,\mu) = \frac{1}{2\pi i(\lambda-\mu)}\begin{bmatrix}\e^{-\im\mu}&-\e^{\im\mu}\end{bmatrix}
{\bf Y}^{-1}\left(\frac{\mu}{t}\right){\bf Y}\left(\frac{\lambda}{t}\right) \begin{bmatrix}\e^{\im\lambda}\\ \e^{-\im\lambda}\end{bmatrix}
\end{equation}
Therefore, if one knew the large $t$-asymptotics of the solution ${\bf{Y}}(\lambda;t,\gamma)$ to RHP \ref{sineRHP} which are uniform with respect 
to $\gamma$ and $\lambda$, the needed large $t$ asymptotics of $F_B$ {\it including the constant term}  could have been determined via 
double integration,
\begin{equation}\label{gammaint1}
 \ln F_{B}(t, \gamma) = -\int_{0}^{\gamma}\int_{-t}^tR_{B}(\lambda, \lambda;t;\gamma)\,\d\lambda \frac{\d\gamma}{\gamma}.
\end{equation}
In case $\gamma \in [0,1)$ all needed asymptotic information about ${\bf{Y}}(\lambda;t,\gamma)$ can indeed be extracted 
via the nonlinear steepest descent analysis of Riemann-Hilbert problem \ref{sineRHP} (\cite{BDIK2}; see 
also Section \ref{sec:4} where this analysis is reproduced). However, the relevant formul\ae, though explicit, are very complicated.
The asymptotic evaluation of the double integral in \eqref{gammaint1} becomes enormously difficult, and in fact has never been 
done. The main difficulty lies in the non-locality of the differential identity \eqref{gammadif0} and, as a consequence, of  the integral formula 
\eqref{gammaint1} in the variable $\lambda$. 
There is though a way to circumvent this double integration, and it uses the already mentioned  isomonodromy connection 
of the distribution function $F_B$. Indeed, with respect to $\lambda$, the matrix function  ${\bf{Y}}(\lambda)$ satisfies a linear differential
equation with rational coefficients (see \cite{DIZ}) which allows  one to evaluate the most challenging integral in (\ref{gammaint1}) -  namely, the integral 
in $\lambda$, and replace the differential identity (\ref{gammadif0}) by  the following  formula (cf. \cite{BI}),
$$
\frac{\partial}{\partial\gamma} \ln F_{B}(t;\gamma) = \frac{\partial}{\partial\gamma}\ln\det(1-\gamma K_{\sin}\upharpoonright_{L^2(-t,t)})
= -2\im t\frac{\partial}{\partial\gamma}Y_1^{11} +
\frac{\gamma}{2\pi\im}\textnormal{trace}\left\{\widehat{\bf Y}^{-1}(1)\frac{\partial}{\partial\gamma}\widehat{\bf Y}(1)\begin{bmatrix}
-1&1\cr-1&1\end{bmatrix}\right\}
$$
\begin{equation}\label{gammadif2}
- \frac{\gamma}{2\pi\im}\textnormal{trace}\left\{\widehat{\bf Y}^{-1}(-1)\frac{\partial}{\partial\gamma}\widehat{\bf Y}(-1)\begin{bmatrix}
-1&1\cr-1&1\end{bmatrix}\right\},
\end{equation}
where $Y_1^{11}$ and $\widehat{\bf Y}(\lambda)$ are defined at the beginning of Section \ref{sec:4} - see formulation of RHP \ref{sineRHP}, properties (3) and (4). This relation is still not very simple, 
but it involves only  {\it local}  characteristica of the solution ${\bf{Y}}(\lambda)$. This locality allows one to use \eqref{gammadif2} for evaluating  the
large $t$ asymptotics of $F_B$, although the calculations which one has to go through are still very tough.\smallskip

The value of Theorem \ref{theo:2} lies in the fact that it yields  an alternative to \eqref{gammadif2} local $\gamma$ - differential formula for 
$F_B$, as well as similar formul\ae\, for the other two distribution functions, which would simplify dramatically their asymptotic analysis. 
Indeed, identities \eqref{JME:24} - \eqref{JME:26} transform the original Hamiltonian integrals of Theorem \ref{theo:1}
to the {\it action integrals} $I_r$ plus explicit terms. The latter are either already localized, i.e. without any integrals, or integral terms as in \eqref{JME:26} that admit a straightforward Riemann-Hilbert representation (see Section \ref{sec:5} below). The great advantage of having the full classical action integral instead of its 
truncated form  lies in the  fundamental fact that the variational derivatives of the classical action and, in particular, the  $\gamma$ - derivatives of the classical actions $I_{B,S,H}$ are simple local functions of the canonical variables $p$ and $q$. For  instance, for the integral $I_B$, we would have,
$$
\frac{\partial I_B}{\partial\gamma} = \int_{0}^{t}\left(q_sp_{\gamma} + p (q_{\gamma})_s -
\frac{\partial \mathcal{H}_{B} }{\partial p}p_{\gamma} -\frac{\partial \mathcal{H}_{B} }{\partial q}q_{\gamma}\right)\d s; \quad 
\,\textnormal{where}\, f_{\gamma}=\frac{\partial f}{\partial\gamma},
$$
and integrating by parts the second term,
$$
\frac{\partial I_B}{\partial\gamma} = pq_{\gamma}\Big|_{s=0}^{t} +  \int_{0}^{t}\left(q_sp_{\gamma} - p_{s}q_{\gamma} -
\frac{\partial \mathcal{H}_{B} }{\partial p}p_{\gamma} -\frac{\partial \mathcal{H}_{B} }{\partial q}q_{\gamma}\right)\d s,
$$
$$
=  pq_{\gamma}\Big|_{s=0}^{t} +  \int_{0}^{t}\left(\Bigl(q_s -
\frac{\partial \mathcal{H}_{B} }{\partial p}\Bigr)p_{\gamma} -\Bigl(p_{s} +\frac{\partial \mathcal{H}_{B} }{\partial q}\Bigr)q_{\gamma}\right)\d s
= pq_{\gamma}
$$
where the remaining integral term vanishes due to the dynamical equations \eqref{HH:1} and $pq_{\gamma}|_{s=0} = 0$ because
of the boundary behavior of $q(t)$ and $p(t)$ at $t =0$. Similar calculations can be done for the other two action integrals and we arrive at
the following important formula\ae\,.
\begin{prop}\label{logic} With $(q,p)$ as in Theorem \ref{theo:1},
\begin{equation*}
	\frac{\partial I_r}{\partial\gamma}=pq_{\gamma}\ \ \ r=B,S;\ \ \ \ \ \ \ \ \ \ \ \frac{\partial I_{H}}{\partial\gamma}=\begin{cases}pq_{\gamma},&\alpha\geq 0\\ pq_{\gamma}-\frac{\alpha}{2\gamma},&-1<\alpha<0\end{cases};\ \ \ \ \ \ f_{\gamma}=\frac{\partial f}{\partial\gamma}.
\end{equation*}
Moreover,
\begin{equation*}
	\frac{\partial I_H}{\partial\alpha}=\begin{cases} pq_{\alpha}+\frac{1}{\alpha}L(t;\alpha,\gamma),&\alpha> 0\\
	pq_{\alpha}-\frac{\alpha}{2}\ln t+\alpha\frac{\d}{\d\alpha}\ln\big(2^{\alpha}\Gamma(1+\alpha)\big)+\frac{1}{\alpha}L(t,\alpha;\gamma),&-1<\alpha<0\end{cases};\ \ \ \ \ \ \ f_{\alpha}=\frac{\partial f}{\partial\alpha}.
\end{equation*}	
\end{prop}
One clearly notices how much simpler the differential formul\ae\, of this Proposition are than identity (\ref{gammadif2}).\smallskip

The formal proof of Theorem \ref{theo:2} is easy, and it can be obtained through $t$-differentiation of both sides of equations \eqref{JME:24}, \eqref{JME:25} and \eqref{JME:26}  with the simultaneous use of the respective Hamiltonian systems. This formal proof will be presented in Section \ref{sec:2}. The methodological deficiency of this proof is that it does not provide any clue on where equations \eqref{JME:24}, \eqref{JME:25} and \eqref{JME:26} came from. 
In Section \ref{sec:3_0} we outline an alternative proof of these formul\ae\, which, simultaneously, reveals their theoretical origin.  
This alternative proof is  based on the tau function interpretation of the gap/distribution functions $F_{B,S,H}$, and it uses the  extension of the Jimbo-Miwa-Ueno tau function differential
to a differential 1-form whose external derivative coincides with the corresponding symplectic form.
Hence  the connection of the distribution functions in question to the relevant action integrals is not an accident;
in fact, it is their  {\it intrinsic property}.

\subsection{Large gap expansions and phase transitions}\label{sec:16} As mentioned above, a direct application of \eqref{JME:24}, \eqref{JME:25} and \eqref{JME:26} is provided with the efficient and quick derivation of tail expansions for all three functions $F_B(t;\gamma), F_S(t;\gamma),F_H(t,\alpha;\gamma)$ in case $\gamma\in[0,1)$ and $\alpha>-1$ are fixed. Our third result is as follows.
\begin{theo}\label{theo:3} For any fixed $\gamma\in[0,1),\alpha>-1$ there exist positive constants $t_0=t_0(\gamma,\alpha)$ and $c_r=c_r(\alpha,\gamma),r=B,S,H$ such that
\begin{equation}\label{JME:31}
	\ln F_B(t;\gamma)=-\frac{2v}{\pi}t+\frac{v^2}{2\pi^2}\ln(4t)+2\ln\big(G\left(1+\frac{\im v}{2\pi}\right)G\left(1-\frac{\im v}{2\pi}\right)\big)+r_B(t;\gamma)\ \ \ \forall\,t\geq t_0,
\end{equation}
followed by
\begin{equation}\label{JME:32}
	\ln F_S(t;\gamma)=-\frac{2v}{3\pi}|t|^{\frac{3}{2}}+\frac{v^2}{4\pi^2}\ln\big(8|t|^{\frac{3}{2}}\big)+\ln\big(G\left(1+\frac{\im v}{2\pi}\right)G\left(1-\frac{\im v}{2\pi}\right)\big)+r_S(t;\gamma)\ \ \ \forall\,(-t)\geq t_0,
\end{equation}
and concluding with 
\begin{equation}\label{JME:33}
	\ln F_H(t,\alpha;\gamma)=-\frac{v}{\pi}\sqrt{t}+\frac{v^2}{8\pi^2}\ln(16t)+\frac{\alpha}{2}v+\ln\big(G\left(1+\frac{\im v}{2\pi}\right)G\left(1-\frac{\im v}{2\pi}\right)\big)+r_H(t,\alpha;\gamma)\ \ \forall\,t\geq t_0.
\end{equation}
Here, $v=-\ln(1-\gamma)\in[0,+\infty)$, $G(z)$ is the Barnes $G$-function and the $t$-differentiable error terms satisfy
\begin{equation*}
	\big|r_B(t;\gamma)\big|\leq\frac{c_B(\gamma)}{t}\ \ \forall\,t\geq t_0;\ \ \ \big|r_S(t;\gamma)\big|\leq\frac{c_S(\gamma)}{|t|^{\frac{3}{4}}}\ \ \forall\,(-t)\geq t_0;\ \ \ \ \big|r_H(t,\alpha;\gamma)\big|\leq\frac{c_H(\alpha;\gamma)}{\sqrt{t}}\ \ \forall\,t\geq t_0.
\end{equation*}
\end{theo}
\begin{rem} Expansion \eqref{JME:31} was first derived in \cite{BuBu} with the indicated error estimate given in \cite{BDIK2}. This expansion also follows from the general results of \cite{BW}. The first proof of \eqref{JME:32} appeared recently in \cite{BB}, resolving an earlier conjecture posed in \cite{BCI}. Expansion \eqref{JME:33}, to the best of our knowledge, is completely new. 
\end{rem}
Expansions \eqref{JME:31}, \eqref{JME:32} and \eqref{JME:33} are valid provided each particle $\mu_j^r,r=S,H,B$ has been removed with positive probability $1-\gamma\in(0,1]$. As such they are in sharp contrast to the following three expansions (compare e.g. \eqref{JME:21} above), \cite{E1,K,DIKZ}
\begin{equation}\label{DW}
	\ln F_B(t;1)=-\frac{t^2}{2}-\frac{1}{4}\ln t+3\z'(-1)+\frac{1}{12}\ln 2+o(1),\ \ \ t\rightarrow+\infty;
\end{equation}
and \cite{DIK,BBD}
\begin{equation}\label{DIKas}
	\ln F_S(t;1)=\frac{t^3}{12}-\frac{1}{8}\ln(-t)+\z'(-1)+\frac{1}{24}\ln 2+o(1),\ \ \ t\rightarrow-\infty;
\end{equation}
as well as \cite{E2,DKV}
\begin{equation}\label{TE}
	\ln F_H(t,\alpha;1)=-\frac{t}{4}+\alpha\sqrt{t}-\frac{\alpha^2}{4}\ln t+\ln\left(\frac{G(1+\alpha)}{(2\pi)^{\frac{\alpha}{2}}}\right)+o(1),\ \ \ t\rightarrow+\infty.
\end{equation}
In each case the gap or distribution function approaches zero (or one in case of $1-F_H$) faster in the complete setup opposed to thinned version thereof. Thus a non-trivial phase transition occurs as $\gamma\uparrow 1$. In fact, using Theorem \ref{theo:3} and our previous discussion in Subsection \ref{sec:14} we see that
\begin{equation}\label{JME:34}
	\lim_{\gamma\downarrow 0}F_B(t\gamma^{-1};\gamma)=\begin{cases}\e^{-\frac{2}{\pi}t},&t\geq 0\\ 0,&t<0\end{cases},\ \ \ \ \ \ \ \ \lim_{\gamma\downarrow 0}F_S(t\gamma^{-\frac{2}{3}};\gamma)=\begin{cases} 1,&t>0\\ \e^{-\frac{2}{3\pi}|t|^{\frac{3}{2}}},&t\leq 0\end{cases},
\end{equation}
and
\begin{equation}\label{JME:35}
	\lim_{\gamma\downarrow 0}\big(1-F_H(t\gamma^{-2},\alpha;\gamma)\big)=\begin{cases}1-\e^{-\frac{1}{\pi}\sqrt{t}},&t\geq 0\\ 0,&t<0\end{cases}.
\end{equation}
Hence by varying $\gamma\in[0,1]$ we interpolate between particle systems obeying random matrix theory statistics ($\gamma=1$) and systems that follow Poisson statistics (bulk) or Weibull statistics (soft-edge and hard-edge).
\begin{rem} The occurrence of (transformed) Weibull distribution functions in the limits of $F_S$ and $1-F_H$ in \eqref{JME:34} and \eqref{JME:35} is consistent with the thinning process. We are effectively dealing with a sequence of independent random variables once $\gamma\downarrow 0$, and for such a sequence its extreme values (which are described by $F_S$ and $1-F_H$) follow generically either Gumbel, Fr\'echet or Weibull statistics. For the same reason we shouldn't expect either of these three families in the limit $\gamma\downarrow 0$ for the bulk function $F_B$.
\end{rem}
\subsection{Numerical comparison} We offer a short comparison of the results given in Theorem \ref{theo:3} to the numerically computed values of $F_B(t;\gamma),F_S(t;\gamma)$ and $F_H(t,\alpha;\gamma)$. Those values were calculated by MATLAB implementing the algorithm given in \cite{Bor}, i.e. we discretize the relevant Fredholm determinants \eqref{Fred:1}, \eqref{Fred:2} by the Nystr\"om method using an $m$-point Gauss-Legendre quadrature rule. The results are shown in Figure \ref{num1} for the bulk gap function and in Figures \ref{num2}, respectively \ref{num3}, \ref{num4} and \ref{num5} for the extremal distribution functions.
\begin{center}
\begin{figure}[tbh]
\includegraphics[width=12.5cm,height=6.5cm]{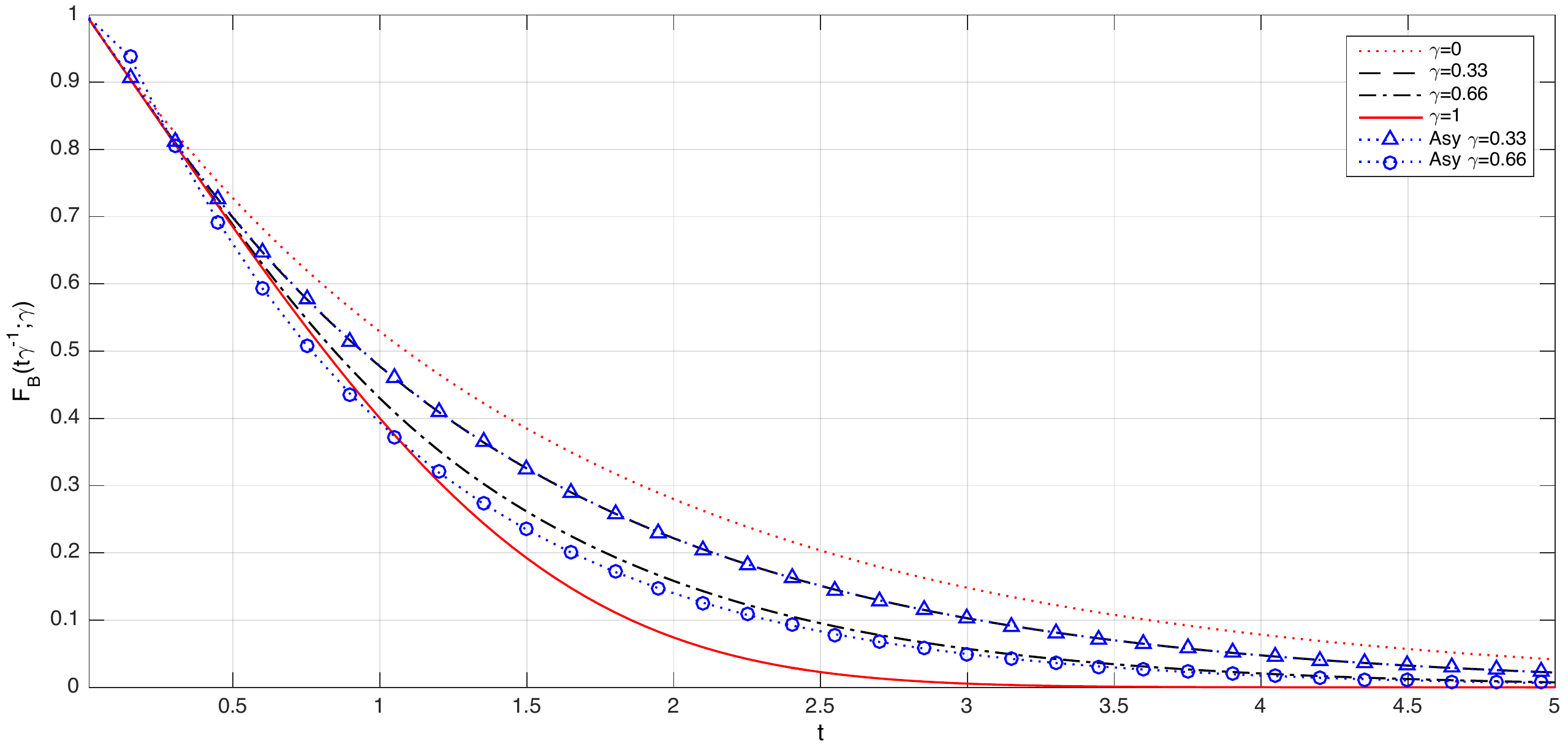}
\caption{Plot of the gap function $F_B(t\gamma^{-1};\gamma)$ for various values of $\gamma\in[0,1]$. The result is computed with $m=50$ quadrature points and checked against \eqref{JME:31} in blue.}
\label{num1}
\end{figure}
\end{center} 
\begin{center}
\begin{figure}[tbh]
\includegraphics[width=12.65cm, height=6.5cm]{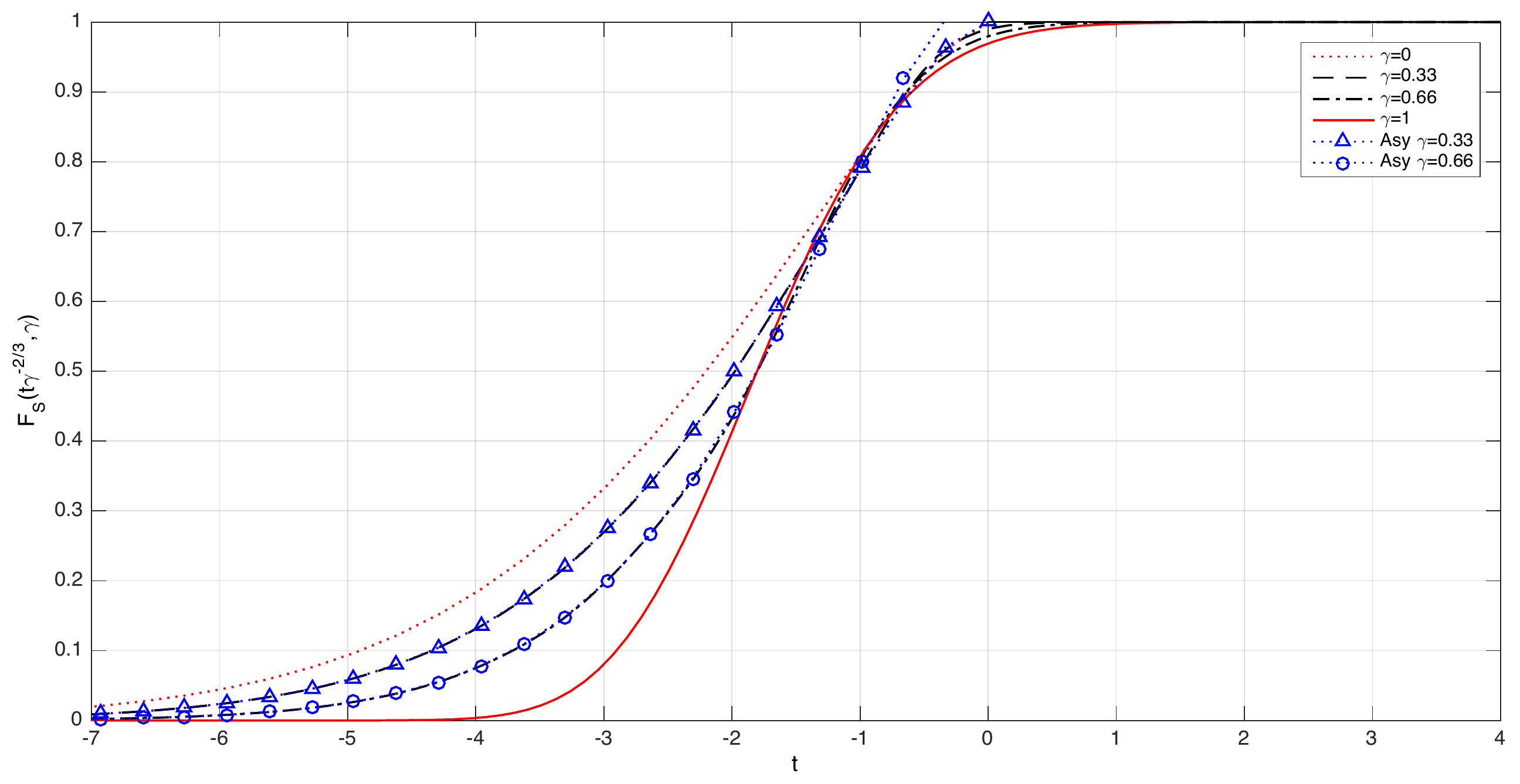}
\caption{Plot of the distribution function $F_S(t\gamma^{-\frac{2}{3}};\gamma)$ for various values of $\gamma\in[0,1]$. The result is computed with $m=50$ quadrature points and checked against \eqref{JME:32} in blue.}
\label{num2}
\end{figure}
\end{center} 
\subsection{The constant problem}

The exact evaluation of constant factors in asymptotic expansions of distribution, gap or correlation functions occurring in statistical mechanics or field theories is a long standing and challenging problem. The first rigorous solution of a constant problem for Painlev\'e equations (a special Painlev\'e~III transcendent appearing in the Ising model) has been obtained in the work of Tracy \cite{T}.  Other constant problems have been studied in the works \cite{BT,BB,K,E1,E2,DIKZ,DIK,DKV,Liso11} and \cite{BBD}.  
The tau functions that appear in all these papers 
correspond to very special families of Painlev\'e functions, and, as it has already been mentioned above, the success in their analysis  
was due to  the presence of operator theoretical structures (Fredholm-, Toeplitz-, Hankel-determinants). The first results concerning the general two-parameter families of solutions of Painlev\'e equations have been obtained
only recently in \cite{ILT13,ILT14}. These works  are based
on  {\it conformal block representations}  of isomonodromic tau functions --- see \cite{GIL12,GIL13,ILTe}.
 Although very powerful, the conformal block approach still has to be put on  rigorous ground. In the recent  papers \cite{IP} and  \cite{ILP}, it was shown that with the help of Riemann-Hilbert techniques the conjectural formul\ae\,of \cite{ILT14} and  \cite{ILT13} for the constant factor in the asymptotics of the  Painlev\'e III and Painlev\'e VI tau functions can be proven. 
Moreover, a new result - the formula for the constant factor in the asymptotics of a generic Painev\'e II tau function was established. Finally,
in the most recent work \cite{GL} in the case of the Painless\'e VI, a determinant formula for the generic Painlev\'e VI tau function has been obtained, which also, in particular, provides a  rigorous proof of the results of \cite{ILT13}.\smallskip

A central role in the constructions  of papers \cite{IP} and  \cite{ILP} is played by an extension of the Jimbo-Miwa-Ueno differential to the full
space of the extended monodromy data of the associated linear systems. This extension has been inspired by the works of Malgrange
\cite{Mal}  and Bertola \cite{Ber2, Ber1}, and, as a by-product, it  has established  a very interesting new fact
about the  Jimbo-Miwa-Ueno differential. It turns out that the original Jimbo-Miwa-Ueno differential form coincides
up to a total derivative with the classical action differential. This in turn lead to
Theorem \ref{theo:2} which, as we have already briefly explained, yields a new and much simpler way 
to derive the large gap asymptotics, as featured in this paper. In other words, a principal methodological message of our paper is that the original, very special, cases of tau functions have also benefited from the scheme that has been designed for the analysis of the general cases.

\subsection{Outline of paper} A short outline for the remaining sections is as follows. We derive Theorem \ref{theo:1} and Theorem \ref{theo:2} in Section \ref{sec:2} through the use of Fredholm determinant formul\ae\, for $F_B(t),F_S(t)$ and $F_H(t,\alpha)$ and straightforward differentiation. 
In Section \ref{sec:3_0} we outlined the above mentioned alternative proof of Theorem \ref{theo:2} based on the theory of isomonodromic tau functions. This alternative proof explains the origin and general theoretical meaning of the theorem's statement. Section \ref{sec:3} is devoted to the proof of \eqref{JME:32} where we rely on the well-known Riemann-Hilbert representation \cite{FIKN} of the Ablowitz-Segur solution to Painlev\'e-II. The underlying Riemann-Hilbert problem is solved in \cite{FIKN} and we only require   a minor extension of the known nonlinear steepest descent techniques (\cite{FIKN} focuses on $q(t;\gamma)$ only, but here we need $p(t;\gamma)$ and $\mathcal{H}_S(q,p,t)$ as well). Somewhat similar is our approach in the proof of \eqref{JME:31} given in Section \ref{sec:4}. The Riemann-Hilbert problem has been analyzed previously in \cite{DIZ,BDIK2} and we simply read off $p(t;\gamma)$ and $\mathcal{H}_B(q,p,t)$. This changes in Section \ref{sec:5} when we address \eqref{JME:33}. The Riemann-Hilbert problem for $q(t,\alpha;\gamma)$ and $p(t,\alpha;\gamma)$ is known from \cite{B2} but was not asymptotically solved in the required scaling regime when $t\rightarrow+\infty$ and $\gamma\in[0,1)$ is fixed. For this reason we provide all necessary details following the roadmap of \cite{DZ}: matrix factorizations, a $g$-function transformation, contour deformations, local model problems with Bessel and confluent hypergeometric functions and finally small norm estimates and iterations. The result of these steps is summarized in Theorem \ref{DZBess}. After that we simply extract all relevant expansions and combine them in \eqref{JME:26}, leading to the final expansion \eqref{JME:33}.
%
\begin{center}
\begin{figure}[tbh]
\includegraphics[width=12.5cm, height=6.5cm]{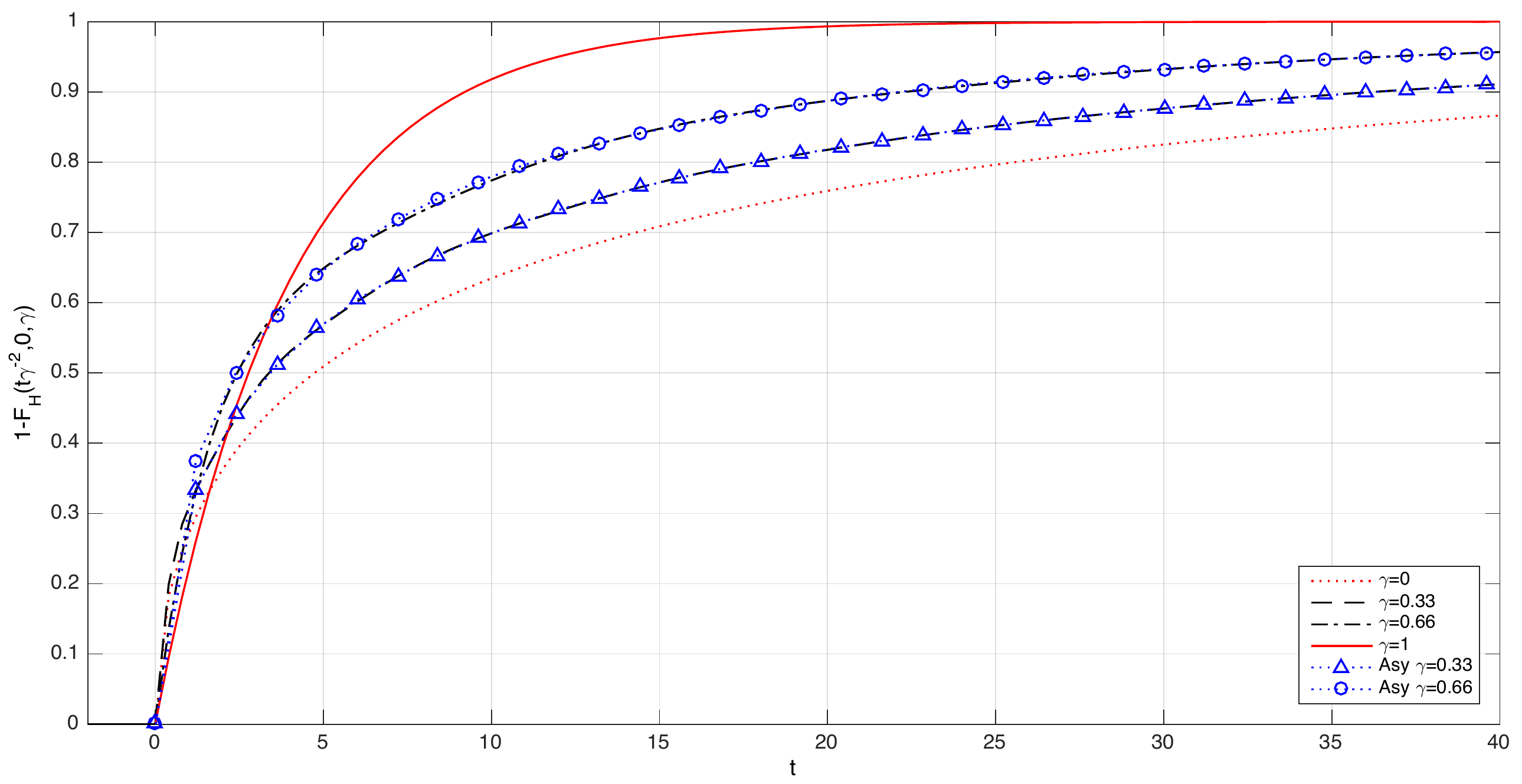}
\caption{Plot of the distribution function $1-F_H(t\gamma^{-2},0;\gamma)$ for various values of $\gamma\in[0,1]$. The result is computed with $m=50$ quadrature points and checked against \eqref{JME:33} for $\alpha=0$ in blue.}
\label{num3}
\end{figure}
\end{center} 
\begin{center}
\begin{figure}[tbh]
\includegraphics[width=12.5cm, height=6.5cm]{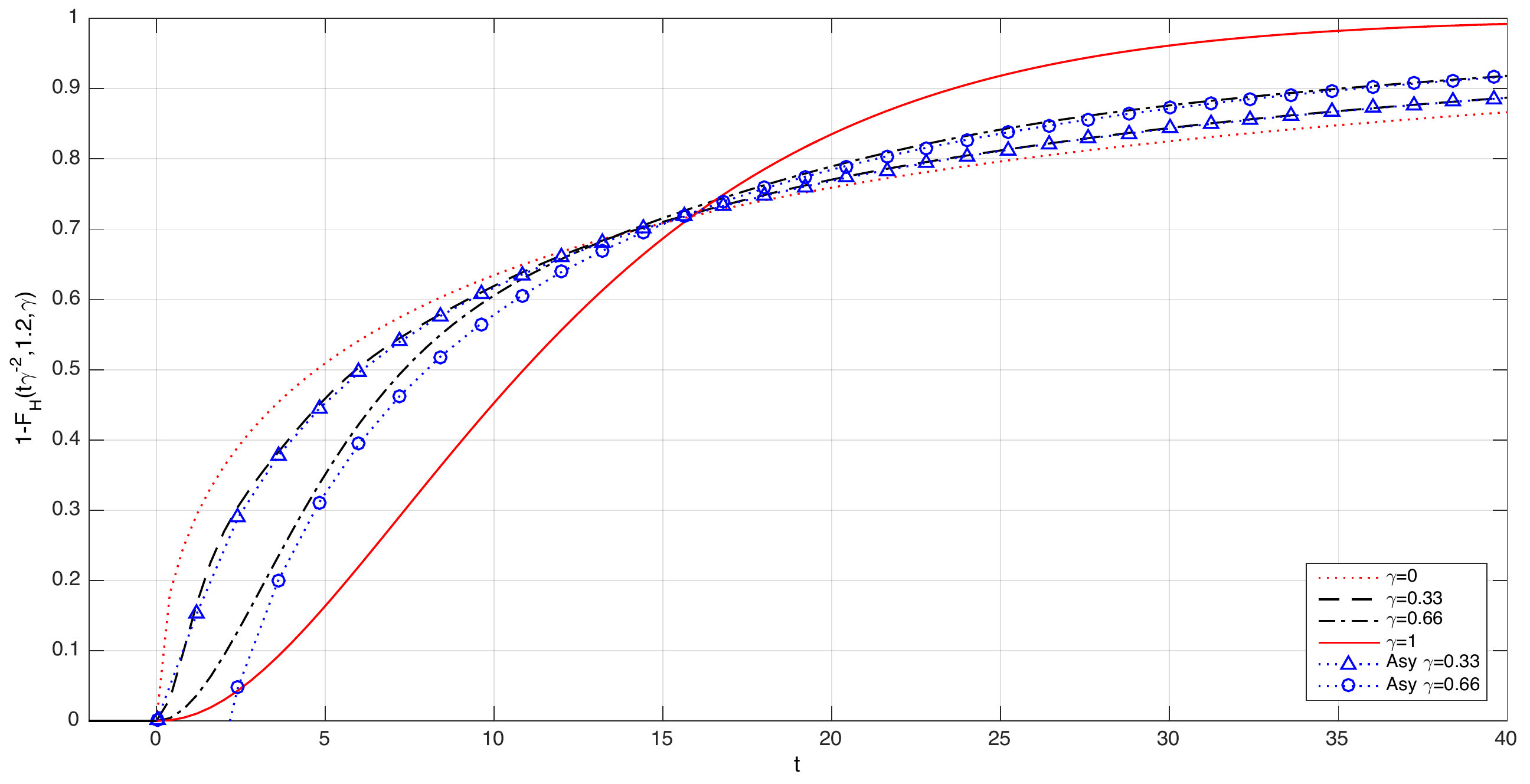}
\caption{Plot of the distribution function $1-F_H(t\gamma^{-2},1.2;\gamma)$ for various values of $\gamma\in[0,1]$. The result is computed with $m=50$ quadrature points and checked against \eqref{JME:33} for $\alpha=1.2$ in blue.}
\label{num4}
\end{figure}
\end{center} 
\begin{center}
\begin{figure}[tbh]
\includegraphics[width=12.5cm, height=6.6cm]{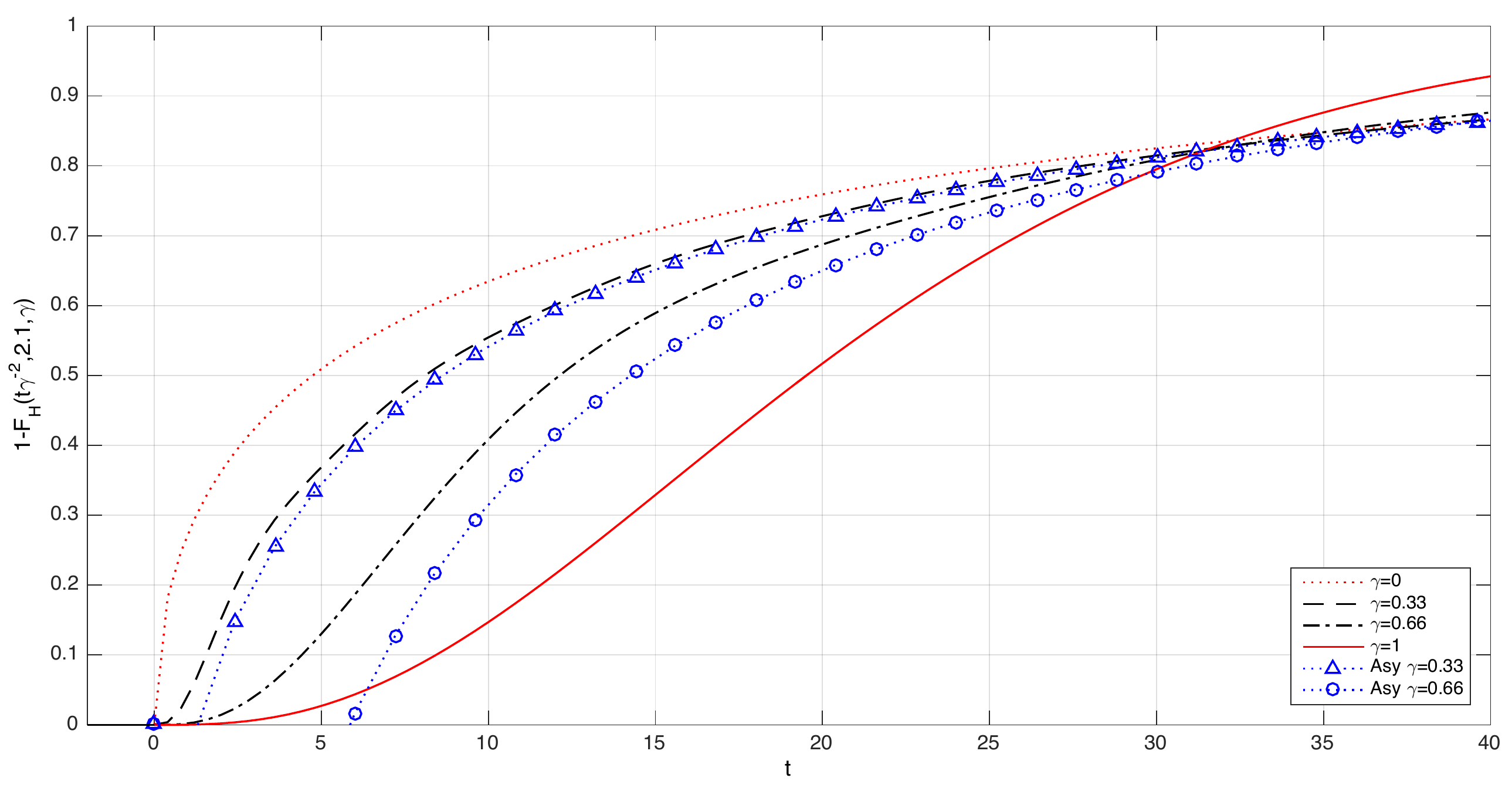}
\caption{Plot of the distribution function $1-F_H(t\gamma^{-2},2.1;\gamma)$ for various values of $\gamma\in[0,1]$. The result is computed with $m=50$ quadrature points and checked against \eqref{JME:33} for $\alpha=2.1$ in blue.}
\label{num5}
\end{figure}
\end{center} 

\section{Proof of Theorem \ref{theo:1} and Theorem \ref{theo:2}}\label{sec:2}
For the proof of Theorem \ref{theo:1} we shall rely on \cite{F}, Chapter $9$.
\begin{proof} Recall the well known \cite{TW1,TW2} Fredholm representations of the limiting distribution and gap functions in the complete Wishart ensemble,
\begin{equation*}
	F_B(t)=\det\big(1-K_{\sin}\upharpoonright_{L^2(-t,t)}\big),\ \ \ F_S(t)=\det\big(1-K_{\textnormal{Ai}}\upharpoonright_{L^2(t,\infty)}\big),\ \ \ F_H(t,\alpha)=\det\big(1-K_{\textnormal{Bess}}^{\alpha}\upharpoonright_{L^2(0,t)}\big),
\end{equation*}
using the kernels from \eqref{JME:22} and \eqref{JME:23}. Also, the limiting probability of having exactly $m\in\mathbb{Z}_{\geq 0}$ bulk, or soft-edge or hard-edge scaled eigenvalues $\mu_j^r, r=B,S,H$ in the interval $(-t,t)$, or $(t,\infty)$ or $(0,t)$ equals \cite{F},
\begin{equation*}
	E_B(m,(-t,t))=\frac{(-1)^m}{m!}\frac{\partial^m}{\partial\xi^m}E_B((-t,t),\xi)\Big|_{\xi=1},\ \ \ E_S(m,(t,\infty))=\frac{(-1)^m}{m!}\frac{\partial^m}{\partial\xi^m}E_S((t,\infty),\xi)\Big|_{\xi=1}
\end{equation*}
and
\begin{equation*}
	E_H(m,(0,t),\alpha)=\frac{(-1)^m}{m!}\frac{\partial^m}{\partial\xi^m}E_H((0,t),\xi,\alpha)\Big|_{\xi=1}
\end{equation*}
in terms of the generating functions
\begin{equation}\label{JME:36}
	E_B((-t,t),\xi)=\det\big(1-\xi K_{\sin}\upharpoonright_{L^2(-t,t)}\big),\ \ \ \ \ E_S((t,\infty),\xi)=\det\big(1-\xi K_{\textnormal{Ai}}\upharpoonright_{L^2(t,\infty)}\big)
\end{equation}
and
\begin{equation*}
	E_H((0,t),\xi,\alpha)=\det\big(1-\xi K_{\textnormal{Bess}}^{\alpha}\upharpoonright_{L^2(t,\infty)}\big).
\end{equation*}
Returning to \eqref{JME:16} (the case of the thinned extremal distributions is handled analogously) we have then
\begin{eqnarray*}
	F_B(t,\gamma)=\sum_{m=0}^{\infty}E_B(m,(-t,t))(1-\gamma)^m&=&\sum_{m=0}^{\infty}\frac{1}{m!}\frac{\partial^m}{\partial\xi^m}E_B((-t,t),\xi)\Big|_{\xi=1}(\gamma-1)^m\\
	&=&E_B((-t,t),\gamma)=\det\big(1-\gamma K_{\sin}\upharpoonright_{L^2(-t,t)}\big),
\end{eqnarray*}
where we used the definition of the incomplete Wishart ensemble in the first equality (each particle is removed independently with probability $1-\gamma$), identity \eqref{JME:36} in the second and Taylor's theorem in the third. Since we have now established \eqref{Fred:1} and \eqref{Fred:2}, the remainder of the proof (Hamiltonian representations and boundary conditions) follows at once from \cite{F}, Chapter $9$: indeed, for the limiting gap function, use \cite{F},(9.27) and Proposition 3.33. For the limiting distribution function of the largest eigenvalue after thinning, use \cite{F}, (9.26) and (9.43). Finally, for the limiting distribution function of the smallest eigenvalue after thinning, use \cite{F}, (9.62), (9.67) and (9.69).
\end{proof}
We now address Theorem \ref{theo:2}.
\begin{proof}
Take the $t$-derivative of the right hand side in \eqref{JME:24},
\begin{equation*}
	\frac{\partial}{\partial t}\ln F_B(t;\gamma)=t\big(\mathcal{H}_B\big)_t-p_tq,
\end{equation*}
and now use the Hamiltonian system \eqref{HH:1},
\begin{equation*}
	\frac{\partial}{\partial t}\ln F_B(t;\gamma)=t\big(\mathcal{H}_B\big)_t-p_tq=\mathcal{H}_B(q,p,t).
\end{equation*}
Thus, recall Proposition \ref{theo:1}, left and right hand side in \eqref{JME:24} can only differ by a $t$-independent constant. But from \eqref{JME:17} we find $\mathcal{H}_{\textnormal{B}}(q,p,t)\sim-\frac{2\gamma}{\pi},t\downarrow 0$, and thus
\begin{equation*}
	\ln F_B(t;\gamma)=\int_0^t\mathcal{H}_B\big(q(s;\gamma),p(s;\gamma),s\big)\,\d s=-\frac{2\gamma t}{\pi}+\mathcal{O}\big(t^2\big),\ \ t\downarrow 0,
\end{equation*}
which matches exactly the small $t$-behavior of the right hand side in \eqref{JME:24}, hence the aforementioned constant is zero. Next, take $t$-derivatives of the right hand side in \eqref{JME:25},
\begin{equation*}
	\frac{\partial}{\partial t}\ln F_S(t;\gamma)=-\frac{1}{3}\mathcal{H}_S+\frac{2t}{3}\big(\mathcal{H}_S\big)_t-\frac{1}{3}p_tq+\frac{2}{3}pq_t,
\end{equation*}
and use the system \eqref{JME:8} which leads to
\begin{equation*}
	\frac{\partial}{\partial t}\ln\textnormal{F}_S(t;\gamma)=\mathcal{H}_S(q,p,t).
\end{equation*}
But from \eqref{JME:18},
\begin{equation*}
	\ln F_S(t;\gamma)=-\int_t^{\infty}\mathcal{H}_S\big(q(s;\gamma),p(s;\gamma),s\big)\,\d s=-\frac{\gamma}{16\pi}t^{-\frac{3}{2}}\e^{-\frac{4}{3}t^{\frac{3}{2}}}\big(1+o(1)\big),\ \ \ t\rightarrow+\infty
\end{equation*}
which again matches the large positive $t$-behavior of the right hand side in \eqref{JME:25}, so the identity follows. Finally turn to \eqref{JME:26} and take $t$-derivatives of both sides,
\begin{equation*}
	\frac{\partial}{\partial t}\ln F_H(t,\alpha;\gamma)=\mathcal{H}_H+2t\big(\mathcal{H}_H\big)_t-\frac{\alpha^2}{2t}\frac{q^2}{q^2-1}+pq_t,\ \ \ \ \alpha>-1.
\end{equation*}
But with \eqref{HH:2} this implies that
\begin{equation*}
	\frac{\partial}{\partial t}\ln F_H(t,\alpha;\gamma)=\mathcal{H}_H,
\end{equation*}
and thus both sides in \eqref{JME:26} can only differ by a $t$-independent term. As mentioned in Remark \ref{poles}, the Hamiltonian is integrable on $(0,t)$ for $t>0$ and we find (compare Subsection \ref{sec:14}),
\begin{equation}\label{v:1}
	\int_0^t\mathcal{H}_H\big(q(s,\alpha;\gamma),p(s,\alpha;\gamma),s,\alpha)\,\d s=-\frac{\gamma}{\Gamma^2(2+\alpha)}\left(\frac{t}{4}\right)^{\alpha+1}\big(1+o(1)\big),\ \ \ t\downarrow 0.
\end{equation}
On the other hand the integrands of $L(t,\alpha;\gamma)$ and $I_H(t,\alpha;\gamma)$, see \eqref{JME:26}, are singular at all points $t_k\in(0,+\infty)$ where $q^2(t_k,\alpha;\gamma)=1$. In fact the differential equation \eqref{S:8} leads to the Taylor expansion
\begin{equation*}
	q(t,\alpha;\gamma)=\pm 1+d_k(t-t_k)+\mathcal{O}\big((t-t_k)^2\big),\ \ \ \ d_k^2t_k^2=\frac{1}{4}\alpha^2,\ \ \ \ |t-t_k|<r,
\end{equation*}
so that near $t_k$,
\begin{equation*}
	pq_t-\mathcal{H}_H(q,p,t,\alpha)=\pm\frac{\alpha^2}{4d_kt_k}\frac{1}{t-t_k}+\mathcal{O}(1),\ \ \ \ \frac{\alpha^2q^2}{2t(q^2-1)}=\pm\frac{\alpha^2}{4d_kt_k}\frac{1}{t-t_k}+\mathcal{O}(1),
\end{equation*}
and
\begin{equation*}
	\frac{\alpha^2}{2t(q^2-1)}=\pm\frac{\alpha^2}{4d_kt_k}\frac{1}{t-t_k}+\mathcal{O}(1).
\end{equation*}
For this reason we choose the path of integration for $L$ and $I_H$ in the right half-plane from $s=0$ to $s=t$ and we avoid all points $t_k$. With this choice, for $\alpha\geq 0$,
\begin{equation*}
	L(t,\alpha;\gamma)=-\frac{\alpha\gamma}{2}\left(\frac{t}{4}\right)^{\alpha}\frac{1}{\Gamma^2(1+\alpha)}\big(1+o(1)\big),\ \ \ t\downarrow 0,
\end{equation*}
as well as
\begin{equation*}
	I_H(t,\alpha;\gamma)=-\frac{\alpha\gamma}{2}\left(\frac{t}{4}\right)^{\alpha}\big(1+o(1)\big)+\frac{\gamma}{\Gamma^2(2+\alpha)}\left(\frac{t}{4}\right)^{\alpha+1}\big(1+o(1)\big),\ \ t\downarrow 0.
\end{equation*}
This implies that 
\begin{equation*}
	2t\,\mathcal{H}_H(q(t,\alpha;\gamma),p(t,\alpha;\gamma),t,\alpha)-L(t,\alpha;\gamma)+I_H(t,\alpha;\gamma)=\mathcal{O}\left(t^{\alpha+1}\right),\ \ t\downarrow 0;\ \ \ \alpha\geq 0
\end{equation*}
which matches in turn the vanishing order in \eqref{v:1}, i.e. the aforementioned $t$-independent term is identically zero for $\alpha\geq 0$. For $-1<\alpha<0$, we have instead, as $t\downarrow 0$,
\begin{equation*}
	L(t,\alpha;\gamma)=-\frac{\alpha}{2\gamma}\Gamma^2(1+\alpha)\left(\frac{t}{4}\right)^{-\alpha}\big(1+o(1)\big),\ \ \ 
	I_H(t,\alpha;\gamma)=\mathcal{O}\big(t^{-\alpha}\big)+\frac{\gamma}{\Gamma^2(2+\alpha)}\left(\frac{t}{4}\right)^{\alpha+1}\big(1+o(1)\big),
\end{equation*}
so again both sides in \eqref{JME:26} vanish as $t\downarrow 0$, i.e. also for $-1<\alpha<0$ the identity holds true.
\end{proof}

\section{Isomonodromic tau function and alternative proof of Theorem \ref{theo:2}}\label{sec:3_0}
As it has already been mentioned in our introduction, in this section we outline an alternative proof of Theorem \ref{theo:2} which
is based on the Jimbo-Miwa-Ueno theory of the isomonodromic tau function. We will restrict ourselves to the soft edge  case 
\eqref{JME:24} only. The bulk and hard edge cases can be done in a similar way.
\subsection{Lax system and classical Jimbo-Miwa-Ueno differential}\label{sec:12} 
Consider  the following $2\times 2$ system of ordinary differential equations in the complex $\lambda$-plane,
\begin{equation}\label{JME:10}
	\frac{\d{\bf X}}{\d\lambda}=\left\{-4\im\lambda^2\sigma_3+4\im\lambda\begin{bmatrix}0 & q\\ -q &0\end{bmatrix}+\begin{bmatrix}-\im t-2\im q^2 & -p\\ -p & \im t+2\im q^2\end{bmatrix}\right\}{\bf X}\equiv {\bf A}(\lambda,t){\bf X},\ \ \ \ \sigma_3=\begin{bmatrix}1 & 0\\ 0 & -1\end{bmatrix},
\end{equation}
where $t,q,p$ are viewed as external parameters. This system has an irregular singular point at $\lambda=\infty$ of Poincar\'e rank $3$ thus in turn (cf. \cite{FIKN}) seven canonical solutions $\{{\bf X}_j(\lambda),\lambda\in\Omega_j\}_{j=1}^7$ to \eqref{JME:10} exist which are uniquely specified by the asymptotic expansion 
\begin{equation}\label{JME:11}
	{\bf X}_j(\lambda)\sim\widehat{{\bf X}}(\lambda)\e^{-\im(\frac{4}{3}\lambda^3+t\lambda)\sigma_3},\ \ \lambda\rightarrow\infty,\ \lambda\in\Omega_j.
\end{equation}
Here
\begin{equation*}
	\Omega_j=\left\{\lambda\in\mathbb{C}:\ \textnormal{arg}\,\lambda\in\left(\frac{\pi}{3}(j-2),\frac{\pi}{3}j\right)\right\},\ \ \ j=1,2,\ldots,7,
\end{equation*}
denote the canonical sectors, and $\widehat{{\bf X}}(\lambda)$ is the formal series,
\begin{equation}\label{Xhat}
 \widehat{{\bf X}}(\lambda) = \mathbb{I} + \sum_{k=1}^{\infty}\frac{{\bf X}_k}{\lambda^k},
 \end{equation}
whose  matrix coefficients are explicitly expressed in terms of $q$ and $p$; for instance,
\begin{equation}\label{X1}
{\bf X}_1 = \begin{bmatrix}-\frac{\im}{2}(\frac{p^2}{4}-tq^2-q^4) & \frac{q}{2}\\ \frac{q}{2} & \frac{\im}{2}(\frac{p^2}{4}-tq^2-q^4)\end{bmatrix}.
\end{equation}
The  space $\mathcal{M}$ of monodromy data of system \eqref{JME:10} is generically  two dimensional over the field of complex numbers, cf. \cite{FIKN}, and it consists  of the non-trivial entries in the {\it Stokes matrices }
\begin{equation*}
	{\bf S}_j={\bf X}_j^{-1}(\lambda){\bf X}_{j+1}(\lambda)=\begin{cases}\begin{bmatrix}1 & 0\\ s_j & 1\end{bmatrix},&j\equiv 1\mod 2\smallskip\\ \begin{bmatrix} 1 & s_j\\ 0 & 1\end{bmatrix},&j\equiv 0\mod 2\end{cases},
\end{equation*}
which satisfy the following  cyclic and symmetry  constraints 
$$
{\bf S}_1{\bf S}_2{\bf S}_3{\bf S}_4{\bf S}_5{\bf S}_6=\mathbb{I}, \quad
\sigma_2{\bf S}_j\sigma_2= {\bf S}_{j+3}, \quad  \sigma_2=\begin{bmatrix}0 & -i\\ i& 0\end{bmatrix}
$$
That is, the space $\mathcal{M}$ can be identified with the following affine cubic in ${\Bbb C}^3$,
$$
\mathcal{M} = \left\{s \equiv (s_1, s_2, s_3) \in {\Bbb C}^3: s_1 -s_2 +s_3 + s_1s_2s_3 = 0\right\}.
$$
The remarkable fact of the modern theory of Painlev\'e equations is that the Stokes parameters $s_j\equiv s_j(q,p,t)$ are the first integrals of the 
second Painlev\'e equation \cite{FN},
\begin{equation}\label{P22}
\frac{\d^2q}{\d t^2}=tq+2q^3,\quad p=2\frac{\d q}{\d t}.
\end{equation}
Moreover, in terms of these integrals, the Ablowitz-Segur  solution of (\ref{P22}) which we need in our study of the incomplete Wishart ensemble is
characterized (see, e.g. \cite{FIKN}) by the equations,
$$
s_1 = -s_3 = -s_4 = s_6=  -\im\sqrt{\gamma}, \quad s_2 = s_5=0.
$$
This also means that, in the case of the Ablowitz-Segur family (\ref{JME:18}) for Painlev\'e-II, 
the space of monodromy data reduces to the complex plane ${\Bbb C}$,
$$
\mathcal{M} = \left\{\gamma \in {\Bbb C} \right\}.
$$
Another way to describe the relation of the linear system \eqref{JME:10} to the Painlev\'e equation (\ref{P22}) is to say that the
latter describes the {\it isomonodromic deformations} of the former. In fact, the dynamical system  (\ref{P22}) is equivalent to the
differential matrix equation,
\begin{equation}\label{JME:13}
	\frac{\partial{\bf A}}{\partial t}-\frac{\partial{\bf U}}{\partial\lambda}=[{\bf U},{\bf A}].
\end{equation}
where
$$
{\bf U}(\lambda, t) = -\im\lambda\sigma_3+\im\begin{bmatrix} 0 & q\\ -q & 0\end{bmatrix}.
$$
The nonlinear matrix equation (\ref{JME:13}) is usually called a zero curvature, or Lax equation, and it is a compatibility condition of two linear equations - system \eqref{JME:10} and the $t$-differential equation
\begin{equation}\label{JME:12}
	\frac{\partial {\bf X}}{\partial t}= {\bf U}(\lambda,t){\bf X}.
\end{equation}
The pair of linear systems,  \eqref{JME:10} and  \eqref{JME:12} constitutes a {\it Lax pair} of the second Painlev\'e equation which was discovered 
by Flashcka and Newell in 1980 \cite{FN}. We are now passing to the isomonodromic tau function associated  with this Lax pair.\smallskip

The notion of isomonodromic tau functions was introduced  by Jimbo, Miwa, and Ueno in 1980 in \cite{JMU} for an arbitrary
system of linear ordinary differential equations with rational coefficients. Their theory is based on a special $1$-form  defined on the 
space of the parameters of the system which is closed on the trajectories of the corresponding isomonodromy deformation equations. 
In the case of system \eqref{JME:10}, the Jimbo-Miwa-Ueno $1$-form is defined by the equation (see \cite{JMU}, equation (5.1))
\begin{equation}\label{JME:14}
	\omega_{\textnormal{JMU}}=-\res_{\lambda=\infty}\textnormal{Tr}\left\{\big(\widehat{{\bf X}}(\lambda)\big)^{-1}\partial_{\lambda}\widehat{{\bf X}}(\lambda)\d_{t}\Theta(\lambda)\right\}
\end{equation}
where
$$
\ \ \ \ \ \ \ \ \ \Theta(\lambda)=-\im\left(\frac{4}{3}\lambda^3+t\lambda\right)\sigma_3,\quad \mbox{and}\quad \d_tf := \frac{\partial f}{\partial t} \d t.
$$
Using \eqref{Xhat} and (\ref{X1}) one can easily transform (\ref{JME:14}) into
\begin{equation}\label{JMU:141}
	\omega_{\textnormal{JMU}}=\left(\frac{p^2}{4}-tq^2-q^4\right)\,\d t\equiv\mathcal{H}_S(q,p,t)\,\d t.
\end{equation}
Denote
$$
\omega_{\textnormal{JMU}}(t; \gamma) \equiv \mathcal{H}_S(q(t; \gamma),p(t; \gamma),t)\,\d t
$$
as the restriction of the form $\omega_{\textnormal{JMU}}$ on the Ablowits-Segur solution of the Painev\'e II equation (\ref{P22}). 
The tau function corresponding to the Ablowits-Segur solution of the Painev\'e II equation (\ref{P22}) is then defined by the relation,
\begin{equation}\label{Tau:1}
	\d_{t}\ln\tau=\omega_{\textnormal{JMU}}(t; \gamma).
\end{equation}
Comparing this with  the equations stated in Theorem \ref{theo:1} , we see that the soft edge distribution function $F_{S}(t; \gamma)$ can be identified with
the isomonodromic tau function corresponding to the Ablowitz-Segur Painlev\'e-II transcendent,
\begin{equation}\label{tayFS}
\tau(t;\gamma) \equiv F_{S}(t; \gamma).
\end{equation}
 In the next subsection we show how one can use \eqref{tayFS} in the derivation of \eqref{JME:25} in Theorem \ref{theo:2}. 
\subsection{Extended Jimbo-Miwa-Ueno differential}\label{sec:15} 
In \cite{JMU} it is also shown that the form $\omega_{\textnormal{JMU}}$ can be alternatively defined as
$$
\omega_{\textnormal{JMU}} =\res_{\lambda=\infty}\textnormal{Tr}
\left\{{\bf A}(\lambda)\big(\d_{t}\widehat{{\bf X}}(\lambda)\big)\big(\widehat{{\bf X}}(\lambda)\big)^{-1}\right\}.
$$
Following Section 4.2 of \cite{ILP}, where the generic two parameter 
family of the solutions of the second Painlev\'e equation is studied, we use this alternative definition and pass from the Jimbo-Miwa-Ueno form 
$\omega_{\textnormal{JMU}}(t; \gamma)$ to the following $1$-form,
\begin{equation}\label{extform}
	\omega_{\textnormal{ext}}=\res_{\lambda=\infty}\textnormal{Tr}\left\{{\bf A}(\lambda)\big(\d\widehat{{\bf X}}(\lambda)\big)\big(\widehat{{\bf X}}(\lambda)\big)^{-1}\right\},\ \ \ \d=\d_{t}+\d_{\gamma},
\end{equation}
 defined on the extended space, $\{t\} \times \{\gamma\}$. Similar to the derivation of \eqref{JMU:141}, we can substitute formula \eqref{Xhat} for $\widehat{{\bf X}}(\lambda)$ into equation (\ref{extform}) and compute $\omega_{\textnormal{ext}}$ explicitly  in terms of $p$ and $q$ (cf. \cite{ILP}, (4.39)),
\begin{equation}\label{JME:29}
	\omega_{\textnormal{ext}}=\mathcal{H}_S(q,p,t)\,\d t+\frac{2}{3}\left(pq_{\gamma}-\frac{1}{2}p_{\gamma}q-2tq_{\gamma}\big(tq+2q^3\big)+\frac{t}{2}pp_{\gamma}\right)\,\d\gamma.
\end{equation}
It should be noticed though that in order to arrive at this formula we now need, in addition to (\ref{X1}), the exact expressions for the matrix coefficients
${\bf X}_2$ and  ${\bf X}_3$ which can be found in \cite{ILP} - see equation $(4.38)$. Two important facts about the form  $\omega_{\textnormal{ext}}$ can be extracted from \eqref{JME:29}:
\begin{itemize}
\item On the trajectories of the second Painlev\'e equation the form $\omega_{\textnormal{ext}}$ coincides with the Jimbo-Miwa-Ueno
form $\omega_{\textnormal{JMU}}$, i.e.,
\begin{equation}\label{extprop10}
\omega_{\textnormal{ext}}(t; \gamma = \textnormal{const.}) = \omega_{\textnormal{JMU}}(t; \gamma)  \equiv d_{t}\ln F_{S}.
\end{equation}
\item The form  $\omega_{\textnormal{ext}}$ differs from the classical action differential, $p\d q-\mathcal{H}_S\,\d t$,
by a total differential. Indeed, one can check by a direct differentiation that 
\begin{equation}\label{extprop20}
	\omega_{\textnormal{ext}}=\frac{1}{3}\d\big(2t\,\mathcal{H}_S-pq\big)+p\d q-\mathcal{H}_S\,\d t.
\end{equation}
\end{itemize}
 Restricting equation (\ref{extprop20}) to the Ablowitz-Segur trajectory $q=q(t;\gamma)$, $p=p(t;\gamma)$, $\gamma \equiv$ const. and
 taking into account equation (\ref{extprop10}) we arrive at the differential version of \eqref{JME:25}.
 The remaining two action formul\ae\, \eqref{JME:24} and \eqref{JME:26} can be derived in a similar way using, instead of the Lax
pair \eqref{JME:10}, \eqref{JME:12}, the Lax pairs corresponding to the dynamical systems \eqref{HH:1} and \eqref{HH:2}, respectively.\smallskip

We complete our presentation of this alternative proof of Theorem \ref{theo:2} by  showing that
the transformation of equation  (\ref{JME:29}) into equation (\ref{extprop20}) is not an accident. In fact, there is a deep reason
why this transformation takes place. To this end, let us  consider the form  $\omega_{\textnormal{ext}}$ on the whole extended monodromy space, $\{t\}\times \mathcal{M}$, i.e. we pass from the one parameter Ablowitz-Segur family of solutions to Painlev\'e-II to the general two parameter set of solutions, 
$$
q \equiv q(t; s_1,s_2), \quad p \equiv p(t; s_1, s_2),
$$
(we chose $s_1$ and $s_2$ as the local coordinates on $\mathcal{M}$). This means, that the  differentiation $\d$ in \eqref{extform} now means
$\d = \d_t + \d_{s_{1}} + \d_{s{_2}}$ and equation \eqref{JME:29} is replaced by the whole equation  (4.39) of \cite{ILP},
$$
\omega_{\textnormal{ext}}=\mathcal{H}_S(q,p,t)\,\d t+\frac{2}{3}\left(pq_{s_{1}}-\frac{1}{2}p_{s_{1}}q-2tq_{s_{1}}\big(tq+2q^3\big)+\frac{t}{2}pp_{s_{1}}\right)\,\d s_1
$$
\begin{equation}\label{JME:291}
+ \frac{2}{3}\left(pq_{s_{2}}-\frac{1}{2}p_{s_{2}}q-2tq_{s_{2}}\big(tq+2q^3\big)+\frac{t}{2}pp_{s_{2}}\right)\,\d s_{2}.
\end{equation}
The general key fact about the extended form $\omega_{\textnormal{ext}}$ is that its external derivative is a 2-form on  $\mathcal{M}$
and it does not depend on $t$. In fact, one can check directly that (cf. (4.48) in \cite{ILP})
$$
\d\omega_{\textnormal{ext}} = \bigl(p_{s_1}q_{s_2} - p_{s_2}q_{s_1} \bigr)\d s_{1}\wedge\d s_{2} \equiv \Omega,
$$
where $\Omega$ is the canonical symplectic form on the phase space $\{(p,q)\}$. A classical fact of Hamiltonian mechanics
is that the external derivative of the classical action differential equals the same symplectic form,
$$
\d\bigl(p\d q - \mathcal{H}_S\,\d t\bigr) = \bigl(p_{s_1}q_{s_2} - p_{s_2}q_{s_1} \bigr)\d s_{1}\wedge\d s_{2} \equiv \Omega
$$
Therefore, 
$$
\omega_{\textnormal{ext}} - \bigl(p\d q - \mathcal{H}_S\bigr)\,\d t = \mbox{total differential}
$$
The fact that this total differential equals $\frac{1}{3}\d\big(2t\,\mathcal{H}_S-pq\big)$ is the result of a concrete calculation. We do not yet have a conceptual way to find this differential.

  
\section{Proof of Theorem \ref{theo:3}, expansion \eqref{JME:32}}\label{sec:3}
It is well known, cf. \cite{FIKN}, that we can characterize the functions $(q,p)$ in \eqref{JME:8}, \eqref{JME:18} through the solution of the following Riemann-Hilbert problem (RHP)
\begin{problem}\label{master} Let $t\in\mathbb{R}$ and $\gamma\in[0,1]$. Determine the piecewise analytic function ${\bf Y}={\bf Y}(\lambda;t,\gamma)\in\mathbb{C}^{2\times 2}$ such that
\begin{enumerate}
	\item[(1)] $Y(\lambda)$ is analytic for $\lambda\in\mathbb{C}\backslash(\Gamma_1\cup\Gamma_3\cup\Gamma_4\cup\Gamma_6)$ with the four rays
	\begin{equation*}
		\Gamma_1=\left\{\lambda\in\mathbb{C}:\ \ \textnormal{arg}\,\lambda=\frac{\pi}{6}\right\},\ \ \ \Gamma_3=\left\{\lambda\in\mathbb{C}:\ \ \textnormal{arg}\,\lambda=\frac{5\pi}{6}\right\}
	\end{equation*}
	\begin{equation*}
		\Gamma_4=\left\{\lambda\in\mathbb{C}:\ \ \textnormal{arg}\,\lambda=\frac{7\pi}{6}\right\},\ \ \ \Gamma_6=\left\{\lambda\in\mathbb{C}:\ \ \textnormal{arg}\,\lambda=\frac{11\pi}{6}\right\}
	\end{equation*}
	oriented from the origin $\lambda=0$ towards infinity, compare Figure \ref{figure1} below.
	\begin{figure}[tbh]
\begin{center}
\resizebox{0.6\textwidth}{!}{\includegraphics{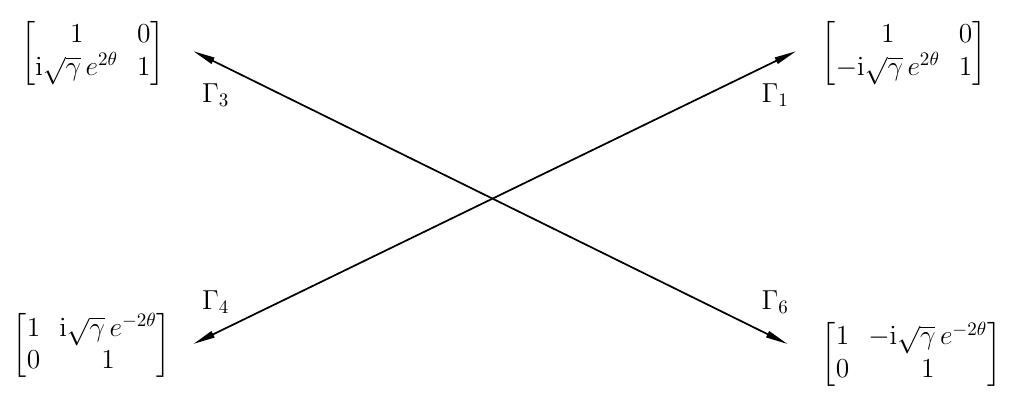}}
\caption{The oriented jump contours for the master function ${\bf Y}(\lambda;t,\gamma)$ of RHP \ref{master} in the complex $\lambda$-plane.}
\label{figure1}
\end{center}
\end{figure}
	\item[(2)] The boundary values ${\bf Y}_+(\lambda)$ (or ${\bf Y}_-(\lambda)$) from the left (or right) side of the oriented contour $\Gamma_k$ satisfy the jump relation
	\begin{equation*}
		{\bf Y}_+(\lambda)={\bf Y}_-(\lambda)\e^{-\theta(\lambda,t)\sigma_3}{\bf S}_k\e^{\theta(\lambda,t)\sigma_3},\ \ \lambda\in\Gamma_k,\ \ \ k=1,2,3,4
	\end{equation*}
	with
	\begin{equation*}
		\theta(\lambda,t)=\im\left(\frac{4}{3}\lambda^3+t\lambda\right),\ \ \ \ \ \sigma_3=\bigl[\begin{smallmatrix}1&0\\0&-1\end{smallmatrix}\bigr];
	\end{equation*}
	and the $\lambda$-independent matrices
	\begin{equation*}
		{\bf S}_1=\bigl[\begin{smallmatrix}1&0\\ -\im\sqrt{\gamma}&1\end{smallmatrix}\bigr],\ \ \ {\bf S}_3=\bigl[\begin{smallmatrix}1&0\\ \im\sqrt{\gamma}&1\end{smallmatrix}\bigr],\ \ \ {\bf S}_4=\bigl[\begin{smallmatrix}1&\im\sqrt{\gamma}\\ 0&1\end{smallmatrix}\bigr]\ \ \ {\bf S}_6=\bigl[\begin{smallmatrix}1&-\im\sqrt{\gamma}\\ 0&1\end{smallmatrix}\bigr].
	\end{equation*}
	\item[(3)] As $\lambda\rightarrow\infty$, ${\bf Y}(\lambda)$ is normalized in the following way
	\begin{equation*}
		{\bf Y}(\lambda)=\mathbb{I}+{\bf Y}_1\lambda^{-1}+{\bf Y}_2\lambda^{-2}+\mathcal{O}\left(\lambda^{-3}\right),\ \ \ \ {\bf Y}_{\ell}=\big(Y_{\ell}^{jk}\big)_{j,k=1}^2
	\end{equation*}
\end{enumerate}
\end{problem}
As proven in \cite{BIK}, the latter problem for ${\bf Y}(\lambda)$ is uniquely solvable for all $t\in\mathbb{R},\gamma\in[0,1]$ and its solution determines the Ablowitz-Segur transcendents via
\begin{equation}\label{e:3}
	q(t;\gamma)=2Y_1^{12},\ \ \ p(t;\gamma)=2q_t(t;\gamma)=-8\im\big(Y_2^{12}+Y_1^{12}Y_1^{11}\big);
\end{equation}
Moreover the Hamiltonian function $\mathcal{H}_S=\mathcal{H}_S(q,p,t)$ can be read off directly from RHP \ref{master} as well,
\begin{equation}\label{e:4}
	\mathcal{H}_S\big(q(t;\gamma),p(t;\gamma),t\big)=2\im Y_1^{11}.
\end{equation}
The Riemann-Hilbert representation \eqref{e:3} has been used numerous times in the literature to derive the leading asymptotic behavior of $q(t;\gamma)$ as $t\rightarrow\pm\infty$ and $\gamma\in[0,1]$ is kept fixed, cf. \cite{FIKN} for more on the history of this subject. For our purposes (i.e. the proof of Theorem \ref{theo:3}, expansion \eqref{JME:32}) the estimates given in \cite{FIKN} have to be slightly extended. With this goal in mind we shall not reproduce all steps carried out in \cite{FIKN}, instead we only provide references and jump immediately to the key estimates.
\subsection{Nonlinear steepest descent analysis for RHP \ref{master} (in a nutshell)} Our goal is to solve RHP \ref{master} for ${\bf Y}(\lambda;t,\gamma)\in\mathbb{C}^{2\times 2}$ for all values $(-t,v)\in\mathbb{R}_+^2$ such that
\begin{equation}\label{s:1}
	(-t)\geq t_0,\ \ \ \textnormal{and}\ \  0\leq v=-\ln(1-\gamma)<+\infty\ \  \textnormal{is fixed}.
\end{equation}
This is achieved by first rescaling the initial function with the large parameter, ${\bf X}(\lambda)={\bf Y}(\lambda\sqrt{-t}\,),\lambda\in\mathbb{C}\backslash(\Gamma_1\cup\Gamma_3\cup\Gamma_4\cup\Gamma_6)$. Secondly, contour deformations ${\bf X}(\lambda)\mapsto{\bf T}(\lambda)$, see \cite{FIKN}, Figure $9.4$ and thirdly, matrix factorizations and opening of lens transformations ${\bf T}(\lambda)\mapsto {\bf S}(\lambda)$, see \cite{FIKN}, (9.4.7) and Figures $9.5,9.6$. After those initial three transformations the RHP for ${\bf S}(\lambda)$ is already in a localized state since the underlying jump matrix ${\bf G}_{\bf S}(\lambda;t,v)$ obeys (see \cite{FIKN}, (9.4.30))
\begin{equation}\label{es:2}
	{\bf G}_{\bf S}(\lambda;t,v)=\mathbb{I}+\mathcal{O}\left(\e^{v-4(-t)^{\frac{3}{2}}|\lambda\mp\frac{1}{2}|^2}\right),\ \ \ \ (-t)\geq t_0,
\end{equation}
for $\lambda$ along the eight contours shown in \cite{FIKN}, Figure $9.6$ that extend to infinity. Hence one needs to focus only on the line segment $(-\frac{1}{2},\frac{1}{2})\subset\mathbb{R}$ and two small vicinities of the endpoints $\pm\frac{1}{2}$. But all parametrices are well known, e.g. for the segment (see \cite{FIKN}, (9.4.8)) we take
\begin{equation*}
	{\bf P}^{(\infty)}(\lambda)=\left(\frac{\lambda+\frac{1}{2}}{\lambda-\frac{1}{2}}\right)^{\nu\sigma_3},\ \ \ \ \lambda\in\mathbb{C}\Big\backslash\left[-\frac{1}{2},\frac{1}{2}\right];\ \ \ \ \nu=\frac{v}{2\pi\im}\in\im\mathbb{R},
\end{equation*}
and for the neighborhoods of $\lambda=\pm\frac{1}{2}$ standard parabolic cylinder functions \cite{NIST} come into play. We shall denote those parametrices by ${\bf P}^{(\frac{1}{2})}(\lambda)$, see \cite{FIKN}, (9.4.20) and ${\bf P}^{(-\frac{1}{2})}(\lambda)$, compare \cite{FIKN}, (9.4.24). The three explicit model functions ${\bf P}^{(\infty)}(\lambda),{\bf P}^{(\frac{1}{2})}(\lambda)$ and ${\bf P}^{(-\frac{1}{2})}(\lambda)$ are then compared locally to the unknown ${\bf S}(\lambda)$,
\begin{equation}\label{e:14}
	{\bf R}(\lambda)={\bf S}(\lambda)\begin{cases}\big({\bf P}^{(\frac{1}{2})}(\lambda)\big)^{-1},&\lambda\in\mathbb{D}_r(\frac{1}{2})\\ \big({\bf P}^{(-\frac{1}{2})}(\lambda)\big)^{-1},&\lambda\in\mathbb{D}_r(-\frac{1}{2})\\
	\big({\bf P}^{(\infty)}(\lambda)\big)^{-1},&\lambda\notin\big(\mathbb{D}_r(\frac{1}{2})\cup\mathbb{D}_r(-\frac{1}{2})\big)\end{cases},\ \ \ \ \ \mathbb{D}_r(\lambda_0)=\{\lambda\in\mathbb{C}:\ |\lambda-\lambda_0|<r\},
\end{equation}
with fixed radius $0<r<\frac{1}{8}$. Recalling the model function properties we obtain the following RHP for the ratio function ${\bf R}(\lambda)$.
\begin{problem}\label{ratio} Find ${\bf R}(\lambda)={\bf R}(\lambda;t,v)\in\mathbb{C}^{2\times 2}$ with $(-t,v)\in\mathbb{R}^2_+$ such that
\begin{enumerate}
	\item[(1)] ${\bf R}(\lambda)$ is analytic for $\lambda\in\mathbb{C}\backslash\Sigma_{\bf R}$ where $\Sigma_{\bf R}=\partial\mathbb{D}_r(-\frac{1}{2})\cup\partial\mathbb{D}_r(\frac{1}{2})\cup\Sigma_{\infty}$ is shown in Figure \ref{figure7} below.
	\begin{figure}[tbh]
\begin{center}
\resizebox{0.5\textwidth}{!}{\includegraphics{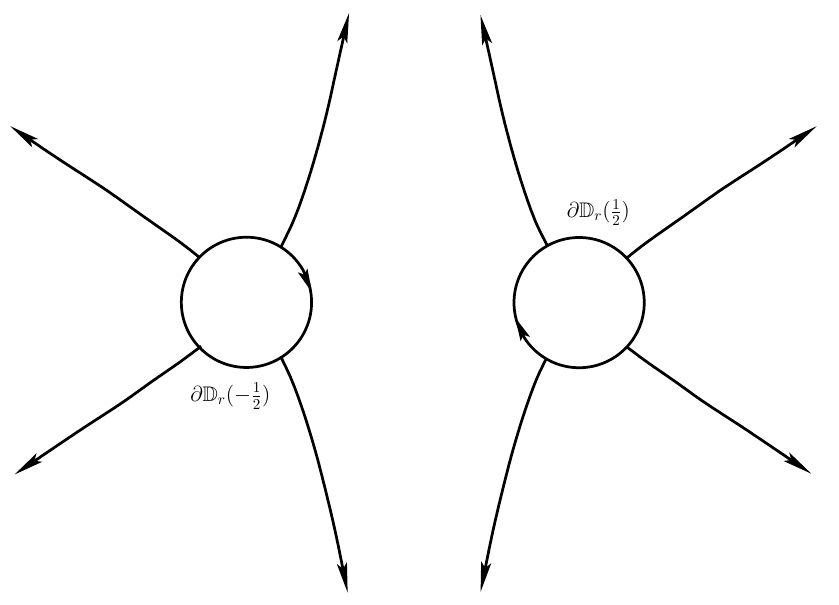}}
\caption{The oriented jump contours for the ratio function ${\bf R}(\lambda)$ in the complex $\lambda$-plane. The eight contours extending to infinity are summarized as $\Sigma_{\infty}$.}
\label{figure7}
\end{center}
\end{figure}
	\item[(2)] Along the contour $\Sigma_{\bf R}$ we have jumps ${\bf R}_+(\lambda)={\bf R}_-(\lambda){\bf G}_{\bf R}(\lambda;t,\gamma),\lambda\in\Sigma_{\bf R}$ with
	\begin{equation*}
		{\bf G}_{\bf R}(\lambda;t,v)={\bf P}^{(\infty)}(\lambda){\bf G}_{\bf S}(\lambda;t,v)\big({\bf P}^{(\infty)}(\lambda)\big)^{-1},\ \ \lambda\in\Sigma_{\infty}
	\end{equation*}
	and
	\begin{equation*}
		{\bf G}_{\bf R}(\lambda;t,v)={\bf P}^{(\pm\frac{1}{2})}(\lambda)\big({\bf P}^{(\infty)}(\lambda)\big)^{-1},\ \ \lambda\in\partial\mathbb{D}_r\left(\pm\frac{1}{2}\right).
	\end{equation*}
	By construction, there are no jumps in the interior of $\mathbb{D}_r(\pm\frac{1}{2})$ and along $[-\frac{1}{2},\frac{1}{2}]$.
	\item[(3)] As $\lambda\rightarrow\infty$, 
	\begin{equation*}
		{\bf R}(\lambda)=\mathbb{I}+\mathcal{O}\left(\lambda^{-1}\right).
	\end{equation*}
\end{enumerate}
\end{problem}
We now see how the constraint \eqref{s:1} guarantees that all jump matrices in RHP \ref{ratio} are close to the identity in the same scaling regime. First turn towards $\Sigma_{\infty}$: from \eqref{es:2} and the fact that $\nu\in\im\mathbb{R}$ we obtain at once,
\begin{prop}\label{DZ:1} There exist constants $t_0>0$ and $c>0$ such that
\begin{equation*}
	\|{\bf G}_{\bf R}(\cdot;t,v)-\mathbb{I}\|_{L^2\cap L^{\infty}(\Sigma_{\infty})}\leq c\,\e^{2v-4(-t)^{\frac{3}{2}}r^2},\ \ \ \ \ \forall\, (-t)\geq t_0,\ \ v\geq 0.
\end{equation*}
The parameter $0<r<\frac{1}{8}$ has been introduced previously in \eqref{e:14}.
\end{prop}
Second, for $\partial\mathbb{D}_r(\pm\frac{1}{2})$ we recall \cite{FIKN}, (9.4.23) and (9.4.33), 
\begin{prop}\label{DZ:11} For any fixed $v\in[0,+\infty)$ there exist positive constants $t_0=t_0(v)$ and $c=c(v)$ such that
\begin{equation*}
	\|{\bf G}_{\bf R}(\cdot;t,v)-\mathbb{I}\|_{L^2\cap L^{\infty}(\partial\mathbb{D}_r(\pm\frac{1}{2}))}\leq c\,(-t)^{-\frac{3}{4}},\ \ \ \ \forall\, (-t)\geq t_0.
\end{equation*}
\end{prop}
By general theory, cf. \cite{DZ}, the last two estimates ensures unique solvability of the ratio problem \ref{ratio} in the scaling regime \eqref{s:1}, in fact 
\begin{theo}\label{DZend} For any fixed $v\in[0,+\infty)$ there exist $t_0=t_0(v)>0$ and $c=c(v)>0$ such that the ratio problem \ref{ratio} is uniquely solvable in $L^2(\Sigma_{\bf R})$ for all $(-t)\geq t_0$. We can compute its solution iteratively from the integral equation
\begin{equation*}
	{\bf R}(\lambda)=\mathbb{I}+\frac{1}{2\pi\im}\int_{\Sigma_{\bf R}}{\bf R}_-(w)\big({\bf G}_{\bf R}(w)-\mathbb{I}\big)\frac{\d w}{w-\lambda},\ \ \ \lambda\in\mathbb{C}\backslash\Sigma_{\bf R},
\end{equation*}
using that
\begin{equation*}
	\|{\bf R}_-(\cdot;s,v)-\mathbb{I}\|_{L^2(\Sigma_{\bf R})}\leq c\,(-t)^{-\frac{3}{4}},\ \ \forall\, (-t)\geq t_0.
\end{equation*}
\end{theo}
At this point we can extract all relevant asymptotic information via \eqref{e:3} and \eqref{e:4}.
\subsection{Extraction of asymptotics and proof of expansion \eqref{JME:32}} Tracing back all explicit and invertible transformations, i.e. the sequence
\begin{equation*}
	{\bf Y}(\lambda)\mapsto {\bf X}(\lambda)\mapsto {\bf T}(\lambda)\mapsto {\bf S}(\lambda)\mapsto {\bf R}(\lambda),
\end{equation*}
we obtain the following formul\ae,
\begin{equation*}
	{\bf Y}_1=\sqrt{-t}\left(\nu\sigma_3+\frac{\im}{2\pi}\int_{\Sigma_{\bf R}}{\bf R}_-(w)\big({\bf G}_{\bf R}(w)-\mathbb{I}\big)\d w\right),
\end{equation*}
and
\begin{equation*}
	{\bf Y}_2=-t\left(\frac{\nu^2}{2}\mathbb{I}+\frac{\im\nu}{2\pi}\int_{\Sigma_{\bf R}}{\bf R}_-(w)\big({\bf G}_{\bf R}(w)-\mathbb{I}\big)\d w\,\sigma_3+\frac{\im}{2\pi}\int_{\Sigma_{\bf R}}{\bf R}_-(w)\big({\bf G}_{\bf R}(w)-\mathbb{I}\big)w\,\d w\right).
\end{equation*}
We begin with the asymptotic estimation of the integrals
\begin{equation*}
	{\bf J}=\frac{\im}{2\pi}\int_{\Sigma_{\bf R}}\big({\bf G}_{\bf R}(w)-\mathbb{I}\big)\,\d w,\ \ {\bf J}=\big(J^{jk}\big)_{j,k=1}^2;\hspace{1cm}{\bf K}=\frac{\im}{2\pi}\int_{\Sigma_{\bf R}}\big({\bf G}_{\bf R}(w)-\mathbb{I}\big)w\,\d w,\ \ \ \ \ {\bf K}=\big(K^{jk}\big)_{j,k=1}^2
\end{equation*}
using \cite{FIKN}, (9.4.31), (9.4.32) and the residue theorem.
\begin{lem}
As $t\rightarrow-\infty$, with $s=(-t)^{\frac{3}{2}}$,
\begin{equation*}
	J^{11}=-\frac{\im\nu^2}{12s}+\mathcal{O}\left(s^{-2}\right),\ \ \ J^{12}=\frac{1}{2}\sqrt{\frac{v}{\pi s}}\cos\big(\phi(s,v)\big)+\mathcal{O}\big(s^{-\frac{3}{2}}\big),
\end{equation*}
and
\begin{equation*}
	K^{12}=\frac{\im}{4}\sqrt{\frac{v}{\pi s}}\,\sin\big(\phi(s,v)\big)+\mathcal{O}\big(s^{-\frac{3}{2}}\big);\hspace{1cm}\phi(s,v)=\frac{2}{3}s-\frac{v}{2\pi}\ln(8s)+\frac{\pi}{4}-\textnormal{arg}\,\Gamma\left(\frac{v}{2\pi\im}\right).
\end{equation*}
All error terms are uniform with respect to $v$ chosen from compact subsets of $[0,+\infty)$.
\end{lem}
Next we obtain from Theorem \ref{DZend} that for $w\in\Sigma_R$,
\begin{equation*}
	{\bf R}_-(w)-\mathbb{I}=\frac{1}{2\pi\im}\int_{\Sigma_{\bf R}}\big({\bf G}_{\bf R}(\mu)-\mathbb{I}\big)\frac{\d\mu}{\mu-w_-}+\mathcal{O}\big((-t)^{-\frac{3}{2}}\big),
\end{equation*}
and thus, iterating once, where
\begin{equation*}
	{\bf L}=\frac{\im}{2\pi}\int_{\Sigma_{\bf R}}\big({\bf R}_-(w)-\mathbb{I}\big)\big({\bf G}_{\bf R}(w)-\mathbb{I}\big)\,\d w,\ \ \ \ {\bf L}=\big(L^{jk}\big)_{j,k=1}^2,
\end{equation*}
\begin{lem} As $t\rightarrow-\infty$ with $s=(-t)^{\frac{3}{2}}$,
\begin{equation*}
	L^{11}=-\frac{2\im\nu^2}{3s}+\frac{\nu}{4s}\sin\big(2\phi(s,v)\big)+\mathcal{O}\big(s^{-\frac{3}{2}}\big),
\end{equation*}
and the error term is uniform with respect to $v$ chosen from compact subsets of $[0,+\infty)$.
\end{lem}
Now we go back to \eqref{e:3} and \eqref{e:4},
\begin{cor} As $t\rightarrow-\infty$, with $s=(-t)^{\frac{3}{2}}$ and fixed $v\in[0,+\infty)$,
\begin{equation}\label{cor:1}
	q(t;\gamma)=(-t)^{-\frac{1}{4}}\sqrt{\frac{v}{\pi}}\,\cos\big(\phi(s,v)\big)+\mathcal{O}\left((-t)^{-\frac{7}{4}}\right),\ \ \ p(t;\gamma)=2(-t)^{\frac{1}{4}}\sqrt{\frac{v}{\pi}}\,\sin\big(\phi(s,v)\big)+\mathcal{O}\big((-t)^{-\frac{1}{2}}\big),
\end{equation}
and
\begin{equation*}
	\mathcal{H}_S\big(q(t;\gamma),p(t;\gamma),t\big)=\frac{v}{\pi}\sqrt{-t}+\frac{3v^2}{8\pi^2 t}-\frac{v}{4\pi t}\sin\big(2\phi(s,v)\big)+\mathcal{O}\big((-t)^{-\frac{7}{4}}\big).
\end{equation*}
\end{cor}
The last result allows us to determine all $t$-dependent leading terms, compare Proposition \ref{theo:1}
\begin{cor} As $t\rightarrow-\infty$,
\begin{equation*}
	\ln F_S(t;\gamma)=-\frac{2v}{3\pi}|t|^{\frac{3}{2}}+\frac{v^2}{4\pi^2}\ln\big(|t|^{\frac{3}{2}}\big)+D(v)+\mathcal{O}\big(|t|^{-\frac{3}{4}}\big),
\end{equation*}
where $D(v)$ is $t$-independent and the error is uniform with respect to $v$ chosen from compact subset of $[0,+\infty)$.
\end{cor}
As for $D(v)$, we now use \eqref{JME:25} and Corollary \ref{logic}. First, as $t\rightarrow-\infty$,
\begin{equation*}
	\frac{1}{3}(2t\,\mathcal{H}_S-pq)=-\frac{2v}{3\pi}|t|^{\frac{3}{2}}+\frac{v^2}{4\pi^2}-\frac{v}{2\pi}\sin\big(2\phi(s,v)\big)+\mathcal{O}\big(|t|^{-\frac{3}{4}}\big).
\end{equation*}
And second, with \eqref{cor:1},
\begin{equation*}
	\frac{\partial I_S}{\partial\gamma}=pq_{\gamma}=\frac{\d}{\d\gamma}\left(\frac{v}{2\pi}\sin\big(2\phi(s,v)\big)+\frac{v^2}{4\pi^2}\ln(8s)\right)-\frac{v}{\pi}\frac{\d}{\d\gamma}\textnormal{arg}\,\Gamma\left(\frac{\im v}{2\pi}\right)+\mathcal{O}\left(|t|^{-\frac{3}{2}}\ln|t|\right)
\end{equation*}
so that all together (since $F_S(t;0)=1$),
\begin{prop} As $t\rightarrow-\infty$,
\begin{equation*}
	\ln F_S(t;\gamma)=-\frac{2v}{3\pi}|t|^{\frac{3}{2}}+\frac{v^2}{4\pi^2}\ln\big(8|t|^{\frac{3}{2}}\big)+\frac{v^2}{4\pi^2}-\frac{1}{\pi}\int_0^{\gamma}v(\gamma')\frac{\d}{\d\gamma'}\textnormal{arg}\,\Gamma\left(\frac{\im v(\gamma')}{2\pi}\right)\,\d\gamma'+\mathcal{O}\big(|t|^{-\frac{3}{4}}\big)
\end{equation*}
uniformly for $\gamma\in[0,1)$ chosen from compact subsets.
\end{prop}
We now only have to recall the following standard property of the Barnes $G$-function, cf. \cite{NIST} ,
\begin{equation}\label{Barnesid}
	\int_0^z\ln\Gamma(1+t)\,\d t=\frac{z}{2}\ln(2\pi)-\frac{z}{2}(z+1)+z\ln\Gamma(1+z)-\ln G(1+z),\ \ \ \ z\in\mathbb{C}:\ \ \Re z>-1
\end{equation}
and Theorem \ref{theo:3}, expansion \eqref{JME:32} follows at once.
\section{Proof of Theorem \ref{theo:3}, expansion \eqref{JME:31}}\label{sec:4}
It is known from \cite{DIZ}, Section $4$ that we can characterize the functions $(q,p)$ in \eqref{HH:1}, \eqref{JME:17} through the solution of the following Riemann-Hilbert problem.
\begin{problem}\label{sineRHP} Let $t\in\mathbb{R}_{\geq 0}$ and $\gamma\in[0,1]$. Determine the piecewise analytic function ${\bf Y}={\bf Y}(\lambda;t,\gamma)\in\mathbb{C}^{2\times 2}$ such that
\begin{enumerate}
	\item[(1)] ${\bf Y}(\lambda)$ is analytic for $\lambda\in\mathbb{C}\backslash[-1,1]$ with the line segment $[-1,1]$ oriented from left to right as shown in Figure \ref{figure8} below.
	\item[(2)] The boundary values ${\bf Y}_+(\lambda)$ (or ${\bf Y}_-(\lambda)$) from the left (or right) side of the oriented contour $(-1,1)$ obey the jump relation
	\begin{equation}\label{sinejump}
		{\bf Y}_+(\lambda)={\bf Y}_-(\lambda)\e^{\im t\lambda\sigma_3}\begin{bmatrix}1-\gamma & \gamma\\ -\gamma & 1+\gamma\end{bmatrix}\e^{-\im t\lambda\sigma_3},\ \ \lambda\in(-1,1).
	\end{equation}
	\begin{figure}[tbh]
	\begin{center}
\resizebox{0.3\textwidth}{!}{\includegraphics{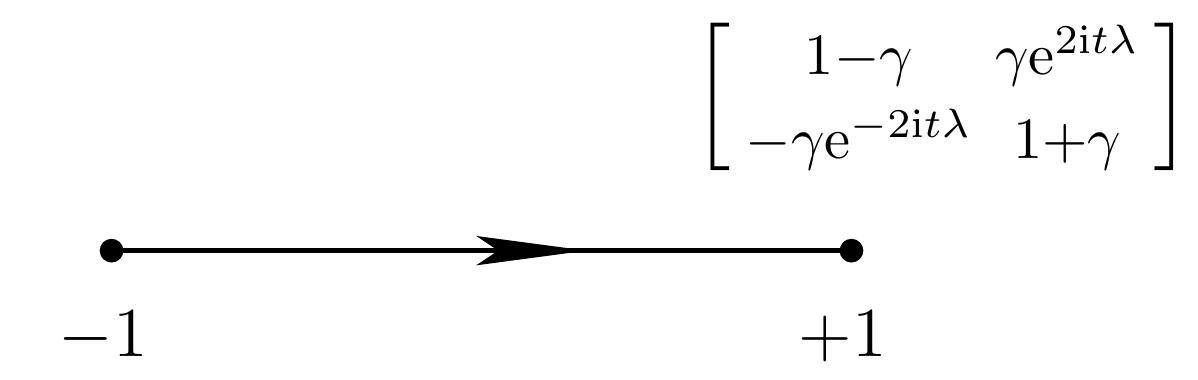}}
\caption{The oriented jump contours for the master function ${\bf Y}(\lambda;t,\gamma)$ of RHP \ref{sineRHP} in the complex $\lambda$-plane.}
\label{figure8}
\end{center}
\end{figure}
	\item[(3)] Near the endpoints $\lambda=\pm 1$, we have the singular behavior
	\begin{equation*}
		{\bf Y}(\lambda)=\widehat{{\bf Y}}(\lambda)\left\{I+\frac{\gamma}{2\pi\im}\begin{bmatrix}-1 & 1\\ -1 & 1\end{bmatrix}\ln\left(\frac{\lambda-1}{\lambda+1}\right)\right\}\e^{-\im t\lambda\sigma_3},\ \ \lambda\rightarrow\pm 1
	\end{equation*}
	where $\ln$ denotes the principal branch of the logarithm and $\widehat{{\bf Y}}(\lambda)$ is analytic at $\lambda=\pm 1$.
	\item[(4)] As $\lambda\rightarrow\infty$, ${\bf Y}(\lambda)$ is normalized as
	\begin{equation*}
		{\bf Y}(\lambda)=\mathbb{I}+{\bf Y}_1\lambda^{-1}+\mathcal{O}\left(\lambda^{-2}\right),\ \ \ {\bf Y}_{\ell}=\big(Y_{\ell}^{jk}\big)_{j,k=1}^2.
	\end{equation*}
\end{enumerate}
\end{problem}
As proven in \cite{IIKS,JMMS}, the latter problem for ${\bf Y}(\lambda)$ is uniquely solvable for all $t\geq 0,\gamma\in[0,1]$ and its solution determines the Jimbo-Miwa-Mori-Sato transcendents via
\begin{equation}\label{J:1}
	q(t;\gamma)=\frac{1}{2}Y_1^{11}+\im t\big(Y_1^{12}\big)^2,\ \ \ \ \ \sinh\left(\frac{1}{2}p(t;\gamma)\right)=-\frac{\im t Y_1^{12}}{q(t;\gamma)}.
\end{equation}
We have in addition for the Hamiltonian function
\begin{equation}\label{J:2}
	\mathcal{H}_B\big(q(t;\gamma),p(t;\gamma),t\big)=-2\im Y_1^{11},
\end{equation}
and equations \eqref{J:1}, \eqref{J:2} are the starting point for our asymptotic analysis. The Riemann-Hilbert problem \ref{sineRHP} was solved asymptotically for $t\rightarrow+\infty$ and $\gamma\in[0,1)$ (fixed and for certain moving values of $\gamma\uparrow 1$) in \cite{BDIK2}\footnote{The reference \cite{BDIK2} uses $v=-\frac{1}{2}\ln(1-\gamma)$ instead of $v=-\ln(1-\gamma)$. This has to be remembered in Subsection \ref{sec41}.}. Similar to Section \ref{sec:3} this allows us to save time and space. 
\subsection{Nonlinear steepest descent analysis for RHP \ref{sineRHP} (in a nutshell)}\label{sec41}
Our goal is to solve RHP \ref{sineRHP} for ${\bf Y}(\lambda;t,\gamma)\in\mathbb{C}^{2\times 2}$ for all values $(t,v)\in\mathbb{R}_+^2$ such that
\begin{equation}\label{sinescale}
	t\geq t_0,\ \ \ \ \textnormal{and}\ \ \ 0\leq v=-\ln(1-\gamma)<+\infty\ \ \ \textnormal{is fixed}.
\end{equation}
To achieve this we first use matrix factorizations and an opening of lens transformation, ${\bf Y}(\lambda)\mapsto {\bf S}(\lambda)$, see \cite{BDIK2}, Figure $1$ and RHP $2.2$. After this step the problem is already localized since the jump matrix ${\bf G}_{{\bf S}}(\lambda;t,v)$ obeys (see \cite{BDIK2}, page $218$ top)
\begin{equation}\label{es:3}
	{\bf G}_{{\bf S}}(\lambda;t,v)=\mathbb{I}+\mathcal{O}\left(\e^{v-2t|\Im\lambda|}\right),\ \ \ \ \ t\geq t_0
\end{equation}
for $\lambda$ on the contours in the upper and lower half-plane, see \cite{BDIK2}, Figure $1$. Hence we address the local problems on $(-1,1)\subset\mathbb{R}$ and in the vicinities of the endpoints $\pm 1$. The parametrices are again standard,
\begin{equation*}
	{\bf P}^{(\infty)}(\lambda)=\left(\frac{\lambda+1}{\lambda-1}\right)^{\nu\sigma_3},\ \ \lambda\in\mathbb{C}\backslash[-1,1];\ \ \ \ \nu=\frac{v}{2\pi\im}\in\im\mathbb{R}
\end{equation*}
is chosen for the segment (see \cite{BDIK2}, (2.1)) and near $\lambda=\pm 1$ confluent hypergeometric functions come into play. Let ${\bf P}^{(1)}(\lambda)$ and ${\bf P}^{(-1)}(\lambda)$ denote the required model functions, see \cite{BDIK2}, (2.4) and (2.6). These functions are unimodular and can be compared to the above ${\bf S}(\lambda)$,
\begin{equation*}
	{\bf R}(\lambda)={\bf S}(\lambda)\begin{cases}\big({\bf P}^{(1)}(\lambda)\big)^{-1},&\lambda\in\mathbb{D}_r(1)\\ \big({\bf P}^{(-1)}(\lambda)\big)^{-1},&\lambda\in\mathbb{D}_r(-1)\\ \big({\bf P}^{(\infty)}(\lambda)\big)^{-1},&\lambda\notin\big(\mathbb{D}_r(1)\cup\mathbb{D}_r(-1)\big)\end{cases},
\end{equation*}
where $0<r<\frac{1}{4}$ is kept fixed, see \cite{BDIK2}, (2.8). In turn we find the problem outlined below.
\begin{problem}\label{J:ratio} The function ${\bf R}(\lambda)={\bf R}(\lambda;t,v)\in\mathbb{C}^{2\times 2}$ has the following properties.
\begin{enumerate}
	\item[(1)] ${\bf R}(\lambda)$ is analytic for $\lambda\in\mathbb{C}\backslash\Sigma_{\bf R}$ where $\Sigma_{\bf R}=\partial\mathbb{D}_r(-1)\cup\partial\mathbb{D}_r(1)\cup\widehat{\gamma}_U\cup\widehat{\gamma}_L$ is displayed in Figure \ref{figure10}.
	\item[(2)] The limiting values ${\bf R}_{\pm}(\lambda),\lambda\in\Sigma_{\bf R}$ obey
	\begin{eqnarray*}
		{\bf R}_+(\lambda)&=&{\bf R}_-(\lambda)\begin{bmatrix}1 & \gamma\e^v\e^{2\im t\lambda}(\frac{\lambda+1}{\lambda-1})^{2\nu}\\ 0&1\end{bmatrix},\ \lambda\in\widehat{\gamma}_U;\\
		{\bf R}_+(\lambda)&=&{\bf R}_-(\lambda)\begin{bmatrix}1 &0\\ -\gamma\e^v\e^{-2\im t\lambda}(\frac{\lambda+1}{\lambda-1})^{-2\nu} & 1\end{bmatrix},\ \lambda\in\widehat{\gamma}_L.
	\end{eqnarray*}
	on the lens boundaries and
	\begin{equation*}
		{\bf R}_+(\lambda)={\bf R}_-(\lambda){\bf P}^{(\pm 1)}(\lambda)\big({\bf P}^{(\infty)}(\lambda)\big)^{-1},\ \lambda\in\partial\mathbb{D}_r(\pm1)		\end{equation*}
	on the circles.
	\item[(3)] As $\lambda\rightarrow\infty$, we have ${\bf R}(\lambda)\rightarrow\mathbb{I}$.
\end{enumerate}
\end{problem}
\begin{figure}[tbh]
	\begin{center}
\resizebox{0.45\textwidth}{!}{\includegraphics{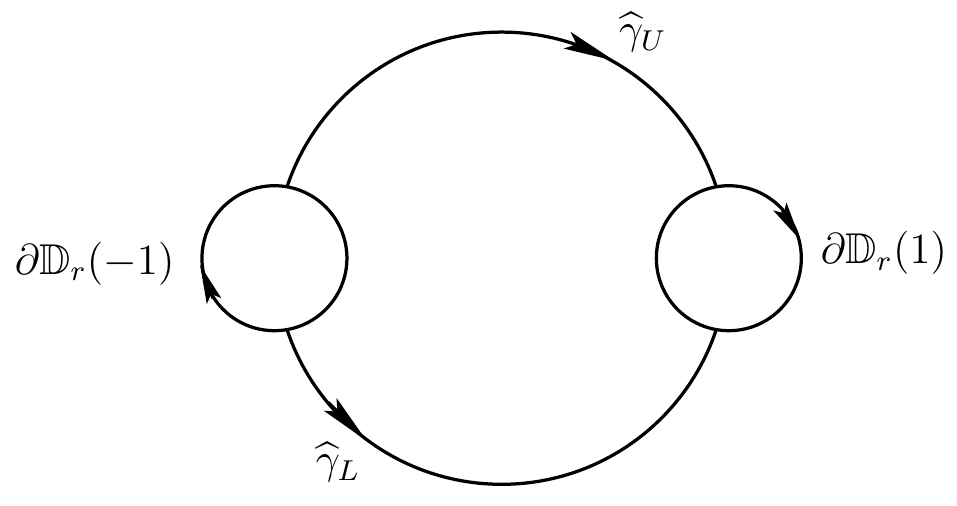}}
\caption{The oriented jump contours for the ratio function ${\bf R}(\lambda;t,v)$ in the complex $\lambda$-plane.}
\label{figure10}
\end{center}
\end{figure}
We now argue that the last RHP is asymptotically solvable in the scaling region \eqref{sinescale} by deriving small norm estimates for the underlying jump matrix ${\bf G}_{\bf R}(\lambda;t,v)$ and using the general theory of \cite{DZ}. First, from property (2) in RHP \ref{J:ratio} and \eqref{es:3},
\begin{prop} There exists $t_0>0$ such that
\begin{equation*}
	\|{\bf G}_{\bf R}(\cdot;t,v)-\mathbb{I}\|_{L^2\cap L^{\infty}(\widehat{\gamma}_U\cup\widehat{\gamma}_L)}\leq \e^{v-2tr-v\alpha(r)},\ \ \ \ \forall\,t\geq t_0,\ \ v\in[0,+\infty)
\end{equation*}
where $\alpha(r)=\frac{1}{2}(1-\frac{2}{\pi}\arctan(\frac{r}{2}))$ and $0<r<\frac{1}{4}$ is fixed.
\end{prop}
Second, using property (2) again and the matching relations (2.5), (2.7) in \cite{BDIK2},
\begin{prop} For any fixed $v\in[0,+\infty)$ there exist positive constants $t_0=t_0(v)$ and $c=c(v)$ such that
\begin{equation*}
	\|{\bf G}_{\bf R}(\cdot;t,v)-\mathbb{I}\|_{L^2\cap L^{\infty}(\partial\mathbb{D}_r(\pm 1))}\leq\frac{c}{t},\ \ \ \ \forall\,t\geq t_0
\end{equation*}
where $0<r<\frac{1}{4}$ is fixed throughout.
\end{prop}
Now combining these two estimates and using \cite{DZ}, we arrive at
\begin{theo}\label{sinefinish} For any fixed $v\in[0,+\infty)$ there exist $t_0=t_0(v)>0$ and $c=c(v)>0$ such that the ratio RHP \ref{J:ratio} is uniquely solvable in $L^2(\Sigma_{\bf R})$ for all $t\geq t_0$. The solution can be computed iteratively through the integral equation
\begin{equation*}
	{\bf R}(\lambda)=\mathbb{I}+\frac{1}{2\pi\im}\int_{\Sigma_{\bf R}}{\bf R}_-(w)\big({\bf G}_{\bf R}(w)-\mathbb{I}\big)\frac{\d w}{w-\lambda},\ \ \lambda\in\mathbb{C}\backslash\Sigma_{\bf R}
\end{equation*}
with the help of
\begin{equation*}
	\|{\bf R}_-(\cdot;t,v)-\mathbb{I}\|_{L^2(\Sigma_{\bf R})}\leq c\,t^{-1},\ \ \ \ \forall\,t\geq t_0.
\end{equation*}
\end{theo}
It is now time to extract the relevant asymptotic expansions and substitute the information back into \eqref{J:1} and \eqref{J:2}.
\subsection{Extraction of asymptotics and proof of expansion \eqref{JME:31}}
From the transformation sequence
\begin{equation*}
	{\bf Y}(\lambda)\mapsto {\bf X}(\lambda)\mapsto {\bf R}(\lambda),
\end{equation*}
we obtain at once the exact identity
\begin{equation*}
	{\bf Y}_1=2\nu\sigma_3+\frac{\im}{2\pi}\int_{\Sigma_{\bf R}}{\bf R}_-(w)\big({\bf G}_{\bf R}(w)-\mathbb{I}\big)\,\d w.
\end{equation*}
Now let
\begin{equation*}
	{\bf M}=\frac{\im}{2\pi}\int_{\Sigma_{\bf R}}\big({\bf G}_{\bf R}(w)-\mathbb{I}\big)\,\d w,\ \ \ \ {\bf M}=\big(M^{jk}\big)_{j,k=1}^2,
\end{equation*}
so that from an explicit residue computation (using Theorem \ref{sinefinish} and \cite{BDIK2}, (2.5), (2.7)),
\begin{prop} As $t\rightarrow+\infty$,
\begin{equation*}
	M^{11}=-\frac{\im\nu^2}{t}+\mathcal{O}\left(t^{-2}\right),\ \ \ M^{12}=\frac{\nu}{t}\sin\big(\varphi(t,v)\big)+\mathcal{O}\left(t^{-2}\right),
\end{equation*}
where
\begin{equation*}
	\varphi(t,v)=2t-\frac{v}{\pi}\ln(4t)+2\,\textnormal{arg}\,\Gamma\left(\frac{\im v}{2\pi}\right).
\end{equation*}
All error terms are uniform with respect to $v$ chosen from compact subsets of $[0,+\infty)$.
\end{prop}
Thus in turn with \eqref{J:1} and \eqref{J:2},
\begin{cor}\label{cor:J} As $t\rightarrow+\infty$, with $k\in\mathbb{Z}$ and fixed $v\in[0,+\infty)$,
\begin{equation*}
	q(t;\gamma)=-\frac{\im v}{2\pi}\left(1-\frac{v}{4\pi t}+\frac{v}{2\pi t}\sin^2\big(\varphi(t,v)\big)+\mathcal{O}\left(t^{-2}\right)\right);\ \ \ \ \ p(t;\gamma)=2\im\varphi(t,v)+2\pi\im(1+2k)+\mathcal{O}\left(t^{-1}\right)
\end{equation*}
and
\begin{equation*}
	\mathcal{H}_B\big(q(t;\gamma),p(t;\gamma),t\big)=-\frac{2v}{\pi}+\frac{v^2}{2\pi^2t}+\mathcal{O}\left(t^{-2}\right).
\end{equation*}
\end{cor}
The last result, together with Proposition \ref{theo:1}, leads us to
\begin{cor}\label{pe:1} As $t\rightarrow+\infty$,
\begin{equation*}
	\ln F_B(t;\gamma)=-\frac{2v}{\pi}t+\frac{v^2}{2\pi^2}\ln t+E(v)+\mathcal{O}\left(t^{-1}\right),
\end{equation*}
where $E(v)$ is $t$-independent and the error term uniform for $v$ chosen from compact subsets of $[0,+\infty)$.
\end{cor}
Similar to the last section we now determine $E(v)$ through \eqref{JME:24} and Corollary \ref{logic}. First, as $t\rightarrow+\infty$,
\begin{equation*}
	t\,\mathcal{H}_B(q,p,t)-pq=-\frac{4vt}{\pi}+\frac{v^2}{\pi^2}-\frac{v^2}{\pi^2}\sin^2\big(\varphi(t,v)\big)+\frac{v^2}{\pi^2}\ln(4t)-\frac{2v}{\pi}\textnormal{arg}\,\Gamma\left(\frac{\im v}{2\pi}\right)-v(1+2k)+\mathcal{O}\left(t^{-1}\right),
\end{equation*}
with $k\in\mathbb{Z}$ as in Corollary \ref{cor:J}. But using Corollary \ref{cor:J} again it is also easy to see that
\begin{equation*}
	pq_{\gamma}=\frac{\d}{\d\gamma}\left(\frac{2vt}{\pi}-\frac{v^2}{2\pi^2}+\frac{v^2}{\pi^2}\sin^2\big(\varphi(t,v)\big)-\frac{v^2}{2\pi^2}\ln(4t)+ v(1+2k)\right)+\frac{2v_{\gamma}}{\pi}\textnormal{arg}\,\Gamma\left(\frac{\im v}{2\pi}\right)+\mathcal{O}\left(t^{-1}\ln t\right),
\end{equation*}
and therefore together with Corollary \ref{pe:1}, after one final integration by parts (and the fact $F_B(t;0)=1$),
\begin{prop} As $t\rightarrow+\infty$,
\begin{equation*}
	\ln F_B(t;\gamma)=-\frac{2v}{\pi}t+\frac{v^2}{2\pi^2}\ln(4t)+\frac{v^2}{2\pi^2}-\frac{2}{\pi}\int_0^{\gamma}v(\gamma')\frac{\d}{\d\gamma'}\textnormal{arg}\,\Gamma\left(\frac{\im v}{2\pi}\right)\,\d\gamma'+\mathcal{O}\left(t^{-1}\right)
\end{equation*}
uniformly for $\gamma\in[0,1)$ chosen from compact subsets.
\end{prop}
By standard properties of the Barnes G-function, see \eqref{Barnesid}, this results proves Theorem \ref{theo:3}, expansion \eqref{JME:31}.
\section{Proof of Theorem \ref{theo:3}, expansion \eqref{JME:33}}\label{sec:5}
As shown in Appendix \ref{AppA} below, we can characterize the functions $(q,p)$ in \eqref{HH:2}, \eqref{JME:19} through the solution of the following Riemann-Hilbert problem.
\begin{problem}\label{BessRHP} Let $t>0,\alpha>-1$ and $\gamma\in[0,1]$. Determine the piecewise analytic function ${\bf Y}={\bf Y}(\lambda;t,\alpha,\gamma)\in\mathbb{C}^{2\times 2}$ such that
\begin{enumerate}
	\item[(1)] ${\bf Y}(\lambda)$ is analytic for $\lambda\in\mathbb{C}\backslash[0,1]$ with the line segment $[0,1]\subset\mathbb{R}$ oriented from left to right.
	\item[(2)] The limiting values ${\bf Y}_{\pm}(\lambda)=\lim_{\epsilon\downarrow 0}{\bf Y}(\lambda\pm\im\epsilon)$ along $\lambda\in(0,1)$ obey the jump relation
	\begin{equation*}
		{\bf Y}_+(\lambda)={\bf Y}_-(\lambda)\begin{bmatrix}1-\im\pi\gamma\sqrt{\lambda t}J_{\alpha}(\sqrt{\lambda t})J_{\alpha}'(\sqrt{\lambda t})&\im\pi\gamma\big(J_{\alpha}(\sqrt{\lambda t})\big)^2\\ -\im\pi\gamma\big(\sqrt{\lambda t}J_{\alpha}'(\sqrt{\lambda t})\big)^2 & 1+\im\pi\gamma\sqrt{\lambda t}J_{\alpha}(\sqrt{\lambda t})J_{\alpha}'(\sqrt{\lambda t})\end{bmatrix},\ \ \ \lambda\in(0,1).
	\end{equation*}
	\begin{figure}[tbh]
	\begin{center}
\resizebox{0.42\textwidth}{!}{\includegraphics{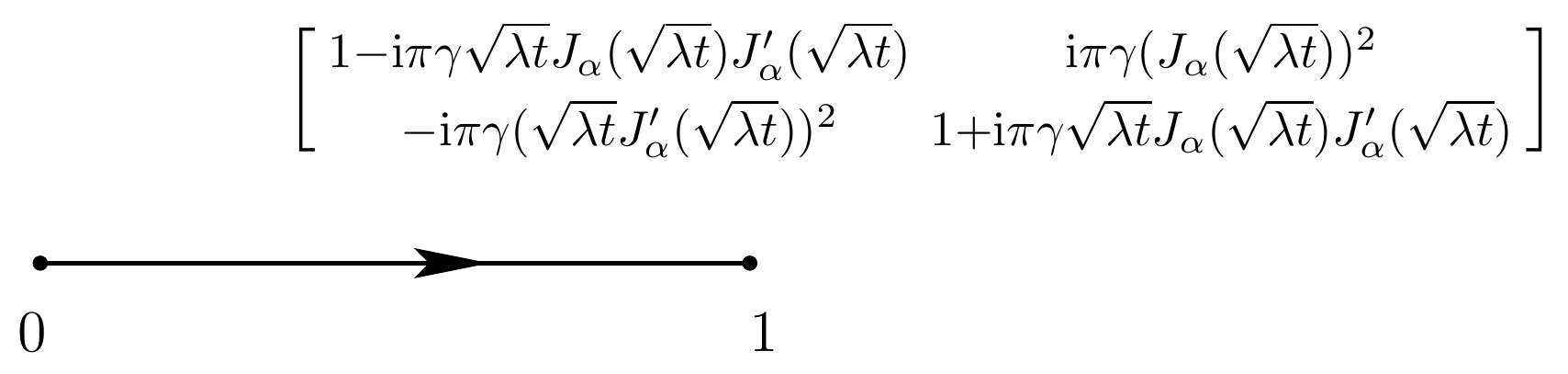}}
\caption{The oriented jump contour for the master function ${\bf Y}(\lambda;t,\alpha,\gamma)$ of RHP \ref{BessRHP} in the complex $\lambda$-plane.}
\label{figure11}
\end{center}
\end{figure}
	\item[(3)] ${\bf Y}(\lambda)$ is square integrable on $[0,1]\subset\mathbb{R}$.
	\item[(4)] As $\lambda\rightarrow\infty$,
	\begin{equation}\label{BessAsy}
		{\bf Y}(\lambda)=\mathbb{I}+{\bf Y}_1\lambda^{-1}+{\bf Y}_2\lambda^{-2}+\mathcal{O}\left(\lambda^{-3}\right),\ \ \ \ {\bf Y}_{\ell}=\big(Y_{\ell}^{jk}\big)_{j,k=1}^2.
	\end{equation}
\end{enumerate}
\end{problem}
We will prove below that the above problem for ${\bf Y}(\lambda)$ is uniquely solvable for all $t\geq t_0$ and $\gamma\in[0,1),\alpha>-1$ fixed. In turn we have the representation formul\ae\,(see Appendix \ref{AppA} below)
\begin{equation}\label{Bess:f1}
	q^2(t,\alpha;\gamma)=t\big(Y_1^{12}\big)^2+2\big(Y_1^{11}-Y_1^{12}\big),\ \ \ \ \ p^2(t,\alpha;\gamma)=\frac{\alpha^2q^2}{(q^2-1)^2}+\frac{2t}{q^2-1}\left(Y_1^{12}+\frac{q^2}{2}\right),
\end{equation}
and
\begin{equation}\label{Bess:f2}
	\mathcal{H}_H\big(q(t,\alpha;\gamma),p(t,\alpha,\gamma),t,\alpha)=\frac{1}{2}Y_1^{12},
	%
\end{equation}
through RHP \ref{BessRHP}. Moreover,
\begin{eqnarray}
	L(t,\alpha;\gamma)&=&-\frac{\alpha}{2}\ln\Big(-\widehat{X}^{11}(0)\big(\widehat{X}^{12}(0)\big)^{-1}\Big)\bigg|_{s=0}^t,\ \ \ \ \ -1<\alpha<0,\label{cool:3}\\
	L(t,\alpha;\gamma)&=&-\frac{\alpha}{2}\ln\Big(-\widehat{X}^{11}(0)\big(\widehat{X}^{12}(0)\big)^{-1} s^{-\alpha}\Big)\bigg|_{s=0}^t,\ \ \ \ \ \ \alpha\geq 0.\label{cool:4}
\end{eqnarray}
in terms of RHP \ref{Bess:X} below.
Formul\ae\,\eqref{Bess:f1}, \eqref{Bess:f2} and \eqref{cool:3}, \eqref{cool:4} are the starting point for our asymptotic analysis, but the necessary nonlinear steepest descent techniques (for $\gamma\in[0,1)$ that is, in case $\gamma=1$ see \cite{B2}, Section $6$) have not appeared in the literature yet, thus we provide the details below.
\subsection{Nonlinear steepest descent analysis for RHP \ref{BessRHP}}
Our goal is to solve RHP \ref{BessRHP} for ${\bf Y}(\lambda;t,\alpha,\gamma)\in\mathbb{C}^{2\times 2}$ for all values $(t,\alpha,v)\in\mathbb{R}_+\times\mathbb{R}_{>-1}\times\mathbb{R}_+$ such that
\begin{equation*}
	t\geq t_0,\ \ \ \textnormal{and}\ \ \ \alpha>-1,\ \ 0\leq v=-\ln(1-\gamma)<+\infty\ \ \ \textnormal{are fixed}.
\end{equation*}
To this end we shall first recall a few key steps from \cite{B2}, Section $6.1$. Let $\Psi_{\alpha}(\z)$ denote the function defined in \eqref{BessUndress} below. It allows us to factorize the jump matrix in RHP \ref{BessRHP} as follows, for $\lambda>0$,
\begin{equation}\label{facto}
	\big(\Psi_{\alpha}(\lambda t)\big)^{-1}_+\begin{bmatrix}1-\im\pi\gamma\sqrt{\lambda t}J_{\alpha}(\sqrt{\lambda t})J_{\alpha}'(\sqrt{\lambda t})&\im\pi\gamma\big(J_{\alpha}(\sqrt{\lambda t})\big)^2\\ -\im\pi\gamma\big(\sqrt{\lambda t}J_{\alpha}'(\sqrt{\lambda t})\big)^2 & 1+\im\pi\gamma\sqrt{\lambda t}J_{\alpha}(\sqrt{\lambda t})J_{\alpha}'(\sqrt{\lambda t})\end{bmatrix}\big(\Psi_{\alpha}(\lambda t)\big)_+=\begin{bmatrix} 1&-\gamma \\ 0 & 1\end{bmatrix},
\end{equation}
and thus motivates an undressing transformation. In more detail, using the model function $\Psi(\z;\alpha)$ from \eqref{Bessundress} and RHP \ref{Bessbetter} from Appendix \ref{AppB}, we define
\begin{equation}\label{BessT:1}
	{\bf X}(\lambda)=\begin{bmatrix} 1 & 0\\ b_1(\alpha)&1\end{bmatrix}{\bf Y}(\lambda)\Psi(\lambda t;\alpha)\begin{cases}\bigl[\begin{smallmatrix} 1 & 0\\ \e^{-\im\pi\alpha} & 1\end{smallmatrix}\bigr],&\lambda\in\widehat{\Omega}_1\smallskip\\ \bigl[\begin{smallmatrix} 1 & 0\\ -\e^{\im\pi\alpha} & 1\end{smallmatrix}\bigr],&\lambda\in\widehat{\Omega}_2\\  \mathbb{I},&\textnormal{else}\end{cases}.
\end{equation}
In view of Figure \ref{figure12}, this step leads us to the following transformed RHP.
\begin{figure}[tbh]
	\begin{center}
\resizebox{0.35\textwidth}{!}{\includegraphics{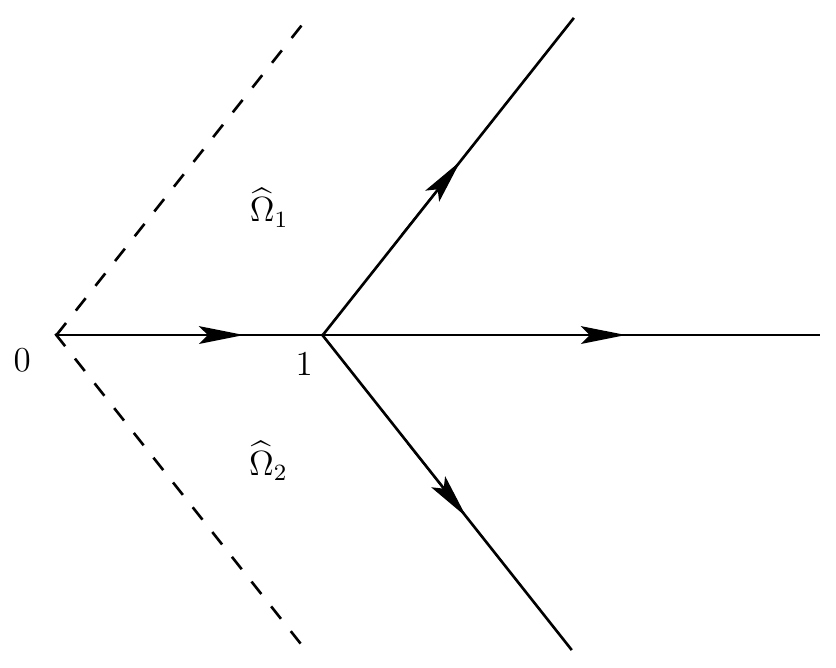}}
\caption{The oriented jump contours as solid black lines for the function ${\bf X}(\lambda)$ defined in \eqref{BessT:1} in the complex $\lambda$-plane.}
\label{figure12}
\end{center}
\end{figure}
\begin{problem}\label{Bess:X} Find ${\bf X}(\lambda)={\bf X}(\lambda;t,\alpha,v)\in\mathbb{C}^{2\times 2}$ with $(t,\alpha,v)\in\mathbb{R}_{>0}\times\mathbb{R}_{>-1}\times\mathbb{R}_{\geq 0}$ such that
\begin{enumerate}
	\item[(1)] ${\bf X}(\lambda)$ is analytic for $\lambda\in\mathbb{C}\backslash\Sigma_{\bf X}$ where the oriented contour $\Sigma_{\bf X}$ is shown in Figure \ref{figure12} as union of solid black lines.
	\item[(2)] The jumps on $\Sigma_{\bf X}$ read as
	\begin{equation*}
		{\bf X}_+(\lambda)={\bf X}_-(\lambda)\begin{bmatrix}\e^{-\im\pi\alpha} & \e^{-v}\\ 0 & \e^{\im\pi\alpha}\end{bmatrix},\ \ \lambda\in(0,1);\ \ \ \ \ {\bf X}_+(\lambda)={\bf X}_-(\lambda)\begin{bmatrix}0 & 1\\ -1 & 0\end{bmatrix},\ \ \lambda\in(1,+\infty)
	\end{equation*}
	on the positive real axis and
	\begin{eqnarray*}
		{\bf X}_+(\lambda)&=&{\bf X}_-(\lambda)\begin{bmatrix}1&0\\ \e^{-\im\pi\alpha} & 1\end{bmatrix},\ \ \textnormal{arg}(\lambda-1)=\frac{\pi}{3};\\
		{\bf X}_+(\lambda)&=&{\bf X}_-(\lambda)\begin{bmatrix} 1 & 0\\ \e^{\im\pi\alpha} & 1\end{bmatrix},\ \ \ \ \textnormal{arg}(\lambda-1)=\frac{5\pi}{3}.
	\end{eqnarray*}
	\item[(3)] In a vicinity of $\lambda=0$,
	\begin{equation*}
		{\bf X}(\lambda)=\widehat{{\bf X}}(\lambda)(-\lambda)^{\frac{\alpha}{2}\sigma_3}\begin{cases}\begin{bmatrix}1 & \frac{\im}{2}\frac{1-\gamma}{\sin\pi\alpha} \\ 0&1\end{bmatrix},&\ \alpha\notin\mathbb{Z}\smallskip\\
		 \begin{bmatrix}1 & -\frac{\e^{\im\pi\alpha}}{2\pi\im}\e^{-v}\ln(-\lambda)\\ 0 & 1\end{bmatrix},&\ \alpha\in\mathbb{Z}\end{cases}
	\end{equation*}
	where $\ln:\mathbb{C}\backslash(-\infty,0]\rightarrow\mathbb{C}$ and $z^{\alpha}:\mathbb{C}\backslash(-\infty,0]\rightarrow\mathbb{C}$ are defined with principal branches and $\widehat{{\bf X}}(\lambda)=\big(\widehat{X}^{jk}(\lambda)\big)_{j,k=1}^2$ is analytic at $\lambda=0$.
	\item[(4)] In a vicinity of $\lambda=1$,
	\begin{align*}
		{\bf X}(\lambda)=&\,\widehat{{\bf X}}(\lambda)\left\{\mathbb{I}+\frac{\gamma}{2\pi\im}\begin{bmatrix}-1 & -\e^{-\im\pi\alpha}\\ \e^{\im\pi\alpha} & 1\end{bmatrix}\ln(\lambda-1)\right\}\begin{cases}\bigl[\begin{smallmatrix}0 & 1\\-1 & 0\end{smallmatrix}\bigr],&\Im\lambda>0\\ \mathbb{I},&\Im\lambda<0\end{cases}\\
		&\,\times\,\begin{cases}\bigl[\begin{smallmatrix}1&0\\ \e^{-\im\pi\alpha} & 1\end{smallmatrix}\bigr],&\textnormal{arg}(\lambda-1)\in(\frac{\pi}{3},\pi)\smallskip\\ \bigl[\begin{smallmatrix}1 & 0\\ -\e^{\im\pi\alpha} & 1\end{smallmatrix}\bigr],&\textnormal{arg}(\lambda-1)\in(\pi,\frac{5\pi}{3})\smallskip\\ \mathbb{I},&\textnormal{else}\end{cases},\ \ \ \ 
	\end{align*}
	with $\widehat{{\bf X}}(\lambda)$ analytic at $\lambda=1$ and the principal branch for $\ln(\lambda-1)$, i.e. $-\pi<\textnormal{arg}(\lambda-1)<\pi$.
	\item[(5)] Using RHP \ref{Bessbetter} we find that as $\lambda\rightarrow\infty, \lambda\notin[0,+\infty)$,
	\begin{equation*}
		{\bf X}(\lambda)=\Big\{\mathbb{I}+{\bf X}_1\lambda^{-1}+\mathcal{O}\left(\lambda^{-2}\right)\Big\}(-\lambda t)^{-\frac{1}{4}\sigma_3}\frac{1}{\sqrt{2}}\begin{bmatrix}1 & -1\\ 1 & 1\end{bmatrix}\e^{-\im\frac{\pi}{4}\sigma_3}\e^{\sqrt{t}(-\lambda)^{\frac{1}{2}}\sigma_3}
	\end{equation*}
	with
	\begin{equation*}
		{\bf X}_1=\begin{bmatrix}1&0\\b_1(\alpha) & 1\end{bmatrix}{\bf Y}_1\begin{bmatrix}1 & 0\\ -b_1(\alpha) & 1\end{bmatrix}-\frac{1}{t}\begin{bmatrix}a_2(\alpha) & -a_1(\alpha)\\ b_1(\alpha)a_2(\alpha)-b_3(\alpha) & b_2(\alpha)-b_1(\alpha)a_1(\alpha)\end{bmatrix}.
	\end{equation*}
\end{enumerate}
\end{problem}
Our next step is the $g$-function transformation given by
\begin{equation}\label{Bess:T2}
	{\bf T}(\lambda)={\bf X}(\lambda)\e^{-\sqrt{t}g(\lambda)\sigma_3},\ \ \ \lambda\in\mathbb{C}\backslash\Sigma_{\bf X};\ \ \ \ \ \ \ g(\lambda)=(-\lambda)^{\frac{1}{2}},\ \ \ \lambda\in\mathbb{C}\backslash[0,\infty)
\end{equation}
where $g(\lambda)$ is defined and analytic for $\lambda\in\mathbb{C}\backslash[0,+\infty)$ such that $(-\lambda)^{\frac{1}{2}}>0$ for $\lambda<0$.
\begin{prop}\label{Bess:gfunc} We have
\begin{equation*}
	g_{\pm}(\lambda)=\sqrt{|\lambda|},\ \ \ \lambda<0;\ \ \ \ \ \ \ \ g_{\pm}(\lambda)=\mp\im\sqrt{\lambda},\ \ \ \ \lambda>0;\ \ \ \ \ \ \ \ \Re\big(g(\lambda)\big)>0,\ \ \ \textnormal{arg}(\lambda-1)=\frac{\pi}{3},\frac{5\pi}{3}
\end{equation*}
and $\Pi(\lambda)=2\im\sqrt{\lambda},\lambda>0$ admits analytic continuation into a small vicinity of $(0,1)$ into the lower and upper half plane. In fact with
\begin{equation*}
	\phi(\lambda)=-2(-\lambda)^{\frac{1}{2}},\ \ \ \ \lambda\in\mathbb{C}\backslash[0,\infty)
\end{equation*}
we observe that
\begin{equation*}
	\phi_+(\lambda)=\Pi(\lambda)=-\phi_-(\lambda),\ \ \lambda>0;\ \ \ \ \ \ \Re\big(\phi(\lambda)\big)<0\ \ \ \textnormal{for}\ \Im\lambda\gtrless 0,\ \ \ \Re\lambda\in(0,1).
\end{equation*}
\end{prop}
Recalling RHP \ref{Bess:X} the transformation \eqref{Bess:T2} leads to the following problem
\begin{problem}\label{Bess:T} Find ${\bf T}(\lambda)={\bf T}(\lambda;t,\alpha,v)\in\mathbb{C}^{2\times 2}$ with $(t,\alpha,v)\in\mathbb{R}_{>0}\times\mathbb{R}_{>-1}\times\mathbb{R}_{\geq 0}$ such that
\begin{enumerate}
	\item ${\bf T}(\lambda)$ is analytic for $\lambda\in\mathbb{C}\backslash\Sigma_{\bf X}$ with $\Sigma_{\bf X}$ shown in Figure \ref{figure12}.
	\item On the contour $\Sigma_{\bf X}$,
	\begin{equation*}
		{\bf T}_+(\lambda)={\bf T}_-(\lambda)\begin{bmatrix}\e^{-\im\pi\alpha} \e^{\sqrt{t}\,\Pi(\lambda)} & 1-\gamma\\ 0 & \e^{\im\pi\alpha}\e^{-\sqrt{t}\,\Pi(\lambda)} \end{bmatrix},\ \ \lambda\in(0,1);\ \ \ \ {\bf T}_+(\lambda)={\bf T}_-(\lambda)\begin{bmatrix}0&1\\ -1 & 0\end{bmatrix},\ \ \lambda\in(1,+\infty)
	\end{equation*}
	followed by
	\begin{equation*}
		{\bf T}_+(\lambda)={\bf T}_-(\lambda)\begin{bmatrix}1 & 0\\ \e^{-\im\pi\alpha}\e^{-2\sqrt{t}\,g(\lambda)} & 1\end{bmatrix},\ \ \textnormal{arg}(\lambda-1)=\frac{\pi}{3},
	\end{equation*}
	\begin{equation*}
		{\bf T}_+(\lambda)={\bf T}_-(\lambda)\begin{bmatrix}1 & 0\\ \e^{\im\pi\alpha}\e^{-2\sqrt{t}\,g(\lambda)} & 1\end{bmatrix},\ \ \textnormal{arg}(\lambda-1)=\frac{5\pi}{3}.
	\end{equation*}
	\item The singular behavior near $\lambda=0$ and $\lambda=1$ is, modulo the right multiplication with the $g$-function, see \eqref{Bess:T2}, unchanged from the corresponding behavior stated in RHP \ref{Bess:X}, compare conditions (3) and (4).
	\item As $\lambda\rightarrow\infty,\lambda\notin[0,+\infty)$, we have that
	\begin{equation*}
		{\bf T}(\lambda)=\Big\{\mathbb{I}+{\bf X}_1\lambda^{-1}+\mathcal{O}\left(\lambda^{-2}\right)\Big\}(-\lambda t)^{-\frac{1}{4}\sigma_3}\frac{1}{\sqrt{2}}\begin{bmatrix} 1 & -1\\ 1 & 1\end{bmatrix}\e^{-\im\frac{\pi}{4}\sigma_3}.
	\end{equation*}
\end{enumerate}
\end{problem}
Since
\begin{equation*}
	\begin{bmatrix}\e^{-\im\pi\alpha}\e^{\sqrt{t}\,\Pi(\lambda)} & 1-\gamma\\ 0 & \e^{\im\pi\alpha}\e^{-\sqrt{t}\,\Pi(\lambda)} \end{bmatrix}=\begin{bmatrix}1 & 0\\ \e^{\sqrt{t}\,\phi_-(\lambda)}\e^{v+\im\pi\alpha} & 1\end{bmatrix}\begin{bmatrix}0 & \e^{-v}\\ -\e^{v} & 0\end{bmatrix}\begin{bmatrix}1 & 0\\ \e^{\sqrt{t}\,\phi_+(\lambda)}\e^{v-\im\pi\alpha} & 1\end{bmatrix},
\end{equation*}
we can use Proposition \ref{Bess:gfunc} and perform our next transformation. Define with the help of Figure \ref{figure13}
\begin{equation}\label{Bess:T3}
	{\bf S}(\lambda)={\bf T}(\lambda)\begin{cases}\Bigl[\begin{smallmatrix}1 & 0\\ -\e^{\sqrt{t}\,\phi(\lambda)}\e^{v-\im\pi\alpha} & 1\end{smallmatrix}\Bigr],&\lambda\in\Omega_U\smallskip\\ \Bigl[\begin{smallmatrix} 1 & 0\\ \e^{\sqrt{t}\,\phi(\lambda)}\e^{v+\im\pi\alpha} & 1\end{smallmatrix}\Bigr],&\lambda\in\Omega_L\smallskip\\ \mathbb{I},&\textnormal{else}\end{cases}
\end{equation}
and obtain the following problem.
\begin{figure}[tbh]
	\begin{center}
\resizebox{0.4\textwidth}{!}{\includegraphics{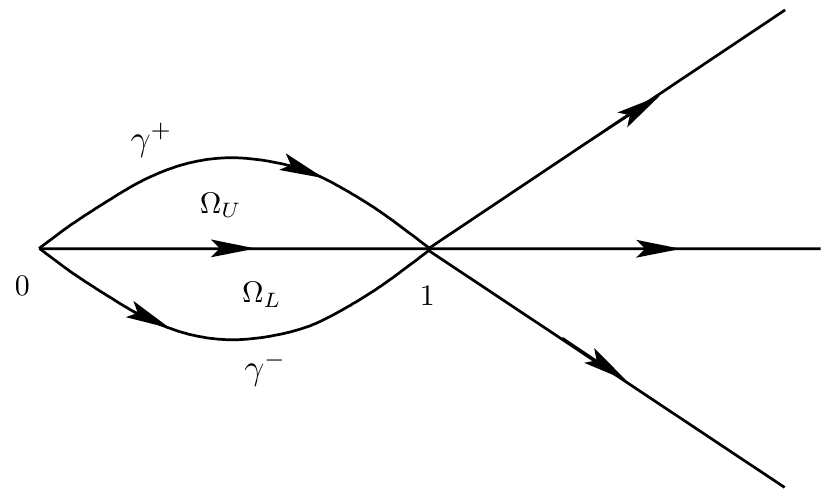}}
\caption{The domains used in Definition \eqref{Bess:T3}. The union of the solid black contours equals $\Sigma_{\bf S}$.}
\label{figure13}
\end{center}
\end{figure}
\begin{problem}\label{Bess:S} Find ${\bf S}(\lambda)={\bf S}(\lambda;t,\alpha,v)\in\mathbb{C}^{2\times 2}$ such that
\begin{enumerate}
	\item ${\bf S}(\lambda)$ is analytic for $\lambda\in\mathbb{C}\backslash\Sigma_{\bf S}$ and $\Sigma_{\bf S}$ is shown in Figure \ref{figure13}.
	\item The jumps are as follows,
	\begin{equation*}
		{\bf S}_+(\lambda)={\bf S}_-(\lambda)\e^{-\frac{v}{2}\sigma_3}\begin{bmatrix}0 & 1\\ -1 & 0\end{bmatrix}\e^{\frac{v}{2}\sigma_3},\ \ \ \lambda\in(0,1);\ \ \ \ \ \ {\bf S}_+(\lambda)={\bf S}_-(\lambda)\begin{bmatrix}0 & 1\\ -1 & 0\end{bmatrix},\ \ \lambda\in(1,+\infty);
	\end{equation*}
	\begin{equation*}
		{\bf S}_+(\lambda)={\bf S}_-(\lambda)\begin{bmatrix}1 & 0\\ \e^{\sqrt{t}\,\phi(\lambda)}\e^{v-\im\pi\alpha} & 1\end{bmatrix},\ \ \lambda\in\gamma^+;\ \ \ {\bf S}_+(\lambda)={\bf S}_-(\lambda)\begin{bmatrix}1 & 0\\ \e^{\sqrt{t}\,\phi(\lambda)}\e^{v+\im\pi\alpha} & 1\end{bmatrix},\ \ \lambda\in\gamma^-;
	\end{equation*}
	\begin{equation*}
		{\bf S}_+(\lambda)={\bf S}_-(\lambda)\begin{bmatrix}1 & 0\\ \e^{\sqrt{t}\,\phi(\lambda)-\im\pi\alpha} & 1\end{bmatrix},\ \ \textnormal{arg}(\lambda-1)=\frac{\pi}{3};
	\end{equation*}
	\begin{equation*}
		{\bf S}_+(\lambda)={\bf S}_-(\lambda)\begin{bmatrix}1 & 0\\ \e^{\sqrt{t}\,\phi(\lambda)+\im\pi\alpha} & 1\end{bmatrix},\ \ \textnormal{arg}(\lambda-1)=\frac{5\pi}{3}.
	\end{equation*}
	\item The singular behavior of ${\bf T}(\lambda)$ near $\lambda=0$ and $\lambda=1$ has to be adjusted according to \eqref{Bess:T3}, i.e. we have to multiply the local expansions by the appropriate right multipliers.
	\item The behavior near $\lambda=\infty$ remains unchanged from RHP \ref{Bess:T}, i.e.
	\begin{equation}\label{Bess:T4}
		{\bf S}(\lambda)=\Big\{\mathbb{I}+{\bf X}_1\lambda^{-1}+\mathcal{O}\left(\lambda^{-2}\right)\Big\}(-\lambda t)^{-\frac{1}{4}\sigma_3}\frac{1}{\sqrt{2}}\begin{bmatrix} 1 & -1\\ 1 & 1\end{bmatrix}\e^{-\im\frac{\pi}{4}\sigma_3},\ \ \ \ \lambda\rightarrow\infty,\ \ \lambda\notin[0,+\infty).
	\end{equation}
\end{enumerate}
\end{problem}
We have now reached the point where the problem is localized. Indeed, in view of Proposition \ref{Bess:gfunc}, we have for the jump matrix ${\bf G}_{{\bf S}}(\lambda;t,\alpha,v)$ away from $\lambda=0,1$, 
\begin{equation}\label{Bess:inf1}
	{\bf G}_{\bf S}(\lambda;t,\alpha,v)=\mathbb{I}+\mathcal{O}\left(\e^{v-c\sqrt{t|\lambda|}}\,\right),\ \ \ c>0,\ \ \ \lambda\in\gamma^+\cup\gamma^-
\end{equation}
and
\begin{equation}\label{Bess:inf2}
	{\bf G}_{\bf S}(\lambda;t,\alpha,v)=\mathbb{I}+\mathcal{O}\left(\e^{-d\sqrt{t}}\,\right),\ \ \ d>0,\ \ \ \ \textnormal{arg}(\lambda-1)=\frac{\pi}{3},\frac{5\pi}{3}.
\end{equation}
For this reason we now focus on the local analysis on $(0,+\infty)\subset\mathbb{R}$ and near $\lambda=0,1$. First, the parametrix ${\bf P}^{(\infty)}(\lambda)$ for the line segment $(0,\infty)\subset\mathbb{R}$ will obey the following conditions:
\begin{problem}\label{problem:1} Determine ${\bf P}^{(\infty)}(\lambda)={\bf P}^{(\infty)}(\lambda;t,v)\in\mathbb{C}^{2\times 2}$ such that
\begin{enumerate}
	\item ${\bf P}^{(\infty)}(\lambda)$ is analytic for $\lambda\in\mathbb{C}\backslash[0,\infty)$
	\item We require that ${\bf P}^{(\infty)}(\lambda)$ assumes square integrable boundary values on $[0,\infty)$ which satisfy the jump conditions
	\begin{equation*}
		{\bf P}^{(\infty)}_+(\lambda)={\bf P}^{(\infty)}_-(\lambda)\e^{-\frac{v}{2}\sigma_3}\begin{bmatrix}0 &1\\ -1 &0\end{bmatrix}\e^{\frac{v}{2}\sigma_3},\ \ \ \lambda\in(0,1);
	\end{equation*}
	and
	\begin{equation*}
		{\bf P}^{(\infty)}_+(\lambda)={\bf P}^{(\infty)}_-(\lambda)\begin{bmatrix}0 &1\\ -1 &0\end{bmatrix},\ \ \ \lambda\in(1,\infty).
	\end{equation*}
	\item As $\lambda\rightarrow\infty$ with $\lambda\notin[0,\infty)\subset\mathbb{R}$, see \eqref{Bess:T4},
	\begin{equation*}
		{\bf P}^{(\infty)}(\lambda)=\left\{\mathbb{I}+\frac{1}{\lambda}\begin{bmatrix}2\nu^2 & -\frac{2\im\nu}{\sqrt{t}}\\ \frac{2\im}{3}\nu(1-4\nu^2)\sqrt{t} & -2\nu^2\end{bmatrix}+\mathcal{O}\left(\lambda^{-2}\right)\right\}(-\lambda t)^{-\frac{1}{4}\sigma_3}\frac{1}{\sqrt{2}}\begin{bmatrix}1 & -1\\ 1 & 1\end{bmatrix}\e^{-\im\frac{\pi}{4}\sigma_3}.
	\end{equation*}
\end{enumerate}
\end{problem}
It is easy to check that
\begin{equation}\label{OP:1}
	{\bf P}^{(\infty)}(\lambda)=\begin{bmatrix}1 & 0\\ -2\im\nu\sqrt{t} & 1\end{bmatrix}(-\lambda t)^{-\frac{1}{4}\sigma_3}\frac{1}{\sqrt{2}}\begin{bmatrix}1 & -1\\ 1 & 1\end{bmatrix}\e^{-\im\frac{\pi}{4}\sigma_3}\big(\mathcal{D}(\lambda)\big)^{-\sigma_3},\ \ \ \ \lambda\in\mathbb{C}\backslash[0,\infty)
\end{equation}
with
\begin{equation}\label{OP:2}
	\mathcal{D}(\lambda)=\exp\left[(-\lambda)^{\frac{1}{2}}\frac{v}{2\pi}\int_0^1\frac{\d w}{\sqrt{w}\,(w-\lambda)}\right]=\left(\frac{(-\lambda)^{\frac{1}{2}}-\im}{(-\lambda)^{\frac{1}{2}}+\im}\right)^{\nu},\ \ \ \ \ \nu=\frac{v}{2\pi\im}\in\im\mathbb{R}
\end{equation}
solves the above problem, provided we choose principal branches for all fractional exponents in \eqref{OP:1} and \eqref{OP:2}. Next, for the parametrix ${\bf P}^{(0)}(\lambda)$ in a vicinity of $\lambda=0$, we require the following properties:
\begin{problem}\label{problem:2} Determine ${\bf P}^{(0)}(\lambda)={\bf P}^{(0)}(\lambda;t,\alpha,v)\in\mathbb{C}^{2\times 2}$ such that
\begin{enumerate}
	\item ${\bf P}^{(0)}(\lambda)$ is analytic for $\lambda\in\mathbb{D}_{\frac{1}{4}}(0)\backslash(\Sigma_{\bf S}\cup\{0\})$ with $\mathbb{D}_r(\lambda_0)=\{\lambda\in\mathbb{C}:\ |\lambda-\lambda_0|<r\}$.
	\item On the three contours near $\lambda=0$ (compare Figure \ref{figure13}), the function ${\bf P}^{(0)}(\lambda)$ behaves as follows,
	\begin{eqnarray*}
		{\bf P}^{(0)}_+(\lambda)&=&{\bf P}^{(0)}_-(\lambda)\e^{-\frac{v}{2}\sigma_3}\begin{bmatrix}0 & 1\\ -1 & 0\end{bmatrix}\e^{\frac{v}{2}\sigma_3},\ \ \ \ \lambda\in(0,1)\cap\mathbb{D}_{\frac{1}{4}}(0);\\
		{\bf P}^{(0)}_+(\lambda)&=&{\bf P}^{(0)}_-(\lambda)\e^{-\frac{v}{2}\sigma_3}\begin{bmatrix}1 & 0\\ \e^{\sqrt{t}\,\phi(\lambda)-\im\pi\alpha} & 1\end{bmatrix}\e^{\frac{v}{2}\sigma_3},\ \ \ \ \lambda\in\gamma^+\cap\mathbb{D}_{\frac{1}{4}}(0);\\
		{\bf P}^{(0)}_+(\lambda)&=&{\bf P}^{(0)}_-(\lambda)\e^{-\frac{v}{2}\sigma_3}\begin{bmatrix}1 & 0\\ \e^{\sqrt{t}\,\phi(\lambda)+\im\pi\alpha} & 1\end{bmatrix}\e^{\frac{v}{2}\sigma_3},\ \ \ \ \lambda\in\gamma^-\cap\mathbb{D}_{\frac{1}{4}}(0).
	\end{eqnarray*}
	\item As $\lambda\rightarrow 0$, the singular behavior of ${\bf P}^{(0)}(\lambda)$ matches the behavior of the function ${\bf S}(\lambda)$ as described in RHP \ref{Bess:S}, condition (3).
	\item As $t\rightarrow+\infty$ with $\gamma\in[0,1)$ fixed, the local functions ${\bf P}^{(\infty)}(\lambda)$ and ${\bf P}^{(0)}(\lambda)$ obey the matching
	\begin{equation*}
		{\bf P}^{(0)}(\lambda)\sim\left\{\mathbb{I}+\sum_{m=1}^{\infty}{\bf E}^{(0)}(\lambda)\mathcal{S}_m(\alpha)\big({\bf E}^{(0)}(\lambda)\big)^{-1}(-\lambda t)^{-m}\right\}{\bf P}^{(\infty)}(\lambda),
	\end{equation*}
	with (compare RHP \ref{Bessbetter} in Appendix \ref{AppB})
	\begin{equation*}
		\mathcal{S}_m(\alpha)=\begin{bmatrix}a_{2m}(\alpha) & -a_{2m-1}(\alpha)\\ b_1(\alpha)a_{2m}(\alpha)-b_{2m+1}(\alpha) & b_{2m}(\alpha)-b_1(\alpha)a_{2m-1}(\alpha)\end{bmatrix}
	\end{equation*}
	which holds uniformly for $0<r_1\leq|\lambda|\leq r_2<\frac{1}{4}$ where $r_1$ and $r_2$ are fixed. The multiplier ${\bf E}^{(0)}(\lambda)$ is defined in \eqref{OP:4} below.
\end{enumerate}
\end{problem}
A solution to this problem is most easily constructed by recalling RHP \ref{Bessbetter}, or equivalently \eqref{Bessundress}: we define
\begin{equation}\label{OP:3}
	{\bf P}^{(0)}(\lambda)={\bf E}^{(0)}(\lambda)\begin{bmatrix}1 & 0\\ b_1(\alpha) & 1\end{bmatrix}\Psi(\lambda t;\alpha)\e^{-\sqrt{t}\,g(\lambda)\sigma_3}\e^{\frac{v}{2}\sigma_3},\ \ \ \ \lambda\in\mathbb{D}_{\frac{1}{4}}(0)\backslash(\Sigma_{\bf S}\cup\{0\})
\end{equation}
using the locally analytic multiplier
\begin{equation}\label{OP:4}
	{\bf E}^{(0)}(\lambda)={\bf P}^{(\infty)}(\lambda)\e^{-\frac{v}{2}\sigma_3}\e^{\im\frac{\pi}{4}\sigma_3}\frac{1}{\sqrt{2}}\begin{bmatrix}1 & 1\\ -1 & 1\end{bmatrix}(-\lambda t)^{\frac{1}{4}\sigma_3},\ \ \ \lambda\in\mathbb{D}_{\frac{1}{4}}(0).
\end{equation}
\begin{rem} As $\lambda\rightarrow 0$ we have
\begin{equation*}
	{\bf E}^{(0)}(\lambda)=t^{-\frac{1}{4}\sigma_3}\begin{bmatrix}1 & -2\im\nu\\ -2\im\nu & 1-4\nu^2\end{bmatrix}\left\{\mathbb{I}+\lambda\begin{bmatrix}-2\nu^2 & -\frac{2\im}{3}\nu(1-4\nu^2)\\ 2\im\nu & 2\nu^2\end{bmatrix}+\mathcal{O}\left(\lambda^2\right)\right\}t^{\frac{1}{4}\sigma_3}
\end{equation*}
and from \eqref{OP:1} we see that ${\bf E}^{(0)}(\lambda)=t^{-\frac{1}{4}\sigma_3}\widehat{{\bf E}}^{(0)}(\lambda)t^{\frac{1}{4}\sigma_3}$ with $\widehat{{\bf E}}^{(0)}(\lambda)$ independent of $t$.
\end{rem}
Our final parametrix near $\lambda=1$ obeys the following conditions:
\begin{problem}\label{problem:3} Find ${\bf P}^{(1)}(\lambda)={\bf P}^{(1)}(\lambda;t,\alpha,v)\in\mathbb{C}^{2\times 2}$ such that
\begin{enumerate}
	\item ${\bf P}^{(1)}(\lambda)$ is analytic for $\lambda\in\mathbb{D}_{\frac{1}{4}}(1)\backslash(\Sigma_{\bf S}\cup\{1\})$, see Figure \ref{figure13} for orientations.
	\item Along $\Sigma_S$, the limiting values ${\bf P}^{(1)}_{\pm}(\lambda)$ are square integrable and
	\begin{align*}
		{\bf P}^{(1)}_+(\lambda)=&\,{\bf P}^{(1)}_-(\lambda)\begin{bmatrix}1 & 0\\ \e^{\sqrt{t}\,\phi(\lambda)-\im\pi\alpha} & 1\end{bmatrix},\ \ \ \lambda\in\left\{\lambda\in\mathbb{C}:\ \textnormal{arg}(\lambda-1)=\frac{\pi}{3}\right\}\cap\mathbb{D}_{\frac{1}{4}}(1);\\
		{\bf P}^{(1)}_+(\lambda)=&\,{\bf P}^{(1)}_-(\lambda)\begin{bmatrix}0 & 1\\ -1 & 0\end{bmatrix},\ \ \ \lambda\in\left\{\lambda\in\mathbb{C}:\ \textnormal{arg}(\lambda-1)=0\right\}\cap\mathbb{D}_{\frac{1}{4}}(1);\\
		{\bf P}^{(1)}_+(\lambda)=&\,{\bf P}^{(1)}_-(\lambda)\begin{bmatrix}1 & 0\\ \e^{\sqrt{t}\,\phi(\lambda)+\im\pi\alpha} & 1\end{bmatrix},\ \ \ \lambda\in\left\{\lambda\in\mathbb{C}:\ \textnormal{arg}(\lambda-1)=\frac{5\pi}{3}\right\}\cap\mathbb{D}_{\frac{1}{4}}(1).
	\end{align*}
	Moreover,
	\begin{align*}
		{\bf P}^{(1)}_+(\lambda)=&{\bf P}^{(1)}_-(\lambda)\e^{-\frac{v}{2}\sigma_3}\begin{bmatrix}1 & 0\\ \e^{\sqrt{t}\,\phi(\lambda)\mp\im\pi\alpha} & 1\end{bmatrix}\e^{\frac{v}{2}\sigma_3},\ \ \ \lambda\in\gamma^{\pm}\cap\mathbb{D}_{\frac{1}{4}}(1);\\
		{\bf P}^{(1)}_+(\lambda)=&\,{\bf P}^{(1)}_-(\lambda)\e^{-\frac{v}{2}\sigma_3}\begin{bmatrix}0 & 1\\ -1 & 0\end{bmatrix}\e^{\frac{v}{2}\sigma_3},\ \ \ \lambda\in\left\{\lambda\in\mathbb{C}:\ \textnormal{arg}(\lambda-1)=\pi\right\}\cap\mathbb{D}_{\frac{1}{4}}(1).
	\end{align*}
	\item Near $\lambda=1$, the singular behavior of ${\bf P}^{(1)}(\lambda)$ is exactly of the form given in RHP \ref{Bess:S}, condition (3).
	\item As $t\rightarrow\infty$ with $\gamma\in[0,1)$ fixed, the two model functions ${\bf P}^{(\infty)}(\lambda)$ and ${\bf P}^{(1)}(\lambda)$ are related via
	\begin{equation*}
		{\bf P}^{(1)}(\lambda)\sim\left\{\mathbb{I}+\sum_{m=1}^{\infty}{\bf E}^{(1)}(\lambda)\mathcal{U}_m(\nu)\big({\bf E}^{(1)}(\lambda)\big)^{-1}\big(\im\z(\lambda)\big)^{-m}\right\}{\bf P}^{(\infty)}(\lambda),
	\end{equation*}
	with (compare RHP \ref{confluprob} below)
	\begin{equation*}
		\mathcal{U}_m(\nu)=\e^{\im\frac{\pi}{2}\nu\sigma_3}\begin{bmatrix}((\nu)_m)^2 & (-1)^mm((1-\nu)_{m-1})^2\frac{\Gamma(1-\nu)}{\Gamma(\nu)}\\ m\,((1+\nu)_{m-1})^2\frac{\Gamma(1+\nu)}{\Gamma(-\nu)} & (-1)^m((-\nu)_m)^2\end{bmatrix}\e^{-\im\frac{\pi}{2}\nu\sigma_3}\frac{1}{m!}
	\end{equation*}
	which holds uniformly for $0<r_1\leq|\lambda-1|\leq r_2<\frac{1}{4}$ with fixed $r_1$ and $r_2$. The multiplier ${\bf E}^{(1)}(\lambda)$ is defined in \eqref{OP:9} below and the local change of coordinates $\z=\z(\lambda)$ in \eqref{OP:8}.
\end{enumerate}
\end{problem}
The last problem is solved in terms of the confluent hypergeometric function $U(a,\z)\equiv U(a,1,\z)$, see \cite{NIST}, and our construction makes use of the model function $\Phi(\z)$ described in \eqref{OP:6} below, see again Appendix \ref{AppB}.
In more detail, define
\begin{equation}\label{OP:7}
	{\bf P}^{(1)}(\lambda)={\bf E}^{(1)}(\lambda)\Phi\big(\z(\lambda)\big)\begin{cases}\e^{-\im\frac{\pi}{2}\alpha\sigma_3}\e^{\frac{\im}{2}(\z(\lambda)+2\sqrt{t})\sigma_3},&\lambda\in\mathbb{D}_{\frac{1}{4}}(1)\backslash\Sigma_{\bf S}:\ \ \textnormal{arg}\,\lambda\in(0,\pi)\smallskip\\ \e^{\im\frac{\pi}{2}\alpha\sigma_3}\e^{-\frac{\im}{2}(\z(\lambda)+2\sqrt{t})\sigma_3},&\lambda\in\mathbb{D}_{\frac{1}{4}}(1)\backslash\Sigma_{\bf S}:\ \ \textnormal{arg}\,\lambda\in(\pi,2\pi)\end{cases}\,\times\e^{\im\frac{\pi}{2}(\nu+1)\sigma_3},
\end{equation}
where
\begin{equation}\label{OP:8}
	\lambda\in\mathbb{D}_{\frac{1}{4}}(1):\ \ \ \z(\lambda)=2\sqrt{t}\big(\im\,\textnormal{sgn}(\Im\lambda)g(\lambda)-1\big)=\sqrt{t}\,(\lambda-1)\left(1-\frac{1}{4}(\lambda-1)+\mathcal{O}\big((\lambda-1)^2\big)\right),\ \ \lambda\rightarrow 1,
\end{equation}
and
\begin{equation}\label{OP:9}
	{\bf E}^{(1)}(\lambda)={\bf P}^{(\infty)}(\lambda)\e^{-\im\frac{\pi}{2}(\nu+1)\sigma_3}\begin{cases}
	\e^{-\im\sqrt{t}\,\sigma_3}\e^{\im\frac{\pi}{2}\alpha\sigma_3}\Bigl[\begin{smallmatrix}0 & \e^{\im\frac{\pi}{2}\nu}\\ -\e^{-\im\frac{3\pi}{2}\nu} & 0\end{smallmatrix}\Bigr]\big(\z(\lambda)\big)^{\nu\sigma_3},&\lambda\in\mathbb{D}_{\frac{1}{4}}(1):\ \textnormal{arg}\,(\lambda-1)\in(\frac{\pi}{2},\pi)\smallskip\\
	\e^{-\im\sqrt{t}\,\sigma_3}\e^{\im\frac{\pi}{2}\alpha\sigma_3}\Bigl[\begin{smallmatrix} 0 & \e^{\im\frac{5\pi}{2}\nu}\\ -\e^{-\im\frac{7\pi}{2}\nu} & 0\end{smallmatrix}\Bigr]\big(\z(\lambda)\big)^{\nu\sigma_3},&\lambda\in\mathbb{D}_{\frac{1}{4}}(1):\ \textnormal{arg}\,(\lambda-1)\in(2\pi,\frac{5\pi}{2})\smallskip\\
	\e^{\im\sqrt{t}\,\sigma_3}\e^{-\im\frac{\pi}{2}\alpha\sigma_3}\Bigl[\begin{smallmatrix}\e^{-\im\frac{5\pi}{2}\nu} & 0\\ 0 & \e^{\im\frac{3\pi}{2}\nu}\end{smallmatrix}\Bigr]\big(\z(\lambda)\big)^{\nu\sigma_3},&\lambda\in\mathbb{D}_{\frac{1}{4}}(1):\ \textnormal{arg}\,(\lambda-1)\in(\pi,2\pi)\end{cases}
\end{equation}
are both analytic at $\lambda=1$ (note that $\z^{\nu}$ is defined with a cut on the positive imaginary $\z$-axis, see \eqref{OP:5} below). Once we recall RHP \ref{confluprob} it is easy to verify that \eqref{OP:7} has all the properties required in RHP \ref{problem:3}.
\begin{rem} Using \eqref{OP:9} we derive the following Taylor expansion of ${\bf E}^{(1)}(\lambda)$ near $\lambda=1$,
\begin{align*}
	{\bf E}^{(1)}(\lambda)=&\,t^{-\frac{1}{4}\sigma_3}\begin{bmatrix}1 & 0\\ -2\im\nu & 1\end{bmatrix}\frac{1}{\sqrt{2}}\begin{bmatrix}-\im & -1\\ 1 & \im\end{bmatrix}4^{\nu\sigma_3}\e^{-\im\frac{\pi}{2}(\nu+1)\sigma_3}\e^{\im\sqrt{t}\,\sigma_3}\e^{-\im\frac{\pi}{2}\alpha\sigma_3}t^{\frac{\nu}{2}\sigma_3}\e^{-\im\frac{\pi}{2}\nu}\\
	&\,\times\left\{\mathbb{I}+\frac{1}{4}(\lambda-1)\begin{bmatrix}\nu & -\im\,16^{-\nu}\e^{\im\pi(\nu+\alpha)}t^{-\nu}\e^{-2\im\sqrt{t}}\\ \im\,16^{\nu}\e^{-\im\pi(\nu+\alpha)}t^{\nu}\e^{2\im\sqrt{t}} & -\nu\end{bmatrix}+\mathcal{O}\big((\lambda-1)^2\big)\right\},\ \ \lambda\rightarrow 1.
\end{align*}
\end{rem}
This concludes the local analysis and we now compare \eqref{OP:1}, \eqref{OP:3} and \eqref{OP:7} to the function ${\bf S}(\lambda)$. Introduce
\begin{equation}\label{OP:10}
	{\bf R}(\lambda)={\bf S}(\lambda)\begin{cases}\big({\bf P}^{(0)}(\lambda)\big)^{-1},&\lambda\in\mathbb{D}_r(0)\\ \big({\bf P}^{(1)}(\lambda)\big)^{-1},&\lambda\in\mathbb{D}_r(1)\\ \big({\bf P}^{(\infty)}(\lambda)\big)^{-1},&\lambda\notin\big(\mathbb{D}_r(0)\cup\mathbb{D}_r(1)\big)\end{cases},
\end{equation}
where $0<r<\frac{1}{4}$ is kept fixed. In view of RHP \ref{Bess:S}, \ref{problem:1}, \ref{problem:2} and \ref{problem:3}, we derive the following RHP for the ratio function \eqref{OP:10}.
\begin{problem}\label{Bess:ratioRHP} Find ${\bf R}(\lambda)={\bf R}(\lambda;t,\alpha,v)\in\mathbb{C}^{2\times 2}$ such that
\begin{enumerate}
	\item ${\bf R}(\lambda)$ is analytic for $\lambda\in\mathbb{C}\backslash\Sigma_{\bf R}$ and assumes square-integrable boundary values ${\bf R}_{\pm}(\lambda)$ on the oriented contour
	\begin{equation*}
		\Sigma_{\bf R}=\partial\mathbb{D}_r(0)\cup\partial\mathbb{D}_r(1)\cup\bigg(\Big(\widehat{\gamma}^+\cup\widehat{\gamma}^-\cup\left\{\textnormal{arg}(\lambda-1)=\frac{\pi}{3},\frac{5\pi}{3}\right\}\Big)\cap\big\{\lambda\in\mathbb{C}:\ |\lambda|>r,\ |\lambda-1|>r\big\}\bigg)
	\end{equation*}
	shown in Figure \ref{figure14} below.
	\item We have ${\bf R}_+(\lambda)={\bf R}_-(\lambda){\bf G}_{\bf R}(\lambda),\lambda\in\Sigma_{\bf R}$ where
	\begin{equation*}
		{\bf G}_{\bf R}(\lambda)={\bf P}^{(0)}(\lambda)\big({\bf P}^{(\infty)}(\lambda)\big)^{-1},\ \ \lambda\in\partial\mathbb{D}_r(0);\ \ \ {\bf G}_{\bf R}(\lambda)={\bf P}^{(1)}(\lambda)\big({\bf P}^{(\infty)}(\lambda)\big)^{-1},\ \ \lambda\in\partial\mathbb{D}_r(1),
	\end{equation*}
	followed by
	\begin{equation*}
		{\bf G}_{\bf R}(\lambda)={\bf P}^{(\infty)}(\lambda)\begin{bmatrix}1 & 0\\ \e^{\sqrt{t}\,\phi(\lambda)}\e^{v\mp\im\pi\alpha} & 1\end{bmatrix}\big({\bf P}^{(\infty)}(\lambda)\big)^{-1},\ \ \ \lambda\in\widehat{\gamma}^{\pm},
	\end{equation*}
	and concluding with
	\begin{equation*}
		{\bf G}_{\bf R}(\lambda)={\bf P}^{(\infty)}(\lambda)\begin{bmatrix}1 & 0\\ \e^{\sqrt{t}\,\phi(\lambda)\mp\im\pi\alpha}&1\end{bmatrix}\big({\bf P}^{(\infty)}(\lambda)\big)^{-1},\ \ \lambda\in\left\{\textnormal{arg}(\lambda-1)=\frac{\pi}{3},\frac{5\pi}{3}\right\}\bigg\backslash\big(\mathbb{D}_r(0)\cup\mathbb{D}_r(1)\big).
	\end{equation*}
	\begin{figure}[tbh]
	\begin{center}
\resizebox{0.45\textwidth}{!}{\includegraphics{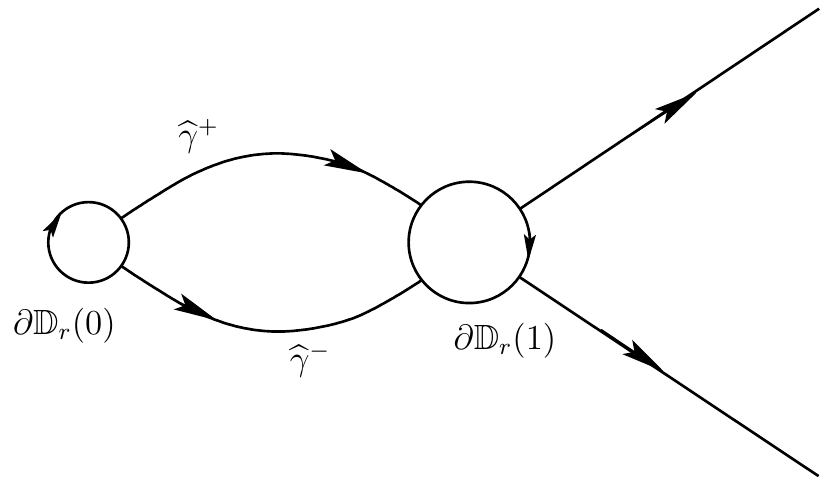}}
\caption{The oriented jump contours for the ratio function ${\bf R}(\lambda;t,\alpha,v)$ in the complex $\lambda$-plane.}
\label{figure14}
\end{center}
\end{figure}
	\item As $\lambda\rightarrow\infty$, we have ${\bf R}(\lambda)\rightarrow\mathbb{I}$.
\end{enumerate}
\end{problem}
Upon return to RHP \ref{problem:2} and RHP \ref{problem:3} we see that
\begin{equation*}
	\widehat{{\bf G}}_{\bf R}(\lambda;t,\alpha,v)=t^{\frac{1}{4}\sigma_3}{\bf G}_{\bf R}(\lambda;t,\alpha,v)t^{-\frac{1}{4}\sigma_3},\ \ \lambda\in\Sigma_{\bf R}
\end{equation*}
satisfies the following small norm estimate
\begin{prop} For any fixed $\alpha>-1,v\in[0,+\infty)$ there exist $t_0=t_0(\alpha,v)>0$ and $c=c(v)>0$ such that
\begin{equation*}
	\|\widehat{{\bf G}}_{\bf R}(\cdot;t,\alpha,v)-\mathbb{I}\|_{L^2\cap L^{\infty}(\partial\mathbb{D}_r(0)\cup\partial\mathbb{D}_r(1))}\leq\frac{c}{\sqrt{t}},\ \ \ \ \ \forall\,t\geq t_0.
\end{equation*}
\end{prop}
Moreover, recalling \eqref{Bess:inf1} and \eqref{Bess:inf2} together with \eqref{OP:1} we see that $\widehat{{\bf G}}_{\bf R}(\lambda;t,\alpha,v)$ is exponentially close to the identity matrix (as $t\rightarrow+\infty$ and $v\in[0,+\infty),\alpha\in(-1,+\infty)$ are fixed) on $\widehat{\gamma}^{\pm}$ and the two contours extending to infinity, see Figure \ref{figure14}. Thus all together, cf. \cite{DZ}, we have 
\begin{theo}\label{DZBess} Given $\alpha>-1,v\in[0,+\infty)$ there exist $t_0=t_0(\alpha,v)$ and $c=c(\alpha,v)$ positive such that RHP \ref{Bess:ratioRHP} is uniquely solvable in $L^2(\Sigma_{\bf R})$ for all $t\geq t_0$. Its solution ${\bf R}(\lambda)={\bf R}(\lambda;t,\alpha,v)$ can be computed iteratively via the integral equation
\begin{equation*}
	\widehat{{\bf R}}(\lambda)=\mathbb{I}+\frac{1}{2\pi\im}\int_{\Sigma_{\bf R}}\widehat{{\bf R}}_-(w)\big(\widehat{{\bf G}}_{\bf R}(w)-\mathbb{I}\big)\frac{\d w}{w-\lambda},\ \ \lambda\in\mathbb{C}\backslash\Sigma_{\bf R};\ \ \ \ \widehat{{\bf R}}(\lambda;t,\alpha,v)=t^{\frac{1}{4}\sigma_3}{\bf R}(\lambda;t,\alpha,v)t^{-\frac{1}{4}\sigma_3}
\end{equation*}
using the estimate
\begin{equation*}
	\|\widehat{{\bf R}}_-(\cdot;t,\alpha,v)-\mathbb{I}\|_{L^2(\Sigma_{\bf R})}\leq \frac{c}{\sqrt{t}},\ \ \ \ \forall\,t\geq t_0.
\end{equation*}
\end{theo}
At this point we return to \eqref{Bess:f1} and \eqref{Bess:f2}.
\subsection{Extraction of asymptotics and proof of expansion \eqref{JME:33}} We split this subsection into several parts.
\subsubsection{Preliminary expansions} Recall the explicit and invertible transformation sequence
\begin{equation*}
	{\bf Y}(\lambda)\mapsto {\bf X}(\lambda)\mapsto{\bf T}(\lambda)\mapsto{\bf S}(\lambda)\mapsto{\bf R}(\lambda)
\end{equation*}
which leads us to the exact identity
\begin{align}
	{\bf Y}_1=\begin{bmatrix}1 & 0\\ -b_1(\alpha)& 1\end{bmatrix}&\,\Bigg\{\begin{bmatrix}2\nu^2&-\frac{2\im\nu}{\sqrt{t}}\\ \frac{2\im}{3}\nu(1-4\nu^2)\sqrt{t} & -2\nu^2\end{bmatrix}+\frac{1}{t}\begin{bmatrix}a_2(\alpha)&-a_1(\alpha)\\ b_1(\alpha)a_2(\alpha)-b_3(\alpha) & -a_2(\alpha)\end{bmatrix}\nonumber\\
	&+\frac{\im}{2\pi}\int_{\Sigma_{\bf R}}{\bf R}_-(w)\big({\bf G}_{\bf R}(w)-\mathbb{I}\big)\,\d w\Bigg\}\begin{bmatrix}1 & 0\\ b_1(\alpha)&1\end{bmatrix},\label{Y1exact}
\end{align}
where the coefficients $a_k(\alpha),b_k(\alpha)$ are defined in RHP \ref{Bessbetter} below. Define
\begin{equation*}
	{\bf N}=\frac{\im}{2\pi}\int_{\Sigma_{\bf R}}\big({\bf G}_{\bf R}(w)-\mathbb{I}\big)\,\d w,\ \ \ \ {\bf N}=\big(N^{jk}\big)_{j,k=1}^2
\end{equation*}
so that from a residue computation (using RHP \ref{problem:2} and \ref{problem:3})
\begin{prop}\label{P:1} As $t\rightarrow+\infty$,
\begin{eqnarray*}
	N^{11}&=&\frac{2\im\nu}{\sqrt{t}}a_1(\alpha)-\frac{\im\nu}{\sqrt{t}}\sin\big(\eta(t,\alpha,v)\big)+\frac{2\nu^2}{\sqrt{t}}\big(\cos\big(\eta(t,\alpha,v)\big)-\im\nu)+\mathcal{O}\big(t^{-1}\big)\\
	N^{12}&=&\frac{1}{t}a_1(\alpha)-\frac{\nu^2}{t}-\frac{\im\nu}{t}\cos\big(\eta(t,\alpha,v)\big)+\mathcal{O}\big(t^{-\frac{3}{2}}\big).
\end{eqnarray*}
All error terms are uniform with respect to $(\alpha,v)$ chosen from compact subsets of $(-1,+\infty)\times[0,+\infty)$ and
\begin{equation*}
	\eta(t,\alpha,v)=2\sqrt{t}-\frac{v}{2\pi}\ln(16t)-\pi\alpha+2\,\textnormal{arg}\,\Gamma\left(\frac{\im v}{2\pi}\right).
\end{equation*}
\end{prop}
Next we also require that from Theorem \ref{DZBess}, RHP \ref{problem:2} and \ref{problem:3}, as $t\rightarrow+\infty$,
\begin{equation*}
	\int_{\Sigma_{\bf R}}\big(\widehat{{\bf R}}_-(w)-\mathbb{I}\big)\big(\widehat{{\bf G}}_{\bf R}(w)-\mathbb{I}\big)\,\d w=\mathcal{O}\big(t^{-1}\big),
\end{equation*}
i.e. for
\begin{equation*}
	{\bf Q}=\frac{\im}{2\pi}\int_{\Sigma_{\bf R}}\big({\bf R}_-(w)-\mathbb{I}\big)\big({\bf G}_{\bf R}(w)-\mathbb{I}\big)\,\d w,\ \ \ \ {\bf Q}=\big(Q^{jk}\big)_{j,k=1}^2
\end{equation*}
we find in turn
\begin{prop}\label{P:2} As $t\rightarrow+\infty$ with fixed $(\alpha,v)\in(-1,+\infty)\times[0,+\infty)$,
\begin{equation*}
	Q^{11}=\mathcal{O}\big(t^{-1}\big),\ \ \ Q^{12}=\mathcal{O}\big(t^{-\frac{3}{2}}\big).
\end{equation*}
\end{prop}
At this point we combine \eqref{Y1exact} and Propositions \ref{P:1}, \ref{P:2},
\begin{cor} As $t\rightarrow+\infty$,
\begin{eqnarray*}
	Y_1^{11}&=&2\nu^2-\frac{\im\nu}{\sqrt{t}}\big(1+\sin\big(\eta(t,\alpha,v)\big)\big)+\frac{2\nu^2}{\sqrt{t}}\big(\cos\big(\eta(t,\alpha,v)\big)-\im\nu\big)+\mathcal{O}\left(t^{-1}\right)\\
	Y_1^{12}&=&-\frac{2\im\nu}{\sqrt{t}}-\frac{\nu^2}{t}-\frac{\im\nu}{t}\cos\big(\eta(t,\alpha,v)\big)+\mathcal{O}\big(t^{-\frac{3}{2}}\big);
\end{eqnarray*}
and all error terms are uniform with respect to $(\alpha,v)$ chosen from compact subsets of $(-1,+\infty)\times[0,+\infty)$.
\end{cor}
Now back in \eqref{Bess:f1} and \eqref{Bess:f2},
\begin{cor}\label{pqcor} As $t\rightarrow+\infty$ with fixed $(v,\alpha)\in[0,+\infty)\times(-1,+\infty)$,
\begin{equation*}
	q^2(t,\alpha;\gamma)=\frac{2\im\nu}{\sqrt{t}}\Big(1-\sin\big(\eta(t,\alpha,v)\big)\Big)+\mathcal{O}\left(t^{-1}\right),\ \ \ p^2(t,\alpha;\gamma)=2\im\nu\sqrt{t}\,\Big(1+\sin\big(\eta(t,\alpha,v)\big)\Big)+\mathcal{O}(1)
\end{equation*}
and
\begin{equation*}
	\mathcal{H}_H\big(q(t,\alpha,\gamma),p(t,\alpha,\gamma),t,\alpha\big)=-\frac{v}{2\pi}\frac{1}{\sqrt{t}}+\frac{v^2}{8\pi^2}\frac{1}{t}-\frac{v}{4\pi t}\cos\big(\eta(t,\alpha,v)\big)+\mathcal{O}\big(t^{-\frac{3}{2}}\big)
\end{equation*}
\end{cor}
The last result allows us already to determine all $t$-dependent terms in Theorem \ref{theo:3}, expansion \eqref{JME:33}, indeed through Proposition \ref{Di:1} we find
\begin{cor}\label{1stresult} As $t\rightarrow+\infty$,
\begin{equation*}
	\ln F_H(t,\alpha;\gamma)=-\frac{v}{\pi}\sqrt{t}+\frac{v^2}{8\pi^2}\ln t +D(\alpha,v)+\mathcal{O}\big(t^{-\frac{1}{2}}\big),
\end{equation*}
where $D(v,\alpha)$ is $t$-independent and the error term is uniform with respect to $(\alpha,v)$ chosen from compact subsets of $(-1,+\infty)\times[0,+\infty)$.
\end{cor}
In order to determine $D(\alpha,v)$ we use \eqref{JME:26} and first derive the following exact identity (recall RHP \ref{Bess:X}, \eqref{Bess:T2}, \eqref{Bess:T3}, \eqref{OP:3}, \eqref{OP:10} and \eqref{Bessundress}), for $\lambda\in\mathbb{D}_r(0)$,
\begin{eqnarray*}
	\widehat{{\bf X}}(\lambda)&=&{\bf R}(\lambda){\bf E}^{(0)}(\lambda)\begin{bmatrix}1 & 0\\ b_1(\alpha) & 1\end{bmatrix}\Psi_{\alpha}(\lambda t)\e^{-\im\frac{\pi}{2}\alpha\sigma_3}\begin{bmatrix}1 & -\frac{\im}{2}\frac{1}{\sin\pi\alpha}\smallskip \\ 0 & 1\end{bmatrix}(-\lambda)^{-\frac{\alpha}{2}\sigma_3}\e^{\frac{v}{2}\sigma_3},\ \ \alpha\notin\mathbb{Z};\\
	\widehat{{\bf X}}(\lambda)&=&{\bf R}(\lambda){\bf E}^{(0)}(\lambda)\begin{bmatrix}1 & 0\\ b_1(\alpha) & 1\end{bmatrix}\Psi_{\alpha}(\lambda t)\e^{-\im\frac{\pi}{2}\alpha\sigma_3}\begin{bmatrix}1 & \frac{\e^{\im\pi\alpha}}{2\pi\im}\ln(-\lambda)\\ 0 & 1\end{bmatrix}(-\lambda)^{-\frac{\alpha}{2}\sigma_3}\e^{\frac{v}{2}\sigma_3},\ \ \alpha\in\mathbb{Z}.
\end{eqnarray*}
However, keeping in mind Remark \ref{locrem} below, we have for $\alpha>-1$,
\begin{eqnarray}
	\widehat{{\bf X}}(0)&=&{\bf R}(0){\bf E}^{(0)}(0)\begin{bmatrix}1 & 0\\ b_1(\alpha) & 1\end{bmatrix}\widehat{\Psi}(0;\alpha)t^{\frac{\alpha}{2}\sigma_3}\e^{\frac{v}{2}\sigma_3},\ \ \ \alpha\neq 0;\label{ev:1}\\
	\widehat{{\bf X}}(0)&=&{\bf R}(0){\bf E}^{(0)}(0)\begin{bmatrix}1 & 0\\ b_1(0) & 1\end{bmatrix}\widehat{\Psi}(0;0)\begin{bmatrix}1 & -\frac{\ln t}{2\pi\im}\smallskip\\ 0 & 1\end{bmatrix}\e^{\frac{v}{2}\sigma_3},\ \ \ \ \alpha=0.\label{ev:2}
\end{eqnarray}
But with Theorem \ref{DZBess},
\begin{equation*}
	\widehat{{\bf R}}(0)=\mathbb{I}+\frac{1}{2\pi\im}\int_{\Sigma_{\bf R}}\big(\widehat{{\bf G}}_{\bf R}(w)-\mathbb{I}\big)\frac{\d w}{w}+\mathcal{O}\big(t^{-1}\big),
\end{equation*}
and by residue computation,
\begin{prop} As $t\rightarrow+\infty$, for any fixed $\alpha>-1$ and $v\in[0,+\infty)$,
\begin{eqnarray*}
	\widehat{R}^{11}(0)&=&1+\frac{2\im\nu}{\sqrt{t}}a_1(1-4\nu^2)+\frac{\im\nu}{\sqrt{t}}\big(\sin\eta(t,\alpha,v)+2\im\nu(\cos\eta(t,\alpha,v)-\im\nu)\big)+\mathcal{O}\big(t^{-1}\big),\\
	 \widehat{R}^{12}(0)&=&\frac{\nu^2}{\sqrt{t}}(1-4a_1)+\frac{\im\nu}{\sqrt{t}}\cos\big(\eta(t,\alpha,v)\big)+\mathcal{O}\big(t^{-1}\big)
\end{eqnarray*}
\end{prop}
Thus back in \eqref{ev:1} and \eqref{ev:2}
\begin{cor}\label{Xo:1} As $t\rightarrow+\infty$ with fixed $\alpha>-1,v\in[0,+\infty)$, 
\begin{equation*}
	\widehat{X}^{11}(0)=\sqrt{\pi}\,\e^{-\im\frac{\pi}{4}}\left(1-\frac{\im\nu}{\sqrt{t}}\Big(1-\sin\big(\eta(t,\alpha,v)\big)\Big)-\frac{2\im\nu}{\sqrt{t}}\alpha+\mathcal{O}\big(t^{-1}\big)\right)\frac{2^{-\alpha}t^{\frac{\alpha}{2}}\e^{\frac{v}{2}}}{\Gamma(1+\alpha)},
\end{equation*}
and for fixed $\alpha>-1,\alpha\neq 0,v\in[0,+\infty)$,
\begin{equation*}
	\widehat{X}^{12}(0)=\sqrt{\pi}\,\e^{-\im\frac{\pi}{4}}\left(1-\frac{\im\nu}{\sqrt{t}}\Big(1-\sin\big(\eta(t,\alpha,v)\big)\Big)+\frac{2\im\nu}{\sqrt{t}}\alpha+\mathcal{O}\big(t^{-1}\big)\right)\frac{\Gamma(\alpha)}{2\pi\im}\,2^{\alpha}t^{-\frac{\alpha}{2}}\e^{-\frac{v}{2}}.
	%
\end{equation*}
\end{cor}
The last Corollary allows us to compute parts of $L(t,\alpha;\gamma)$, see \eqref{cool:3} and \eqref{cool:4}. For the remaining parts we evidently require the small $t$ behavior of $\widehat{{\bf X}}(0)$.
\subsubsection{Small $t$-behavior of ${\bf Y}(\lambda;t,\alpha,\gamma)$ for $\alpha>0$.} We return to RHP \ref{BessRHP} and use the power series expansion 
\begin{equation}\label{Besssmall}
	J_{\alpha}(z)=\left(\frac{z}{2}\right)^{\alpha}\left\{\frac{1}{\Gamma(1+\alpha)}+\mathcal{O}\big(z^2\big)\right\},\ \ \ z\rightarrow 0,\ \ z\notin(-\infty,0]
\end{equation}
to obtain the following small norm estimate.
\begin{prop} For any fixed $\alpha>0$ there exist $t_0=t_0(\alpha)>0$ and $c=c(\alpha)>0$ such that
\begin{equation*}
	\|{\bf G}_{\bf Y}(\cdot;t,\alpha,v)-\mathbb{I}\|_{L^2\cap L^{\infty}(0,1)}\leq c\,t^{\alpha},\ \ \ \ \forall\,t\leq t_0,\ \ \ \gamma\in[0,1].
\end{equation*}
\end{prop}
Hence, by general theory \cite{DZ}, the initial RHP \ref{BessRHP} is solvable as $t\downarrow 0$ and $\alpha>0$
\begin{theo} Given $\alpha>0$ there exist $t_0=t_0(\alpha)$ and $c=c(\alpha)$ positive such that RHP \ref{BessRHP} is solvable in $L^2(0,1)$ for all $t\leq t_0$. Its solution ${\bf Y}(\lambda)={\bf Y}(\lambda;t,\alpha,v)$ can be computed iteratively via the integral equation
\begin{equation*}
	{\bf Y}(\lambda)=\mathbb{I}+\frac{1}{2\pi\im}\int_0^1{\bf Y}_-(w)\big({\bf G}_{\bf Y}(w)-\mathbb{I}\big)\frac{\d w}{w-\lambda},\ \ \lambda\in\mathbb{C}\backslash[0,1]
\end{equation*}
using the estimate
\begin{equation*}
	\|{\bf Y}_-(\cdot;t,\alpha,v)-\mathbb{I}\|_{L^2(0,1)}\leq c\,t^{\alpha},\ \ \ \forall\,\ t\leq t_0,\ \ \gamma\in[0,1].
\end{equation*}
\end{theo}
Combining this last result with RHPs \ref{Bess:X} and \ref{Bessbetter} we find in turn
\begin{cor}\label{Xo:2} As $t\downarrow 0$ with fixed $\alpha>0$,
\begin{equation*}
	\widehat{X}^{11}(0)=\sqrt{\pi}\,\e^{-\im\frac{\pi}{4}}\left(1+\mathcal{O}\left(t^{\min\{\alpha,1\}}\right)\right)\frac{2^{-\alpha}t^{\frac{\alpha}{2}}}{\Gamma(1+\alpha)},\ \ 
	\widehat{X}^{12}(0)=\sqrt{\pi}\,\e^{-\im\frac{\pi}{4}}\left(1+\mathcal{O}\left(t^{\min\{\alpha,1\}}\right)\right)\frac{\Gamma(\alpha)}{2\pi\im}\,2^{\alpha}t^{-\frac{\alpha}{2}},
\end{equation*}
uniformly for any $\gamma\in[0,1]$.
\end{cor}
\subsubsection{Small $t$-behavior for $-1<\alpha<0$.} For this parameter regime we again return to RHP \ref{BessRHP} but apply first the following transformation
\begin{equation*}
	{\bf W}(\lambda)=t^{-\frac{\alpha}{2}\sigma_3}\big(\widehat{\Psi}(0;\alpha)\big)^{-1}{\bf Y}(\lambda)\begin{cases}\widehat{\Psi}(\lambda t;\alpha)t^{\frac{\alpha}{2}\sigma_3}{\bf M}(\lambda),&\lambda\in\mathbb{D}_r(\frac{1}{2})\\ \widehat{\Psi}(0;\alpha)t^{\frac{\alpha}{2}\sigma_3}{\bf M}(\lambda),&\lambda\notin\mathbb{D}_r(\frac{1}{2})\end{cases},\ \ \ \ \ \ \ \frac{1}{2}<r<1\ \ \textnormal{fixed}.
\end{equation*}
The entire function $\widehat{\Psi}(\z;\alpha)$ is defined in RHP \ref{Bessbetter} below and we have introduced
\begin{equation}\label{m:1}
	{\bf M}(\lambda)=\begin{bmatrix}1 & \frac{\gamma}{2\pi\im}\int_0^1\frac{w^{\alpha}}{w-\lambda}\,\d w\smallskip\\ 0 & 1\end{bmatrix},\ \ \ \ \lambda\in\mathbb{C}\backslash[0,1],\ \ \ \ \gamma\in[0,1],\ \ -1<\alpha<0.
\end{equation}
\begin{problem}\label{Mfunction} The function ${\bf M}(\lambda)={\bf M}(\lambda;\alpha,\gamma)\in\mathbb{C}^{2\times 2}$ defined in \eqref{m:1} has the following properties
\begin{enumerate}
	\item ${\bf M}(\lambda)$ is analytic for $\lambda\in\mathbb{C}\backslash[0,1]$.
	\item Orienting the interval $[0,1]\subset\mathbb{R}$ from left to right we have
	\begin{equation*}
		{\bf M}_+(\lambda)={\bf M}_-(\lambda)\begin{bmatrix} 1 & \gamma \lambda^{\alpha} \\ 0 & 1\end{bmatrix}=\begin{bmatrix} 1 & \gamma\lambda^{\alpha}\\ 0 & 1\end{bmatrix}{\bf M}_-(\lambda),\ \ \ \ \lambda\in(0,1).
	\end{equation*}
	\item ${\bf M}(\lambda)$ is square integrable on $[0,1]\subset\mathbb{R}$, in more detail for $\alpha\in(-1,0)$,
	\begin{equation*}
		{\bf M}(\lambda)=\widehat{{\bf M}}(\lambda)\begin{bmatrix}1 & \frac{\im\gamma}{2}\frac{(-\lambda)^{\alpha}}{\sin\pi\alpha}\\ 0 & 1\end{bmatrix},\ \ \ \lambda\rightarrow 0,\ \lambda\notin[0,+\infty),
	\end{equation*}
	where $z^{\alpha}$ is defined with its principal branch. Here, $\widehat{{\bf M}}(\lambda)$ is analytic at $\lambda=0$,
	\begin{equation*}
		\widehat{{\bf M}}(\lambda)=\mathbb{I}+\frac{\gamma}{2\pi\im\alpha}\sigma_++\mathcal{O}\big(\lambda\big),\ \ \ \lambda\rightarrow 0.	\end{equation*}
	\item As $\lambda\rightarrow\infty$,
	\begin{equation*}
		{\bf M}(\lambda)=\mathbb{I}+\frac{\im\gamma}{2\pi}\frac{\sigma_+}{\alpha+1}\frac{1}{\lambda}+\mathcal{O}\left(\lambda^{-2}\right).
	\end{equation*}
\end{enumerate}
\end{problem}
At this point we recall RHP \ref{BessRHP} and make use of \eqref{facto} in order to derive the following RHP for ${\bf Z}(\lambda)$
\begin{problem}\label{ZRHP} Find ${\bf W}(\lambda)={\bf W}(\lambda;t,\alpha,v)\in\mathbb{C}^{2\times 2}$ such that
\begin{enumerate}
	\item ${\bf W}(\lambda)$ is analytic for $\lambda\in\mathbb{C}\backslash\partial\mathbb{D}_r(\frac{1}{2})$ and we orient $\partial\mathbb{D}_r(\frac{1}{2})$ clockwise.
	\item By construction, compare \eqref{facto}, ${\bf W}(\lambda)$ has no jump on the interval $(0,1)\subset\mathbb{R}$. Instead we observe that
	\begin{equation*}
		{\bf W}_+(\lambda)={\bf W}_-(\lambda){\bf M}^{-1}(\lambda)t^{-\frac{\alpha}{2}\sigma_3}\big(\widehat{{\Psi}}(\lambda t;\alpha)\big)^{-1}\widehat{\Psi}(0;\alpha)t^{\frac{\alpha}{2}\sigma_3}{\bf M}(\lambda),\ \ \ \lambda\in\partial\mathbb{D}_r\left(\frac{1}{2}\right).
	\end{equation*}
	\item ${\bf W}(\lambda)$ is bounded at $\lambda=0$ and $\lambda=1$.
	\item As $\lambda\rightarrow\infty$ we have
	\begin{equation*}
		{\bf W}(\lambda)=\mathbb{I}+\mathcal{O}\big(\lambda^{-1}\big).
	\end{equation*}
\end{enumerate}
\end{problem}
Since ${\bf M}(\lambda)$ is $t$-independent and $\widehat{\Psi}(\lambda t;\alpha)$ analytic at $\lambda=0$, we obtain at once
\begin{prop} For any fixed $\alpha\in(-1,0)$ there exist $t_0=t_0(\alpha)>0$ and $c=c(\alpha)>0$ such that
\begin{equation*}
	\|{\bf G}_{{\bf W}}(\cdot;t,\alpha,v)-\mathbb{I}\|_{L^2\cap L^{\infty}(\partial\mathbb{D}_r(\frac{1}{2}))}\leq c\,t^{1+\alpha},\ \ \ \ \ \forall\, t\leq t_0,\ \ \ \gamma\in[0,1].
\end{equation*}
\end{prop}
In short, the transformed problem \ref{ZRHP} is solvable as $t\downarrow 0$ and $\alpha\in(-1,1)$, cf. \cite{DZ}.
\begin{theo} Given $\alpha\in(-1,0)$ there exist $t_0=t_0(\alpha)$ and $c=c(\alpha)$ positive such that RHP \ref{ZRHP} is solvable in $L^2(0,1)$ for all $t\leq t_0$. Its solution ${\bf W}(\lambda)={\bf W}(\lambda;t,\alpha,v)$ can be computed iteratively via the integral equation
\begin{equation*}
	{\bf W}(\lambda)=\mathbb{I}+\frac{1}{2\pi\im}\oint_{\partial\mathbb{D}_r(\frac{1}{2})}{\bf W}_-(w)\big({\bf G}_{\bf W}(w)-\mathbb{I}\big)\frac{\d w}{w-\lambda},\ \ \lambda\in\mathbb{C}\backslash\partial\mathbb{D}_r\left(\frac{1}{2}\right)		
\end{equation*}
using the estimate
\begin{equation*}
	\|{\bf W}_-(\cdot;t,\alpha,v)-\mathbb{I}\|_{L^2(\partial\mathbb{D}_r(\frac{1}{2}))}\leq c\,t^{1+\alpha},\ \ \ \ \forall\,\ t\leq t_0,\ \ \gamma\in[0,1].
\end{equation*}
\end{theo}
This last result allows us to derive small $t$-expansions for $\widehat{{\bf X}}(0)$, compare RHP \ref{Bess:X}, \ref{Bessbetter} and \ref{Mfunction}.
\begin{cor} As $t\downarrow 0$ with fixed $\alpha\in(-1,0)$,
\begin{eqnarray*}
	\widehat{X}^{11}(0)&=&\sqrt{\pi}\,\e^{-\im\frac{\pi}{4}}\left(1+\mathcal{O}\left(t^{1+\alpha}\right)\right)\frac{2^{-\alpha}t^{\frac{\alpha}{2}}}{\Gamma(1+\alpha)},\\
	\widehat{X}^{12}(0)&=&\sqrt{\pi}\,\e^{-\im\frac{\pi}{4}}\left(1+\mathcal{O}\left(t^{\min\{1+\alpha,-\alpha\}}\right)\right)\left(-\frac{\gamma}{2\pi\im\alpha}\right)\frac{2^{-\alpha}t^{\frac{\alpha}{2}}}{\Gamma(1+\alpha)}.
\end{eqnarray*}
\end{cor}
\subsubsection{Derivation of \eqref{JME:33}.} We know from Corollary \ref{pqcor} that as $t\rightarrow+\infty$ and $(v,\alpha)\in[0,+\infty)\times(-1,+\infty)$ are fixed,
\begin{equation}\label{step1}
	2t\,\mathcal{H}_H(q,p,t,\alpha)=-\frac{v}{\pi}\sqrt{t}+\frac{v^2}{4\pi^2}-\frac{v}{2\pi}\cos\eta(t,\alpha,v)+\mathcal{O}\big(t^{-\frac{1}{2}}\big).
\end{equation}
On the other hand, via the same Corollary \ref{pqcor} and through standard manipulations with trigonometric functions,
\begin{cor} As $t\rightarrow+\infty$ with fixed $(v,\alpha)\in[0,+\infty)\times(-1,+\infty)$,
\begin{equation*}
	pq_{\gamma}=\frac{1}{2\pi}\frac{\d}{\d\gamma}\left\{v\cos\big(\eta(t,\alpha,v)\big)+\frac{v^2}{4\pi}\ln(16t)\right\}-\frac{v}{\pi}\frac{\d}{\d\gamma}\textnormal{arg}\,\Gamma\left(\frac{\im v}{2\pi}\right)+\mathcal{O}\left(t^{-\frac{1}{2}}\ln t\right).
\end{equation*}
\end{cor}
Next, through Corollaries \ref{Xo:1} and \ref{Xo:2},
\begin{cor} As $t\rightarrow+\infty$ with fixed $(v,\alpha)\in[0,+\infty)\times(-1,+\infty)$,
\begin{eqnarray*}
	L(t,\alpha;\gamma)&=&-\frac{\alpha}{2}v+\mathcal{O}\big(t^{-\frac{1}{2}}\big),\ \ \alpha>0;\\
	L(t,\alpha;\gamma)&=&-\frac{\alpha^2}{2}\ln t-\frac{\alpha}{2}v-\frac{\alpha}{2}\ln(-\gamma)+\alpha\ln\big(2^{\alpha}\Gamma(1+\alpha)\big)+\mathcal{O}\big(t^{-\frac{1}{2}}\big),\ -1<\alpha<0.
\end{eqnarray*}
\end{cor}
At this point we can start to determine $I_H$ through Corollary \ref{logic},
\begin{cor} As $t\rightarrow+\infty$ with fixed $(v,\alpha)\in[0,+\infty)\times(-1,+\infty)$,
\begin{equation*}
	I_H(t,\alpha;\gamma)=\frac{v}{2\pi}\cos\eta(t,\alpha,v)+\frac{v^2}{8\pi^2}\ln(16t)-\frac{1}{\pi}\int_0^{\gamma}v(\gamma')\frac{\d}{\d\gamma'}\textnormal{arg}\,\Gamma\left(\frac{\im v(\gamma')}{2\pi}\right)\,\d\gamma'+G(t)+\mathcal{O}\big(t^{-\frac{1}{2}}\big),
\end{equation*}
for $\alpha\geq 0$, where $G(t)$ is $(\alpha,v)$-independent. On the other hand, for $-1<\alpha<0$, as $t\rightarrow+\infty$,
\begin{align*}
	I_H(t,\alpha;\gamma)=&\,\frac{v}{2\pi}\cos\eta(t,\alpha,v)+\frac{v^2}{8\pi^2}\ln(16t)-\frac{1}{\pi}\int_0^{\gamma}v(\gamma')\frac{\d}{\d\gamma'}\textnormal{arg}\,\Gamma\left(\frac{\im v(\gamma')}{2\pi}\right)\,\d\gamma'-\frac{\alpha}{2}\ln(-\gamma)-\frac{\alpha^2}{2}\ln t\\
	&+\alpha\ln\big(2^{\alpha}\Gamma(1+\alpha)\big)+H(t)+\mathcal{O}\big(t^{-\frac{1}{2}}\big),
\end{align*}
where $H(t)$ is $(\alpha,v)$-independent.
\end{cor}
Combining the last two Corollaries with \eqref{step1} and Corollary \ref{1stresult} (using also $F_H(t,\alpha;0)=1$ again) we have thus back in \eqref{JME:26},
\begin{prop} As $t\rightarrow+\infty$,
\begin{equation*}
	\ln F_H(t,\alpha;\gamma)=-\frac{v}{\pi}\sqrt{t}+\frac{v^2}{8\pi^2}\ln(16t)+\frac{\alpha}{2}v+\frac{v^2}{4\pi^2}-\frac{1}{\pi}\int_0^{\gamma}v(\gamma')\frac{\d}{\d\gamma'}\textnormal{arg}\,\Gamma\left(\frac{\im v(\gamma')}{2\pi}\right)\,\d\gamma'+\mathcal{O}\big(t^{-\frac{1}{2}}\big),
\end{equation*}
uniformly for fixed $(\alpha,v)\in(-1,+\infty)\times[0,+\infty)$.
\end{prop}
This expansion is exactly equal to \eqref{JME:33} once we recall again standard properties of the Barnes-G function. 

\begin{appendix}
\section{Differential equations}\label{AppA}
The kernel \eqref{JME:23} identifies the corresponding integral operator as integrable in the sense of \cite{IIKS} and a simple rescaling argument identifies RHP \ref{BessRHP} subsequently as the underlying RHP. Moreover, following \cite{IIKS,DIZ} (see also \cite{B2}) we find at once the following differential identity
\begin{prop}\label{Di:1} For any fixed $\alpha>-1$ and $\gamma\in[0,1]$ we have
\begin{equation*}
	\frac{\partial}{\partial t}\ln F_H(t,\alpha;\gamma)=\frac{1}{2}Y_1^{12}
\end{equation*}
in terms of the solution to RHP \ref{BessRHP}, see \eqref{BessAsy}.
\end{prop}
In order to characterize $(q,p)$ in \eqref{HH:2}, \eqref{JME:19} via RHP \ref{BessRHP} we apply the following standard argument: First return to RHP \ref{Bess:X} and note that all jumps in the same problem are $\lambda$- and $t$-independent. Hence the functions $\frac{\partial{\bf Z}}{\partial\lambda}{\bf Z}^{-1}$ and $\frac{\partial{\bf Z}}{\partial t}{\bf Z}^{-1}$ with
\begin{equation*}
	{\bf Z}(\lambda)=\begin{bmatrix}1 & 0\\ -\frac{1}{8}(3+4\alpha^2) & 1\end{bmatrix}{\bf X}(\lambda),\ \ \ \lambda\in\mathbb{C}\backslash\Sigma_{\bf X}
\end{equation*}
are meromorphic in $\lambda\in\mathbb{C}\backslash\{0,1\}$. In fact, using \eqref{BessT:1} together with \eqref{BessAsy} we find that
\begin{align}
	\frac{\partial{\bf Z}}{\partial\lambda}&{\bf Z}^{-1}=-\frac{t}{2}\begin{bmatrix}0 & 0\\ 1 &0\end{bmatrix}+\frac{1}{2\lambda}\begin{bmatrix}-tY_1^{12} & 1\\ 2tY_1^{11}+\alpha^2 & tY_1^{12}\end{bmatrix}\label{Z:1}\\
	&+\frac{1}{\lambda^2}\begin{bmatrix}-Y_1^{11}+\frac{1}{2}(Y_1^{12}\alpha^2-Y_1^{21})-\frac{t}{2}Y_2^{12}+\frac{t}{2}Y_1^{11}Y_1^{12} & -Y_1^{12}+Y_1^{11}+\frac{t}{2}(Y_1^{12})^2\smallskip\\
	-Y_1^{21}-Y_1^{11}\alpha^2-\frac{t}{2}(Y_2^{22}-Y_2^{11})+\frac{t}{2}(1-(Y_1^{11})^2) & Y_1^{11}-\frac{1}{2}(Y_1^{12}\alpha^2-Y_1^{21})+\frac{t}{2}Y_2^{12}-\frac{t}{2}Y_1^{11}Y_1^{12}\end{bmatrix}\nonumber\\
	&+\mathcal{O}\left(\lambda^{-3}\right),\ \ \ \ \lambda\rightarrow\infty.\nonumber
\end{align}
Likewise from RHP \ref{Bess:X}, condition (3),
we obtain as $\lambda\rightarrow 0$,
\begin{equation}\label{Z:2}
	\frac{\partial {\bf Z}}{\partial\lambda}{\bf Z}^{-1}\bigg|_{\alpha\neq 0}=\frac{\alpha}{2\lambda}\widehat{{\bf Z}}(0)\sigma_3\big(\widehat{{\bf Z}}(0)\big)^{-1}+\mathcal{O}(1);\ \ 
	 \frac{\partial{\bf Z}}{\partial\lambda}{\bf Z}^{-1}\bigg|_{\alpha=0}=-\frac{1-\gamma}{2\pi\im\lambda}\,\widehat{{\bf Z}}(0)\begin{bmatrix}0 & 1\\ 0 & 0\end{bmatrix}\big(\widehat{{\bf Z}}(0)\big)^{-1}+\mathcal{O}(1)
\end{equation}
where $\widehat{{\bf Z}}(\lambda)=\bigl[\begin{smallmatrix}1 & 0\\ -\frac{1}{8}(3+4\alpha^2) & 1\end{smallmatrix}\bigr]\widehat{{\bf X}}(\lambda)$.
Finally, from condition (4) in RHP \ref{Bess:X},
\begin{equation}\label{Z:4}
	\frac{\partial {\bf Z}}{\partial\lambda}{\bf Z}^{-1}=\frac{\gamma}{2\pi\im}\widehat{{\bf Z}}(1)\begin{bmatrix}-1 & -\e^{-\im\pi\alpha}\\ \e^{\im\pi\alpha} & 1\end{bmatrix}\big(\widehat{{\bf Z}}(1)\big)^{-1}\frac{1}{\lambda-1}+\mathcal{O}(1),\ \ \ \ \lambda\rightarrow 1.
\end{equation}
Combining \eqref{Z:1},\eqref{Z:2} and \eqref{Z:4} thus
\begin{equation}\label{lambdaeq}
	\frac{\partial {\bf Z}}{\partial\lambda}=\left\{-\frac{t}{2}\begin{bmatrix}0 & 0\\ 1 & 0\end{bmatrix}+\frac{{\bf A}}{\lambda-1}+\frac{{\bf B}}{\lambda}\right\}{\bf Z},
\end{equation}
where we parametrize the coefficient matrices ${\bf A}$ and ${\bf B}$ as
\begin{equation}\label{Adef}
	{\bf A}=\begin{bmatrix}uv& -v^2\\ u^2 & -uv\end{bmatrix},\ \ \ \ \ \ \ \ \ \ \ {\bf B}=\begin{bmatrix}\gamma_1 & \gamma_3\\ \gamma_2 & -\gamma_1\end{bmatrix}:\ \ \gamma_1^2+\gamma_2\gamma_3=\frac{\alpha^2}{4},
\end{equation}
with $u=u(t,\alpha,\gamma),v=v(t,\alpha,\gamma),\gamma_j=\gamma_j(t,\alpha,\gamma)\in\mathbb{C}$. Note that we have with \eqref{Z:1},
\begin{equation}\label{sumid}
	{\bf A}+{\bf B}=\frac{1}{2}\begin{bmatrix}-tY_1^{12} & 1\\ 2tY_1^{11}+\alpha^2 & tY_1^{12}\end{bmatrix},
\end{equation}
and
\begin{equation}\label{A2nd}
	{\bf A}=\begin{bmatrix}-Y_1^{11}+\frac{1}{2}(Y_1^{12}\alpha^2-Y_1^{21})-\frac{t}{2}Y_2^{12}+\frac{t}{2}Y_1^{11}Y_1^{12} & -Y_1^{12}+Y_1^{11}+\frac{t}{2}(Y_1^{12})^2\smallskip\\
	-Y_1^{21}-Y_1^{11}\alpha^2-\frac{t}{2}(Y_2^{22}-Y_2^{11})+\frac{t}{2}(1-(Y_1^{11})^2) & Y_1^{11}-\frac{1}{2}(Y_1^{12}\alpha^2-Y_1^{21})+\frac{t}{2}Y_2^{12}-\frac{t}{2}Y_1^{11}Y_1^{12}\end{bmatrix}.
\end{equation}
This allows us to compute all coefficients in \eqref{lambdaeq} through \eqref{BessAsy}. Next, using again \eqref{BessT:1} and \eqref{BessAsy} we also find
\begin{equation*}
	\frac{\partial{\bf Z}}{\partial t}{\bf Z}^{-1}=-\frac{\lambda}{2}\begin{bmatrix}0 & 0\\ 1 & 0\end{bmatrix}+\frac{1}{2t}\begin{bmatrix}-tY_1^{12} & 1\\ 2tY_1^{11}+\alpha^2 & tY_1^{12}\end{bmatrix}+\frac{1}{\lambda}\left(\frac{1}{t}({\bf A}+{\bf Y}_1)+({\bf Y}_1)_t\right)+\mathcal{O}\left(\lambda^{-2}\right),\ \ \ \lambda\rightarrow\infty
\end{equation*}
and together with $\frac{\partial{\bf Z}}{\partial t}{\bf Z}^{-1}=\mathcal{O}(1)$ as $\lambda\rightarrow 0$ and $\lambda\rightarrow 1$ we have in addition to \eqref{lambdaeq} also
\begin{equation}\label{teq}
	\frac{\partial {\bf Z}}{\partial t}=\left\{-\frac{\lambda}{2}\begin{bmatrix}0 & 0\\1 & 0\end{bmatrix}+\frac{1}{t}({\bf A}+{\bf B})\right\}{\bf Z}\ \ \ \ \textnormal{and}\ \ \ \frac{1}{t}({\bf A}+{\bf Y}_1)+({\bf Y}_1)_t\equiv 0.
\end{equation}
Frobenius integrability of the overdetermined system \eqref{lambdaeq}, \eqref{teq} leads to the zero curvature condition which in turn is equivalent to
\begin{equation}\label{zcc}
	\frac{\partial{\bf A}}{\partial t}=\frac{1}{2}\Big[{\bf A},\bigl[\begin{smallmatrix} 0 & 0\\ 1 & 0\end{smallmatrix}\bigr]\Big]-\frac{1}{t}[{\bf A},{\bf B}],\ \ \ \ \ \ \ \ \frac{\partial{\bf B}}{\partial t}=\frac{1}{t}[{\bf A},{\bf B}].
\end{equation}
Translating \eqref{zcc} into the corresponding matrix entries we find a coupled nonlinear system for $(u,v,\gamma_1,\gamma_2,\gamma_3)$,
\begin{equation}\label{set:1}
	(\gamma_1)_t=-\frac{1}{t}(v^2\gamma_2+u^2\gamma_3),\ \ \ (\gamma_2)_t=\frac{2}{t}(u^2\gamma_1-uv\gamma_2),\ \ \ (\gamma_3)_t=\frac{2}{t}(v^2\gamma_1+uv\gamma_3),
\end{equation}
\begin{equation}\label{set:2}
	(uv)_t=-\frac{v^2}{2}+\frac{1}{t}(v^2\gamma_2+u^2\gamma_3),\ \ \ (v^2)_t=\frac{2}{t}(v^2\gamma_1+uv\gamma_3),\ \ \ (u^2)_t=-uv-\frac{2}{t}(u^2\gamma_1-uv\gamma_2).
\end{equation}
Recall that from \eqref{Adef} and \eqref{sumid}
\begin{equation}\label{cutie}
	uv+\gamma_1=-\frac{t}{2}Y_1^{12},\ \ \ \ \ u^2+\gamma_2=tY_1^{11}+\frac{\alpha^2}{2},\ \ \ \ -v^2+\gamma_3=\frac{1}{2}.
\end{equation}
\begin{rem}\label{cool} In view of \eqref{Z:2} and \eqref{Adef} as well as \eqref{sumid} we can parametrize the diagonalizing matrix $\widehat{{\bf Z}}(0)$ as
\begin{equation}\label{A:1}
	\widehat{{\bf Z}}(0)=\alpha^{-\frac{1}{2}}\begin{bmatrix}1 & -\left(\frac{1}{2}+v^2\right)\smallskip\\ (\frac{t}{2}Y_1^{12}+uv+\frac{\alpha}{2})(\frac{1}{2}+v^2)^{-1} & \frac{\alpha}{2}-\frac{t}{2}Y_1^{12}-uv\end{bmatrix}\left(\frac{1}{2}+v^2\right)^{\frac{1}{2}\sigma_3}w^{\sigma_3},\ \ \ \ \alpha\neq 0,
\end{equation}
with some $w=w(t,\alpha,\gamma)\in\mathbb{C}\backslash\{0\}$. But \eqref{teq} tells us that
\begin{equation}\label{A:2}
	\frac{\partial\widehat{{\bf Z}}}{\partial t}(0)=\frac{1}{t}({\bf A}+{\bf B})\widehat{{\bf Z}}(0),
\end{equation}
hence substituting \eqref{A:1} into \eqref{A:2} we find with the help of \eqref{set:1} and \eqref{set:2},
\begin{equation*}
	\frac{\partial}{\partial t}\ln w=\frac{\alpha}{4t}\frac{1}{\frac{1}{2}+v^2},\ \ \ \ \ \ \ \ \ \ \ w^2=-\widehat{Z}^{11}(0)\big(\widehat{Z}^{12}(0)\big)^{-1}=-\widehat{X}^{11}(0)\big(\widehat{X}^{12}(0)\big)^{-1}.
\end{equation*}
In addition we also have
\begin{equation*}
	\frac{\partial}{\partial t}\ln\big(w\,t^{-\frac{1}{2}\alpha}\big)=-\frac{\alpha}{2t}\frac{v^2}{\frac{1}{2}+v^2}.
\end{equation*}
\end{rem}
The above identities allow us to replace $\gamma_j$ in \eqref{set:1}: first, using the formula for $\gamma_1$ and $\gamma_3$ in the third equation of \eqref{set:1} we find
\begin{equation*}
	v\left(tv_t+\frac{t}{2}vY_1^{12}-\frac{u}{2}\right)\equiv 0\ \ \ \forall\, t.
\end{equation*}
But if $v$ were to vanish identically, then $\gamma_3\equiv\frac{1}{2}$ and $u\equiv 0$, see \eqref{set:2}. Hence all $\gamma_j$ are $t$-independent and we find from \eqref{Adef} that $Y_1^{11}+\frac{t}{2}(Y_1^{12})^2\equiv 0$. But now ${\bf A}\equiv {\bf 0}$ so that with \eqref{A2nd}, $Y_1^{12}\equiv 0$, which contradicts Proposition \ref{Di:1}. In short, we have the differential equation
\begin{equation}\label{S:1}
	tv_t=-\frac{t}{2}vY_1^{12}+\frac{u}{2}.
\end{equation}
Second, using the formul\ae\,for $\gamma_1,\gamma_2$ and $\gamma_3$ in the first equation of \eqref{set:1} we find
\begin{equation*}
	\left(uv+\frac{t}{2}Y_1^{12}\right)_t=\frac{1}{t}\left(v^2tY_1^{11}+v^2\frac{\alpha^2}{2}+\frac{u^2}{2}\right).
\end{equation*}
But $(tY_1^{12})_t=v^2$ (adding the first two equations in \eqref{set:1}, \eqref{set:2} and using \eqref{sumid}) so together with \eqref{S:1} we find from the last equation that
\begin{equation}\label{S:2}
	tu_t=\frac{1}{2}(\alpha^2-t)v+vtY_1^{11}+\frac{1}{2}utY_1^{12}.
\end{equation}
Incidentally, using the formul\ae\,for $\gamma_1$ and $\gamma_2$ in the second equation of \eqref{set:1} one also obtains \eqref{S:2} upon recalling that $(tY_1^{11})_t=-uv$ (adding the last two equations in \eqref{set:1}, \eqref{set:2} and using \eqref{sumid}). We summarize
\begin{prop}[see \cite{F}, Proposition $9.5.2$] The functions $(u,v,Y_1^{11},Y_1^{12})$ defined in \eqref{BessAsy} and \eqref{Adef}, \eqref{A2nd} satisfy the nonlinear dynamical system
\begin{equation}\label{S:3}
	tv_t=-\frac{t}{2}vY_1^{12}+\frac{u}{2},\ \ \ tu_t=\frac{1}{2}(\alpha^2-t)v+vtY_1^{11}+\frac{1}{2}utY_1^{12},\ \ \ (tY_1^{12})_t=v^2,\ \ \ \ \ (tY_1^{11})_t=-uv.
\end{equation}
Moreover, the combination $\sigma=\sigma(t)=\frac{1}{2}tY_1^{12}$ solves the Jimbo-Miwa-Okamoto form of Painlev\'e III,
\begin{equation}\label{S:4}
	(t\sigma_{tt})^2-\alpha^2(\sigma_t)^2-\sigma_t(1+4\sigma_t)(\sigma-t\sigma_t)=0.
\end{equation}
\end{prop}
\begin{proof} The derivation of the differential equation \eqref{S:4} for $\sigma$ is nearly identical to \cite{F}, Chapter $9.5$, see also \cite{TW2}. Apply $t\frac{\d}{\d t}$ to both sides of the first equation in \eqref{S:3} and use the second, third and fourth in the resulting right hand side,
\begin{equation}\label{S:5}
	t(tv_t)_t=\frac{1}{4}(\alpha^2-t)v+\frac{v}{2}\left(tY_1^{11}+\frac{1}{2}(tY_1^{12})^2\right)-\frac{t}{2}v^3.
\end{equation}
But from \eqref{teq} we have $(t{\bf Y}_1)_t=-{\bf A}$ which evaluated at its $12$ entry, see \eqref{A2nd}, gives
\begin{equation*}
	v^2=Y_1^{12}-\left(Y_1^{11}+\frac{1}{2}t(Y_1^{12})^2\right).
\end{equation*}
Multiplying the last identity by $t$ we then rewrite \eqref{S:5} as 
\begin{equation}\label{S:6}
	t(tv_t)_t=\frac{1}{4}(\alpha^2-t)v+\frac{v}{2}tY_1^{12}-tv^3.
\end{equation}
Next, we multiply \eqref{S:6} by $2v_t$ and use the third equation in \eqref{S:3},
\begin{equation*}
	\frac{\d}{\d t}(tv_t)^2=\frac{\d}{\d t}\left(\frac{1}{4}(\alpha^2-t+2tY_1^{12})v^2-\frac{t}{2}v^4+\frac{t}{4}Y_1^{12}\right),
\end{equation*}
so that
\begin{equation}\label{S:7}
	-\frac{1}{2}tY_1^{12}\left(v^2+\frac{1}{2}\right)=\frac{1}{4}(\alpha^2-t)v^2-\frac{t}{2}v^4-(tv_t)^2+C
\end{equation}
involving a $t$-independent expression $C$.  Combining \eqref{S:7} with the third equation in \eqref{S:3} we finally obtain
\begin{equation}\label{Ceqstill}
	(t\sigma_{tt})^2-\alpha^2(\sigma_t)^2-\sigma_t(1+4\sigma_t)(\sigma-t\sigma_t)=2C\sigma_t.
\end{equation}
We will now argue that in fact $C=0$: From the Fredholm series of the Bessel kernel determinant we compute directly the following boundary behavior (which is differentiable with respect to $t$),
\begin{equation*}
	\det(1-\gamma K_{\textnormal{Bess}}^{\alpha}\upharpoonright_{L^2(0,t)})=1-\frac{\gamma\,t^{\alpha+1}}{2^{2\alpha+2}\Gamma^2(2+\alpha)}+\mathcal{O}(t^{\alpha+2}),\ \ \ t\downarrow 0.
\end{equation*}
Hence, by Proposition \ref{Di:1} and the definition of $\sigma(t)$,
\begin{equation*}
	\sigma(t)=-\frac{\gamma\,t^{\alpha+1}}{2^{2\alpha+2}\Gamma(1+\alpha)\Gamma(2+\alpha)}+\mathcal{O}(t^{\alpha+2}),\ \ \ \ t\downarrow 0.
\end{equation*}
But once we substitute this expansion into \eqref{Ceqstill} a balance of exponents can only be achieved for $C=0$.
\end{proof}
Observe that upon multiplication of \eqref{S:6} with $v^2+\frac{1}{2}$ and the use of \eqref{S:7} for the term $tY_1^{12}(v^2+\frac{1}{2})$ we arrive at
\begin{equation*}
	t\left(v^2+\frac{1}{2}\right)(tv_t)_t=v(tv_t)^2+\frac{1}{4}(\alpha^2-t)\frac{v}{2}-\frac{t}{2}v^3(v^2+1).
\end{equation*} 
Provided we let $q=q(t)=\pm\im\sqrt{2}\,v(t)$ (with either choice of the sign), then the last differential equation is equivalent to
\begin{equation}\label{S:8}
	t(q^2-1)(tq_t)_t=q(tq_t)^2+\frac{1}{4}(t-\alpha^2)q+\frac{1}{4}tq^3(q^2-2)
\end{equation}
which is exactly (1.16) in \cite{TW2}. In addition, if we put $\mathcal{H}_H=\frac{1}{2}Y_1^{12}$ then \eqref{S:7} shows that \eqref{S:8} can be reformulated as a Hamiltonian dynamical system
\begin{equation*}
	\frac{\d p}{\d t}=-\frac{\partial\mathcal{H}_H}{\partial q},\ \ \  \\  \frac{\d q}{\d t}=\frac{\partial\mathcal{H}_H}{\partial p},\ \ \ \ \ \ \mathcal{H}_H=\mathcal{H}_H(q,p,t)=\frac{q^2-1}{4t}p^2-\frac{(\alpha^2-t)q^2+tq^4}{4t(q^2-1)}
\end{equation*}
where (see \eqref{Adef}, \eqref{A2nd})
\begin{equation*}
	q^2=t(Y_1^{12})^2+2(Y_1^{11}-Y_1^{12}),\ \ \ \ \ \ p^2=\frac{\alpha^2q^2}{(q^2-1)^2}+\frac{2t}{q^2-1}\left(Y_1^{12}+\frac{q^2}{2}\right),
\end{equation*}
or equivalently (see \eqref{cutie} and \eqref{Z:2}, \eqref{Adef})
\begin{equation*}
	q^2=1-2\gamma_3,\ \ \ \ \ \ \ \ \gamma_3=-\alpha\widehat{X}^{11}(0)\widehat{X}^{12}(0),\ \alpha\neq 0;\ \ \ \ \ \gamma_3=-\frac{1-\gamma}{2\pi\im}\big(\widehat{X}^{11}(0)\big)^2,\ \alpha=0.
\end{equation*}
Moreover, recalling Remark \ref{cool}, we also have that
\begin{equation*}
	\frac{\partial}{\partial t}\ln w=-\frac{\alpha}{2t}\frac{1}{q^2-1},\ \ \ \ \frac{\partial}{\partial t}\ln\big(wt^{-\frac{1}{2}\alpha}\big)=-\frac{\alpha}{2t}\frac{q^2}{q^2-1},\ \ \ w^2=-\widehat{X}^{11}(0)\big(\widehat{X}^{12}(0)\big)^{-1},
\end{equation*}
and thus
\begin{eqnarray}
	L(t,\alpha;\gamma)&=&-\frac{\alpha}{2}\ln\Big(-\widehat{X}^{11}(0)\big(\widehat{X}^{12}(0)\big)^{-1}\Big)\bigg|_{s=0}^t,\ \ \ \ \ -1<\alpha<0,\label{cool:1}\\
	L(t,\alpha;\gamma)&=&-\frac{\alpha}{2}\ln\Big(-\widehat{X}^{11}(0)\big(\widehat{X}^{12}(0)\big)^{-1} s^{-\alpha}\Big)\bigg|_{s=0}^t,\ \ \ \ \ \ \alpha>0.\label{cool:2}
\end{eqnarray}
The function $\widehat{{\bf X}}(\lambda)$ was introduced in RHP \ref{Bess:X}, condition (3).
\section{Bessel and confluent hypergeometric parametrices}\label{AppB}
Following \cite{B2}, section $6.1$ we define for $\z\in\mathbb{C}\backslash[0,\infty)$ the unimodular matrix valued function
\begin{equation}\label{BessUndress}
	\Psi_{\alpha}(\z)=\sqrt{\pi}\e^{-\im\frac{\pi}{4}}\begin{bmatrix}I_{\alpha}((-\z)^{\frac{1}{2}}) & -\frac{\im}{\pi}K_{\alpha}((-\z)^{\frac{1}{2}})\smallskip\\
	(-\z)^{\frac{1}{2}}I_{\alpha}'((-\z)^{\frac{1}{2}}) & -\frac{\im}{\pi}(-\z)^{\frac{1}{2}}K_{\alpha}'((-\z)^{\frac{1}{2}})\end{bmatrix}\e^{\im\frac{\pi}{2}\alpha\sigma_3}
\end{equation}
using the modified Bessel functions $I_{\nu}(z)$ and $K_{\nu}(z)$ both defined with their principal branches, cf. \cite{NIST} and $z^{\alpha}:\mathbb{C}\backslash(-\infty,0]\rightarrow\mathbb{C}$ also in terms of the principal branch. Note that
\begin{equation*}
	\big(\Psi_{\alpha}(\z)\big)_+=\big(\Psi_{\alpha}(\z)\big)_-\begin{bmatrix}\e^{-\im\pi\alpha} & \e^{-\im\pi\alpha} \\ 0 & \e^{\im\pi\alpha}\end{bmatrix},\ \ \z>0;\ \ \ \ \ \ \ \frac{\d\Psi_{\alpha}}{\d\z}=\begin{bmatrix}0 & \frac{1}{2\z}\\ \frac{\alpha^2}{2\z}-\frac{1}{2} & 0\end{bmatrix}\Psi_{\alpha}.
\end{equation*}
and if we assemble the function (compare (6.3) in \cite{B2})
\begin{equation}\label{Bessundress}
	\Psi(\z;\alpha)=\Psi_{\alpha}(\z)\e^{-\im\frac{\pi}{2}\alpha\sigma_3}\begin{cases}\bigl[\begin{smallmatrix}1 & 0\\ -\e^{-\im\pi\alpha} & 1\end{smallmatrix}\bigr],&\textnormal{arg}\,\z\in(0,\frac{\pi}{3})\smallskip\\ \mathbb{I},&\textnormal{arg}\,\z\in(\frac{\pi}{3},\frac{5\pi}{3})\smallskip\\ \bigl[\begin{smallmatrix} 1& 0\\ \e^{\im\pi\alpha} & 1\end{smallmatrix}\bigr],&\textnormal{arg}\,\z\in(\frac{5\pi}{3},2\pi)\end{cases}
\end{equation}
then
\begin{problem}[see \cite{B2}, RHP 6.1.]\label{Bessbetter} The model function $\Psi(\z;\alpha)$ defined in \eqref{Bessundress} has the following properties
\begin{enumerate}
	\item $\Psi(\z;\alpha)$ is analytic for $\z\in\mathbb{C}\backslash(\Sigma_{\Psi}\cup\{0\})$ with $\Sigma_{\Psi}=\bigcup_{j=1}^3\Gamma_j$ where
	\begin{equation*}
		\Gamma_1=\e^{\im\frac{\pi}{3}}(0,\infty),\ \ \ \ \ \ \Gamma_2=(0,\infty),\ \ \ \ \ \ \ \Gamma_3=\e^{\im\frac{5\pi}{3}}(0,\infty)
	\end{equation*}
	are all oriented from zero to infinity.
	\item Along $\Sigma_{\Psi}$ we observe the jumps
	\begin{equation*}
		\Psi_+(\z;\alpha)=\Psi_-(\z;\alpha)\begin{bmatrix}1 & 0\\ \e^{-\im\pi\alpha} & 1\end{bmatrix},\ \ \z\in\Gamma_1;\ \ \ \ \ \ \ \Psi_+(\z;\alpha)=\Psi_-(\z;\alpha)\begin{bmatrix}1 & 0\\ \e^{\im\pi\alpha} & 1\end{bmatrix},\ \ \z\in\Gamma_3
	\end{equation*}
	and
	\begin{equation*}
		\Psi_+(\z;\alpha)=\Psi_-(\z;\alpha)\begin{bmatrix}0&1\\ -1 & 0\end{bmatrix},\ \ \ \z\in\Gamma_2
	\end{equation*}
	\item In a vicinity of $\z=0$, first in case $\alpha\notin\mathbb{Z}$,
	\begin{equation*}
		\Psi(\z;\alpha)=\widehat{\Psi}(\z;\alpha)(-\z)^{\frac{\alpha}{2}\sigma_3}\begin{bmatrix}1 & \frac{\im}{2}\frac{1}{\sin\pi\alpha}\\ 0 & 1\end{bmatrix}\begin{cases}\bigl[\begin{smallmatrix}1 & 0\\ -\e^{-\im\pi\alpha} & 1\end{smallmatrix}\bigr],&\textnormal{arg}\,\z\in(0,\frac{\pi}{3})\\ \mathbb{I},&\textnormal{arg}\,\z\in(\frac{\pi}{3},\frac{5\pi}{3})\\ \bigl[\begin{smallmatrix}1&0\\ \e^{\im\pi\alpha} & 1\end{smallmatrix}\bigr],&\textnormal{arg}\,\z\in(\frac{5\pi}{3},2\pi)\end{cases},
	\end{equation*}
	and second for $\alpha\in\mathbb{Z}$,
	\begin{equation*}
		\Psi(\z;\alpha)=\widehat{\Psi}(\z;\alpha)(-\z)^{\frac{\alpha}{2}\sigma_3}\begin{bmatrix}1 & -\frac{\e^{\im\pi\alpha}}{2\pi\im}\ln(-\z)\\ 0 & 1\end{bmatrix}
		\begin{cases}\bigl[\begin{smallmatrix}1 & 0\\ -\e^{-\im\pi\alpha} & 1\end{smallmatrix}\bigr],&\textnormal{arg}\,\z\in(0,\frac{\pi}{3})\\ \mathbb{I},&\textnormal{arg}\,\z\in(\frac{\pi}{3},\frac{5\pi}{3})\\ \bigl[\begin{smallmatrix}1&0\\ \e^{\im\pi\alpha} & 1\end{smallmatrix}\bigr],&\textnormal{arg}\,\z\in(\frac{5\pi}{3},2\pi)\end{cases}.
	\end{equation*}
	In both cases $\widehat{\Psi}(\z;\alpha)$ is analytic at $\z=0$ and principal branches are chosen throughout.
	\begin{rem}\label{locrem} With the help of \cite{NIST}, Chapter $10.31$ we find
	\begin{equation*}
		\widehat{\Psi}(\z;\alpha)=\sqrt{\pi}\,\e^{-\im\frac{\pi}{4}}\begin{bmatrix}\widehat{\Psi}^{11}(\z;\alpha) & -\frac{\im}{2}\frac{1}{\sin\pi\alpha}\widehat{\Psi}^{11}(\z;-\alpha)\smallskip\\ \widehat{\Psi}^{21}(\z;\alpha) & -\frac{\im}{2}\frac{1}{\sin\pi\alpha}\widehat{\Psi}^{21}(\z;-\alpha)\end{bmatrix},\ \ \alpha\notin\mathbb{Z}
	\end{equation*}
	with
	\begin{equation*}
		\widehat{\Psi}^{11}(\z;\alpha)=\frac{1}{2^{\alpha}}\sum_{k=0}^{\infty}\frac{(-\frac{\z}{4})^k}{k!\,\Gamma(1+\alpha+k)},\ \ \ \widehat{\Psi}^{21}(\z;\alpha)=\frac{\alpha}{2^{\alpha}}\sum_{k=0}^{\infty}\frac{(1+\frac{2k}{\alpha})(-\frac{\z}{4})^k}{k!\,\Gamma(1+\alpha+k)},\ \ \ \ \z\in\mathbb{D}_r(0);
	\end{equation*}
	and 
	\begin{equation*}
		\widehat{\Psi}(\z;\alpha)=\sqrt{\pi}\,\e^{-\im\frac{\pi}{4}}\left\{\begin{bmatrix}\frac{2^{-\alpha}}{\Gamma(1+\alpha)} & \frac{2^{\alpha}\Gamma(\alpha)}{2\pi\im}\smallskip\\ \frac{\alpha 2^{-\alpha}}{\Gamma(1+\alpha)} & -\frac{\alpha 2^{\alpha}\Gamma(\alpha)}{2\pi\im}\end{bmatrix}+\mathcal{O}\big(\z\big)\right\},\ \ \ \alpha\in\mathbb{Z}_{\geq 1},\ \ \ \ \z\in\mathbb{D}_r(0);
	\end{equation*}
	as well as
	\begin{equation*}
		\widehat{\Psi}(\z;0)=\sqrt{\pi}\,\e^{-\im\frac{\pi}{4}}\left\{\begin{bmatrix} 1 & \frac{\im}{\pi}(\gamma_E-\ln 2)\\ 0 & \frac{\im}{\pi}\end{bmatrix}+\mathcal{O}\big(\z\big)\right\},\ \ \ \ \z\in\mathbb{D}_r(0).
	\end{equation*}
	\end{rem}
	\item As $\z\rightarrow\infty,\z\notin\Sigma_{\Psi}$ we have
	\begin{align*}
		\Psi(\z;\alpha)\sim\begin{bmatrix}1 & 0\\ -b_1(\alpha) & 1\end{bmatrix}&\,\left\{\mathbb{I}+\sum_{m=1}^{\infty}\begin{bmatrix}a_{2m}(\alpha) & -a_{2m-1}(\alpha)\\ b_1(\alpha)a_{2m}(\alpha)-b_{2m+1}(\alpha) & b_{2m}(\alpha)-b_1(\alpha)a_{2m-1}(\alpha)\end{bmatrix}(-\z)^{-m}\right\}\\
		\times\,&(-\z)^{-\frac{1}{4}\sigma_3}\frac{1}{\sqrt{2}}\begin{bmatrix}1 & -1\\ 1 & 1\end{bmatrix}\e^{-\im\frac{\pi}{4}\sigma_3}\e^{(-\z)^{\frac{1}{2}}\sigma_3}
	\end{align*}
	with the coefficients (cf. \cite{NIST})
	\begin{equation*}
		a_k(\nu)=\frac{1}{k!\,8^k}(4\nu^2-1^2)(4\nu^2-3^2)\cdot\ldots\cdot(4\nu^2-(2k-1)^2),\ k\in\mathbb{Z}_{\geq 1};\ \ \ \ b_1(\nu)=\frac{1}{8}(4\nu^2+3)
	\end{equation*}
	and
	\begin{equation*}
		b_k(\nu)=\frac{1}{k!\,8^k}\big((4\nu^2-1^2)(4\nu^2-3^2)\cdot\ldots\cdot(4\nu^2-(2k-3)^2)\big)(4\nu^2+4k^2-1),\ \ k\in\mathbb{Z}_{\geq 2}.
	\end{equation*}
\end{enumerate}
\end{problem}
In addition to the Bessel-parametrix \eqref{Bessundress} we require also the following model function built out of confluent hypergeometric functions $U(a,\z)\equiv U(a,1,\z)$. The underlying construction is essentially a rotation of the one given in \cite{BB}, equations (2.19) and (2.21), see also \cite{BI,IK} for similar constructions. In more detail, define
\begin{equation}\label{OP:5}
	\Phi_0(\z)=\begin{bmatrix}U(\nu,\e^{-\im\frac{\pi}{2}}\z)\e^{2\pi\im\nu} & -U(1-\nu,\e^{-\im\frac{3\pi}{2}}\z)\e^{\im\pi\nu}\frac{\Gamma(1-\nu)}{\Gamma(\nu)}\\ -U(1+\nu,\e^{-\im\frac{\pi}{2}}\z)\e^{\im\pi\nu}\frac{\Gamma(1+\nu)}{\Gamma(-\nu)} & U(-\nu,\e^{-\im\frac{3\pi}{2}}\z)\end{bmatrix}\e^{\frac{\im}{2}\z\sigma_3},\ \ \ \textnormal{arg}\,\z\in\left(\frac{\pi}{2},\frac{5\pi}{2}\right)
\end{equation}
with $\nu=\frac{v}{2\pi\im}\in\im\mathbb{R}$. Now assemble
\begin{equation}\label{OP:6}
	\Phi(\z)=\Phi_0(\z)\begin{cases}\mathbb{I},&\textnormal{arg}\,\z\in(\pi,\frac{4\pi}{3})\smallskip\\ \bigl[\begin{smallmatrix}1&0\\ \e^{\im\pi\nu} & 1\end{smallmatrix}\bigr],&\textnormal{arg}\,\z\in(\frac{4\pi}{3},\frac{3\pi}{2})\smallskip\\ \bigl[\begin{smallmatrix}1&0\\ 2\im\sin\pi\nu & 1\end{smallmatrix}\bigr]\bigl[\begin{smallmatrix}1 & 0\\ \e^{-\im\pi\nu} & 1\end{smallmatrix}\bigr],&\textnormal{arg}\,\z\in(\frac{3\pi}{2},\frac{5\pi}{3})\smallskip\\ \bigl[\begin{smallmatrix} 1 & 0\\ 2\im\sin\pi\nu & 1\end{smallmatrix}\bigr],&\textnormal{arg}\,\z\in(\frac{5\pi}{3},2\pi)\end{cases}\begin{cases}\bigl[\begin{smallmatrix}1 & 0\\ 2\im\sin\pi\nu & 1\end{smallmatrix}\bigr]\bigl[\begin{smallmatrix}0&-\e^{\im\pi\nu}\\ \e^{-\im\pi\nu} & 0\end{smallmatrix}\bigr],&\textnormal{arg}\,\z\in(2\pi,\frac{7\pi}{3})\smallskip\\ \bigl[\begin{smallmatrix}1 & 0\\2\im\sin\pi\nu & 1\end{smallmatrix}\bigr]\bigl[\begin{smallmatrix} 1 & -\e^{\im\pi\nu}\\ \e^{-\im\pi\nu} & 0\end{smallmatrix}\bigr],&\textnormal{arg}\,\z\in(\frac{7\pi}{3},\frac{5\pi}{2})\smallskip\\ \bigl[\begin{smallmatrix}1 & -\e^{-\im\pi\nu}\\ \e^{\im\pi\nu} & 0\end{smallmatrix}\bigr],&\textnormal{arg}\,\z\in(\frac{\pi}{2},\frac{2\pi}{3})\smallskip\\ \bigl[\begin{smallmatrix}0&-\e^{-\im\pi\nu}\\ \e^{\im\pi\nu} & 0\end{smallmatrix}\bigr],&\textnormal{arg}\,\z\in(\frac{2\pi}{3},\pi)\end{cases}
\end{equation}
and by recalling analytic and asymptotic properties of $U(a,\z)$, cf. \cite{NIST} we find
\begin{problem}\label{confluprob} The function $\Phi(\z)\in\mathbb{C}^{2\times 2}$ introduced in \eqref{OP:5} and \eqref{OP:6} has the following properties.
\begin{enumerate}
	\item $\Phi(\z)$ is analytic for $\z\in\mathbb{C}\backslash(\{\textnormal{arg}\,\z=\frac{2\pi}{3},\pi,\frac{4\pi}{3},\frac{5\pi}{3},2\pi,\frac{7\pi}{3}\}\cup\{0\})$ and we orient the six rays emanating from $\z=0$ as shown near $\lambda=1$ in Figure \ref{figure13}.
	\item The limiting values $\Phi_{\pm}(\z)$ on the jump contours satisfy
	\begin{equation*}
		\Phi_+(\z)=\Phi_-(\z)\e^{-\im\frac{\pi}{2}\nu\sigma_3}\bigl[\begin{smallmatrix}1 & 0\\ -1&1\end{smallmatrix}\bigr]\e^{\im\frac{\pi}{2}\nu\sigma_3},\ \ \textnormal{arg}\,\z=\frac{2\pi}{3};\ \ \ \ \Phi_+(\z)=\Phi_-(\z)\e^{-\im\frac{\pi}{2}\nu\sigma_3}\bigl[\begin{smallmatrix}0 & -1\\ 1 & 0\end{smallmatrix}\bigr]\e^{\im\frac{\pi}{2}\nu\sigma_3},\ \ \textnormal{arg}\,\z=\pi;
	\end{equation*}
	\begin{equation*}
		\Phi_+(\z)=\Phi_-(\z)\e^{-\im\frac{\pi}{2}\nu\sigma_3}\bigl[\begin{smallmatrix}1 & 0\\ -1&1\end{smallmatrix}\bigr]\e^{\im\frac{\pi}{2}\nu\sigma_3},\ \ \textnormal{arg}\,\z=\frac{4\pi}{3};\ \ \ \ \Phi_+(\z)=\Phi_-(\z)\e^{\im\frac{\pi}{2}\nu\sigma_3}\bigl[\begin{smallmatrix}1 & 0\\ -1 & 1\end{smallmatrix}\bigr]\e^{-\im\frac{\pi}{2}\nu\sigma_3},\ \ \textnormal{arg}\,\z=\frac{5\pi}{3};
	\end{equation*}
	\begin{equation*}
		\Phi_+(\z)=\Phi_-(\z)\e^{\im\frac{\pi}{2}\nu\sigma_3}\bigl[\begin{smallmatrix}0 & -1\\ 1&0\end{smallmatrix}\bigr]\e^{-\im\frac{\pi}{2}\nu\sigma_3},\ \ \textnormal{arg}\,\z=2\pi;\ \ \ \ \Phi_+(\z)=\Phi_-(\z)\e^{\im\frac{\pi}{2}\nu\sigma_3}\bigl[\begin{smallmatrix}1 & 0\\ -1 & 1\end{smallmatrix}\bigr]\e^{-\im\frac{\pi}{2}\nu\sigma_3},\ \ \textnormal{arg}\,\z=\frac{7\pi}{3}.
	\end{equation*}
	By construction, there are no jumps on the vertical axis $\textnormal{arg}\,\z=\frac{\pi}{2},\frac{3\pi}{2}$.
	\item In a vicinity of $\z=0$ we find
	\begin{equation*}
		\Psi(\z)=\widehat{\Psi}(\z)\begin{bmatrix}1 & \frac{\gamma}{2\pi\im}\ln\z\\ 0 & 1\end{bmatrix}\begin{cases}\bigl[\begin{smallmatrix}1 & 0\\ -\e^{2\pi\im\nu} & 1\end{smallmatrix}\bigr],&\textnormal{arg}\,\z\in(-\pi,-\frac{2\pi}{3})\smallskip\\ \mathbb{I},&\textnormal{arg}\,\z\in(-\frac{2\pi}{3},-\frac{\pi}{3})\smallskip\\ \bigl[\begin{smallmatrix}1 & 0\\ -1 & 1\end{smallmatrix}\bigr],&\textnormal{arg}\,\z\in(-\frac{\pi}{3},0)\smallskip\\ \bigl[\begin{smallmatrix}1&-1\\ 0 & 1\end{smallmatrix}\bigr]\bigl[\begin{smallmatrix}1 & 0\\ 1 & 1\end{smallmatrix}\bigr],&\textnormal{arg}\,\z\in(0,\frac{\pi}{3})\smallskip\\ \bigl[\begin{smallmatrix}1 & -1\\ 0 & 1\end{smallmatrix}\bigr],&\textnormal{arg}\,\z\in(\frac{\pi}{3},\frac{2\pi}{3})\smallskip\\ \bigl[\begin{smallmatrix}1 & -1\\ 0 & 1\end{smallmatrix}\bigr]\bigl[\begin{smallmatrix} 1 & 0\\ \e^{2\pi\im\nu} & 1\end{smallmatrix}\bigr],&\textnormal{arg}\,\z\in(\frac{2\pi}{3},\pi)\end{cases}\times\e^{-\im\frac{\pi}{2}\nu\sigma_3}
	\end{equation*}
	with $\widehat{\Psi}(\z)$ analytic at $\z=0$ and $\ln:\mathbb{C}\backslash(-\infty,0]\rightarrow\mathbb{C}$ defined with its principal branch.
	\item As $\z\rightarrow\infty$,
	\begin{align*}
		\Phi(\z)\sim\Bigg\{\mathbb{I}+\sum_{m=1}^{\infty}\e^{\im\frac{\pi}{2}\nu\sigma_3}&\,\begin{bmatrix}((\nu)_m)^2 & (-1)^m\,m((1-\nu)_{m-1})^2\frac{\Gamma(1-\nu)}{\Gamma(\nu)}\\ m\,((1+\nu)_{m-1})^2\frac{\Gamma(1+\nu)}{\Gamma(-\nu)} &(-1)^m((-\nu)_m)^2\end{bmatrix}\e^{-\im\frac{\pi}{2}\nu\sigma_3}\frac{(\im\z)^{-m}}{m!}\Bigg\}\\
		&\,\times\,\z^{-\nu\sigma_3}\e^{\frac{\im}{2}\z\sigma_3}\begin{cases}\Bigl[\begin{smallmatrix}0 & -\e^{\im\frac{3\pi}{2}\nu} \\\e^{-\im\frac{\pi}{2}\nu} & 0\end{smallmatrix}\Bigr],&\textnormal{arg}\,\z\in(\frac{\pi}{2},\pi)\smallskip\\ \Bigl[\begin{smallmatrix}\e^{\im\frac{5\pi}{2}\nu} & 0\\ 0&\e^{-\im\frac{3\pi}{2}\nu}\end{smallmatrix}\Bigr],&\textnormal{arg}\,\z\in(\pi,2\pi)\\ \Bigl[\begin{smallmatrix}0 & -\e^{\im\frac{7\pi}{2}\nu}\\ \e^{-\im\frac{5\pi}{2}\nu} & 0\end{smallmatrix}\Bigr],&\textnormal{arg}\,\z\in(2\pi,\frac{5\pi}{2})\end{cases}
	\end{align*}
	where $(a)_0=1,(a)_m=a(a+1)(a+2)\cdot\ldots\cdot(a+m-1),a\in\mathbb{C},m\in\mathbb{Z}_{\geq 0}$ is the Pochhammer symbol.
\end{enumerate}
\end{problem}
\section{Smoothness of $\mathcal{H}_H $}\label{AppC}
We first show the absence of singularities in $q$ and $p$ using probabilistic arguments. We know from the probabilistic interpretation of $F_H(t,\alpha)$ that $\mathcal{H}(q,p,t,\alpha)$ is smooth for $t\in(0,+\infty)$. But \eqref{HH:2} implies that (see also \eqref{S:8} below)
\begin{equation*}
	t(q^2-1)(tq_t)_t=q(tq_t)^2+\frac{1}{4}(t-\alpha^2)q+\frac{1}{4}tq^3(q^2-2).
\end{equation*}
Hence, near an assumed pole $t_0>0$ of $q(t,\alpha;\gamma)$, we have the Laurent expansion
\begin{equation*}
	q(t,\alpha;\gamma)=\frac{c}{t-t_0}+\mathcal{O}\big(t-t_0\big),\ \ \ c^2=4t_0,\ \ \ 0<|t-t_0|<r,
\end{equation*}
and thus back in \eqref{HH:2},
\begin{equation*}
	\mathcal{H}_H(q,p,t,\alpha)=\frac{1}{t-t_0}+\mathcal{O}(1),\ \ \ \ 0<|t-t_0|<r.
\end{equation*}
But this contradicts the regularity of $\mathcal{H}(q,p,t,\alpha)$, i.e. $q(t,\alpha;\gamma)$ is indeed smooth on $(0,+\infty)$ for $\alpha\in(-1,+\infty),\gamma\in[0,1]$. Alternatively we can start from Proposition \ref{Di:1} and \eqref{S:3},
\begin{equation}\label{Op}
	q^2(t,\alpha;\gamma)=-2\frac{\partial}{\partial t}\left(t\frac{\partial}{\partial t}\ln\det\left(1-\gamma K_{\textnormal{Bess}}^{\alpha}\upharpoonright_{L^2(0,t)}\right)\right),
\end{equation}
so that possible poles of $q(t,\alpha;\gamma)$ coincide with zeros of the Fredholm determinant (as function of $t$). Thus, if we can estimate the operator norm of $\gamma K_{\textnormal{Bess}}^{\alpha}\upharpoonright_{L^2(0,t)}$ for $t\in(0,+\infty)$ above by unity, regularity of $q$ follows. In order to achieve this we draw inspiration from \cite{TW2} and think of $K_{\textnormal{Bess}}^{\alpha}\upharpoonright_{L^2(0,t)}$ as acting on $L^2(0,\infty)$ with kernel
\begin{equation*}
	\chi_{(0,t)}(\lambda)K_{\textnormal{Bess}}^{\alpha}(\lambda,\mu)\chi_{(0,t)}(\mu).
\end{equation*}
Hence, introducing the Hankel transform
\begin{equation*}
	H:L^2(0,\infty)\rightarrow L^2(0,\infty);\ \ \ \ \ (Hf)(\lambda)=\frac{1}{2}\int_0^{\infty}J_{\alpha}(\sqrt{\lambda\mu})f(\mu)\,\d\mu,
\end{equation*}
which is unitary on $L^2(0,\infty)$, cf. \cite{Kos,TW2}, the operator $K_{\textnormal{Bess}}^{\alpha}\upharpoonright_{L^2(0,t)}$ is equal to the square of $P_tHP_t$, where
\begin{equation*}
	P_t:L^2(0,\infty)\rightarrow L^2(0,t);\ \ \ \ \ \ (P_tf)(\lambda)=f(\lambda)\chi_{(0,t)}(\lambda)
\end{equation*}
is the orthogonal projection from $L^2(0,\infty)$ to $L^2(0,t)$ with $t>0$. This means that, for any fixed $t\in(0,+\infty)$,
\begin{equation*}
	\|P_tHP_t\|_{L^2(0,\infty)}\leq \|P_t\|_{L^2(0,\infty)}^2\cdot\|H\|_{L^2(0,\infty)}=1.
\end{equation*}
Now suppose that $P_tHP_t$ were to have eigenvalues $\pm 1$ for some $t\in(0,\infty)$ with eigenfunctions $f_{\pm}\in L^2(0,\infty)$. Then $P_tf_{\pm}$ are also eigenfunctions to the same eigenvalues and
\begin{equation}\label{CS}
	\langle HP_tf_{\pm},P_tf_{\pm}\rangle_{L^2(0,\infty)}=\langle P_tHP_tf_{\pm},P_tf_{\pm}\rangle_{L^2(0,\infty)}=\pm\|P_tf_{\pm}\|_{L^2(0,\infty)}^2.
\end{equation}
But by Cauchy-Schwarz inequality
\begin{equation*}
	|\langle Hu,u\rangle_{L^2(0\infty)}|\leq\|Hu\|_{L^2(0,\infty)}\|u\|_{L^2(0,\infty)}=\|u\|^2_{L^2(0,\infty)},
\end{equation*}
with equality iff $u$ is an eigenfunction. Thus, returning to \eqref{CS}, $P_tf$ is a (nonzero) eigenfunction of the Hankel transform which vanishes for $\lambda>t$. But since $t\in(0,\infty)$ is fixed, $(HP_tf)(\lambda)$ is analytic for $\lambda\in\mathbb{C}\backslash(-\infty,0]$, i.e. we must have $P_tf=0$, which is a contradiction.
\end{appendix}

\end{document}